\RequirePackage{fix-cm}
\documentclass[smallextended]{svjour3}       
\smartqed  
\usepackage{graphicx}

\def\RP{{\mathbb{R}}_{\geq 0}}
\def\RPP{{\mathbb{R}}_{>0}}

\def\R{{\mathbb{R}}}
\def\ZP{{\mathbb{N}}}
\def\CG{{\sf CG}}
\def\N{N}

\def\poa{{\sf PoA}}
\def\LB{{\sf LB}}

\def\Z{\mathbb{Z}_{\geq 0}}
\newcommand{\sg}{{\bm\sigma}}

\newcommand{\EM}{{\sf EM}}
\newcommand{\NN}{\mathbb{N}}

\usepackage{natbib}
\usepackage{microtype}
\usepackage{multirow}
\usepackage{amsmath,amsfonts,amssymb,bbm,bm}
 \def\bibhang{24pt}%
 \bibpunct[, ]{[}{]}{,}{n}{}{,}%

\usepackage[colorlinks=true,breaklinks=true,bookmarks=true,urlcolor=blue,
     citecolor=blue,linkcolor=blue,bookmarksopen=false,draft=false]{hyperref}
     
\allowdisplaybreaks
 \journalname{Pre-print}
\begin{document}

\title{The Impact of Selfish Behavior\\in Load Balancing Games \thanks{
This work was partially supported by the Italian MIUR PRIN 2017 Project ALGADIMAR ``Algorithms, Games, and Digital Markets''. }\thanks{The current affiliation of Cosimo Vinci is the University of Salerno. A substantial part of this work was done when he was affiliated with the Gran Sasso Science Institute (that is his previous affiliation).}\thanks{A preliminary version of this paper appeared in Proceedings of the 25th Annual European Symposium on Algorithms (ESA 2017) \cite{BV17}.}}


\author{Vittorio Bilò         \and
        Cosimo Vinci 
}

\authorrunning{V. Bilò and C. Vinci} 

\institute{V. Bilò \at
              Department of Mathematics and Physics ``Ennio De Giorgi'', University of Salento, 73100 Lecce, Italy.
              \email{vittorio.bilo@unisalento.it}           
           \and
           C. Vinci\at
              Department of Computer Engineering, Electrical Engineering and Applied Mathematics,  University of Salerno, 84084 Fisciano, Italy.
              \at Gran Sasso Science Institute, 67100 L'Aquila, Italy.\\
              \email{cvinci@unisa.it}   
}

\date{Received: date / Accepted: date}

\maketitle

\begin{abstract}
To what extent does the structure of the players' strategy space influence the efficiency of decentralized solutions in congestion games? In this work, we investigate whether better performance are possible when restricting to load balancing games in which players can only choose among single resources. We consider three different solutions concepts, namely, approximate pure Nash equilibria, approximate one-round walks generated by selfish players aiming at minimizing their personal cost and approximate one-round walks generated by cooperative players aiming at minimizing the marginal increase in the sum of the players' personal costs. The last two concepts can be interpreted as solutions of greedy online algorithms for the related resource selection problem. We show that, under fairly general latency functions on the resources, better bounds cannot be achieved if players are either weighted or asymmetric. On the positive side, we prove that, under mild assumptions on the latency functions, improvements on the performance of approximate pure Nash equilibria are possible for load balancing games with weighted and symmetric players in the case of identical resources. We also design lower bounds on the performance of one-round walks in load balancing games with unweighted players and identical resources.
\end{abstract}

\section{Introduction}
\label{intro}
Congestion games \cite{R73} are non-cooperative games in which there is a set of selfish players competing for a set of resources, and each resource incurs a certain latency, expressed by a congestion-dependent function, to the players using it. Each player has a certain weight and an available set of strategies, where each strategy is a non-empty subset of resources, and aims at choosing a strategy minimizing her personal cost, which is defined as the sum of the latencies experienced on all the selected resources. We speak of {\em weighted} games/players when players have arbitrary non-negative weights and of {\em unweighted} games/players when all players have unitary weight.

Stable outcomes in this setting are pure Nash equilibria \cite{N50}: strategy profiles in which no player can lower her cost by unilaterally deviating to another strategy. However, they are demanding solution concepts, as they might not exist in weighted games \cite{FKS05} and, even when their existence is guaranteed, as, for instance, in unweighted games \cite{R73} or in weighted games with affine latency functions \cite{FKS05,HK12}, their computation might be an intractable problem \cite{ARV08,FPT04}. For such a reason, more relaxed solution concepts, such as $\epsilon$-approximate pure Nash equilibria or $\epsilon$-approximate one-round walks,  are also considered in the literature. An $\epsilon$-approximate pure Nash equilibrium is the relaxation of the concept of pure Nash equilibrium in which no player can lower her cost of a factor more than $1+\epsilon$ by unilaterally deviating to another strategy. Efficient algorithms computing approximate pure Nash equilibria in weighted congestion games with polynomial latency functions are known, see \cite{CFGS11,Caragiannis15,FGKS17,Gian18}.
An $\epsilon$-approximate one-round walk is defined as a myopic process in which players arrive in an arbitrary order and, upon arrival, each of them has to make an irrevocably strategic choice aiming at approximatively minimizing a certain cost function. 
In this work, we shall consider two variants of this process: in the first, players choose a strategy approximatively minimizing, up to a factor of $1+\epsilon$, their personal cost ({\em selfish players}), while, in the second, players choose the strategy approximatively minimizing, up to a factor of $1+\epsilon$, the marginal increase in the social cost ({\em cooperative players}) which is defined as the sum of the players' personal costs (for the case of $\epsilon=0$, we use the term {\em exact one-round walk}). Approximate one-round walks can be interpreted as simple greedy online algorithms for the equivalent resource selection problem associated with a given congestion game, and, in most of the cases, these algorithms are optimal in the context of online optimization of load balancing problems \cite{C13}. The worst-case efficiency of these solution concepts with respect to the optimal social cost is termed as the {\em $\epsilon$-approximate price of anarchy} (for the case of pure Nash equilibria, the term {\em price of anarchy} \cite{KP99} is adopted) and as the {\em competitive ratio} of $\epsilon$-approximate one-round walks, respectively.

Interesting special cases of congestion games are obtained by restricting the combinatorics of the players' strategy space. In {\em symmetric congestion games}, all players share the same set of strategies; in {\em network congestion games} the players' strategies are defined as paths in a given network; in {\em matroid congestion games} \cite{ARV08,ARV09}, the strategy set of every player is given by the set of bases of a matroid defined over the set of available resources; in $k$-{\em uniform matroid congestion games} \cite{dKU16}, each player can select any subset of cardinality $k$ from a prescribed player-specific set of resources; in {\em polytopal congestion games} \cite{KleerS17,KleerS21}, the strategies of every player correspond to binary vectors of a player-specific polytope; finally, in {\em load balancing games}, players can only choose single resources. 

To what extent does the structure of the players' strategy space influence the efficiency of decentralized solutions in congestion games? In this work, we investigate whether better performance are possible when restricting to load balancing games. Previous work established that the price of anarchy does not improve when restricting to unweighted load balancing games with polynomial latency functions \cite{CFKKM11,GS07}, while better bounds are possible in unweighted symmetric load balancing games with fairly general latency functions \cite{F10}. Under the assumption of identical resources, improvements are also possible in the following cases: unweighted load balancing games with affine latencies \cite{CFKKM11,STZ07}, unweighted games with polynomial or general latencies \cite{PaccagnanCFM21,VijayalakshmiS20}, weighted symmetric load balancing games with affine latencies \cite{LMMR08} or monomial latencies \cite{GLMM06}. Finally, \cite{BGR10} proves that the price of anarchy does not improve when restricting to weighted symmetric load balancing games under polynomial latency functions. A summary of these known results is shown in Figure \ref{tab1}.

\begin{figure}
\center
\begin{tabular}{||c||c|c|c|c|c|c|c|c||}
\hline
& \sf{LB} & \sf{SYM} & \sf{IDE} & \sf{AFF} or \sf{MONO} & \sf{POLY} & \sf{GEN} & \sf{Improve?} & \sf{Reference}\\\hline\hline
\sf{UNW} & \checkmark & & & & \checkmark & & \sf{NO} & \cite{CFKKM11,GS07}\\\hline
\sf{UNW} & \checkmark & \checkmark & & & & \checkmark & \sf{YES} & \cite{F10}\\\hline
\sf{UNW} & \checkmark & & \checkmark & \checkmark & & & \sf{YES} & \cite{CFKKM11,STZ07}\\\hline
\sf{UNW} & \checkmark & & \checkmark &  & \checkmark  &  & \sf{YES} & \cite{PaccagnanCFM21,VijayalakshmiS20}\\\hline
\sf{UNW} & \checkmark & & \checkmark &  &  & \checkmark & \sf{YES} & \cite{PaccagnanCFM21,VijayalakshmiS20}\\\hline
\sf{WEI} & \checkmark & \checkmark & & & \checkmark & & \sf{NO} & \cite{BGR10}\\\hline
\sf{WEI} & \checkmark & \checkmark & \checkmark & \checkmark & & & \sf{YES} & \cite{GLMM06,LMMR08}\\\hline
\end{tabular}
\caption{Previously known results concerning the price of anarchy. The following acronyms are used: {\sf UNW} for unweighted games, {\sf WEI} for weighted games, {\sf LB} for load balancing games, {\sf SYM} for symmetric games, {\sf IDE} for identical resources, {\sf AFF} for affine latency functions, {\sf MONO} for monomial latency functions, {\sf POLY} for polynomial latency functions, {\sf GEN} for general latency functions (obeying mild assumptions). A check-mark on a column means that the game enjoys the related property. Question {\sf Improve?} asks whether better performance are possible with respect to games not enjoying property {\sf LB}.}\label{tab1}
\end{figure}

For the competitive ratio of exact one-round walks generated by cooperative players, no improvements are possible in unweighted load balancing games with affine latency functions \cite{CFKKM11,STZ07}, while improved performance can be obtained under the additional assumption of identical resources \cite{CFKKM11} (we observe that, in this case, solutions generated by both types of players coincide); however, for weighted players, no improvements are possible even under the assumption of identical resources \cite{C13,CFKKM11}. A summary of these known results is shown in Figure \ref{tab2}. For one-round walks generated by selfish players, instead, no specialized limitations are currently known (except for those inherited by the case of cooperative players in games with identical resources). 

\begin{figure}
\center
\begin{tabular}{||c||c|c|c|c|c|c||}
\hline
& \sf{LB} & \sf{IDE} & \sf{AFF} & \sf{Type of Players} & \sf{Improve?} & \sf{Reference}\\\hline\hline
\sf{UNW} & \checkmark & & \checkmark & \sf{cooperative} & \sf{NO} & \cite{CFKKM11,STZ07}\\\hline
\sf{UNW} & \checkmark & \checkmark & \checkmark & \sf{both} & \sf{YES} & \cite{CFKKM11}\\\hline
\sf{WEI} & \checkmark & \checkmark & \checkmark & \sf{both} & \sf{NO} & \cite{C13,CFKKM11}\\\hline
\end{tabular}
\caption{Previously known results concerning the competitive ratio of exact one-round walks. Symbology is the same as that used in Figure \ref{tab1}.}\label{tab2}
\end{figure}

\subsection{Our Contribution} 

We obtain an almost precise picture of the cases in which improved performance can be obtained in load balancing congestion games. This is done by either solving open problems or extending previously known results to both approximate solution concepts and more general latency functions encompassed within the following definitions.

A class of latency functions $\mathcal{C}$ is {\em closed under ordinate scaling} (resp. {\em abscissa scaling}) if, for any function $f\in\mathcal{C}$ and $\alpha\geq 0$, the function $g$ such that $g(x)=\alpha f(x)$ (resp. $g(x)=f(\alpha x)$) belongs to $\mathcal{C}$. A function $f$ is {\em semi-convex} if $xf(x)$ is convex, it is {\em unbounded} if $\lim_{x\rightarrow\infty}f(x)=\infty$. We observe that the class of polynomial latency functions obeys the following properties:
\begin{itemize}
\item[$\bullet$] it is both abscissa and ordinate scaling,
\item[$\bullet$] all of its functions are semi-convex,
\item[$\bullet$] all of its non-constant functions are unbounded.
\end{itemize}

For the approximate price of anarchy, we prove the following results, which are summarized in Figure \ref{tab3}.
{\em For unweighted players:} if $\mathcal C$ is closed under ordinate scaling, the approximate price of anarchy does not improve when restricting to load balancing games (Theorem \ref{thm3}). {\em For weighted players:} if $\mathcal{C}$ is closed under abscissa and ordinate scaling, the approximate price of anarchy does not improve when restricting to load balancing games (Theorem \ref{thm1}). Furthermore, if $\mathcal{C}$ contains unbounded functions only (except for the eventual constant latency functions), the approximate price of anarchy does not improve even for symmetric load balancing games. However, under the additional hypothesis of identical resources, better performance are still possible. Let $f$ be an increasing, continuous and semi-convex function. We prove that
the approximate price of anarchy of weighted symmetric load balancing games with identical resources having latency function $f$ is at most equal to $$\sup_{x>0}\max_{\lambda\in (0,1)}\frac{\lambda xf(x)+(1-\lambda)[x]_{\epsilon,f} f([x]_{\epsilon,f})}{(\lambda x+(1-\lambda)[x]_{\epsilon,f}) f(\lambda x+(1-\lambda)[x]_{\epsilon,f})},$$
    where $$[x]_{\epsilon,f}:=\textrm{inf}\{t\geq 0:f(x)\leq(1+\epsilon) f(x/2+t)\},$$ and we show that such upper bound is tight under some mild assumptions.

\begin{figure}
\center
\begin{tabular}{||c||c|c|c|c|c|c|c|c||}
\hline
& \sf{LB} & \sf{SYM} & \sf{IDE} & \sf{CUAS} & \sf{CUOS} & \sf{UNB} & \sf{SEMI-CONV} & \sf{Improve?}\\\hline\hline
\sf{UNW} & \checkmark & & & & \checkmark & & & \sf{NO}\\\hline
\sf{WEI} & \checkmark & & & \checkmark & \checkmark & & & \sf{NO}\\\hline
\sf{WEI} & \checkmark & \checkmark & & \checkmark & \checkmark & \checkmark & & \sf{NO}\\\hline
\sf{WEI} & \checkmark & \checkmark & \checkmark & & & & \checkmark & \sf{YES}\\\hline
\end{tabular}
\caption{Our results concerning the approximate price of anarchy. The following additional acronyms are used with respect to those adopted in Figure \ref{tab1}: {\sf CUAS} for games whose latency functions are closed under abscissa scaling, {\sf CUOS} for games whose latency functions are closed under ordinate scaling, {\sf UNB} for games whose all non-constant latency functions are unbounded, {\sf SEMI-CONV} for games whose latency functions are semi-convex.}\label{tab3}
\end{figure}

For the competitive ratio of approximate one-round walks, we prove the following results which are summarized in Figure \ref{tab4}.
{\em For unweighted players:} if $\mathcal C$ is closed under ordinate scaling, the competitive ratio of approximate one-round walks generated by selfish players does not improve when restricting to load balancing games. The same negative result holds with respect to cooperative players under the additional hypothesis that all functions in $\mathcal C$ are semi-convex (Theorem \ref{thm4}). {\em For weighted players:} if $\mathcal{C}$ is closed under abscissa and ordinate scaling, the competitive ratio of approximate one-round walks generated by selfish players does not improve when restricting to load balancing games. Also in this case, to extend the result to cooperative players, we need that all functions in $\mathcal C$ are semi-convex (Theorem~\ref{thm2}). Finally, for the case of identical resources, we design lower bounds on the performance of exact one-round walks in load balancing games with unweighted players. 

\begin{figure}
\center
\begin{tabular}{||c||c|c|c|c|c|c||}
\hline
& \sf{LB} & \sf{CUAS} & \sf{CUOS} & \sf{SEMI-CONV} & \sf{Type of Players} & \sf{Improve?} \\\hline\hline
\sf{UNW} & \checkmark & & \checkmark & & \sf{selfish} & \sf{NO}\\\hline
\sf{UNW} & \checkmark & & \checkmark & \checkmark & \sf{cooperative} & \sf{NO}\\\hline
\sf{WEI} & \checkmark & \checkmark & \checkmark & & \sf{selfish} & \sf{NO}\\\hline
\sf{WEI} & \checkmark & \checkmark & \checkmark & \checkmark & \sf{cooperative} & \sf{NO}\\\hline
\end{tabular}
\caption{Our results concerning the competitive ratio of approximate one-round walks. Symbology is the same as that used in Figure \ref{tab3}.}\label{tab4}
\end{figure}

From a technical point of view, our results are obtained by designing some general load balancing instances parametrized by at most four numbers and two latency functions belonging to $\mathcal{C}$. We show that there exists a choice of these parameters such that the performance of the considered load balancing instances, which thus get measured by some parametrized formulas, match that of general congestion games with latency functions in $\mathcal{C}$. We also show how to quantitatively evaluate these formulas when $\mathcal{C}$ is the set of polynomial of maximum degree $d$.

Independently from our existential results, the load balancing instances we introduce may be useful to get good lower bounds on the performance of load balancing games with latency functions in $\mathcal{C}$, even if we are not able to quantify the exact value of the metric adopted to measure the performance (see Remarks \ref{remalowweig} and \ref{remalowunw}).

\subsection{Significance}

With respect to the (approximate) price of anarchy, it can be appreciated how row $1$ in Figure \ref{tab3} generalizes to approximate equilibria and to more general latency functions row $1$ in Figure \ref{tab1}. Similarly, rows $2$ and $3$ in Figure \ref{tab3} generalize to approximate equilibria and to more general latency functions row $4$ in Figure \ref{tab1} (these issues were posed as open problems in \cite{BGR10}). Finally, row $4$ in Figure \ref{tab3} generalizes to approximate equilibria and to more general latency functions row $5$ in Figure \ref{tab1}. Our negative results, together with the positive ones achieved by \cite{F10} (see row 2 of Figure \ref{tab1}), imply that, whenever resources are not identical, better bounds on the approximate price of anarchy are possible only when dealing with unweighted symmetric load balancing games. 

We would like to stress the following observation here. Row 2 in Figure \ref{tab3} may look redundant, as row 1 tells us that no improvement is possible in unweighted games if the latency functions are ordinate scaling, while, according to row 2, to have the same negative result in weighted games, we additionally need the abscissa scaling property. But, as unweighted games are also weighted ones, we have a representative subclass of weighted games for which no improvement is possible even under the mere assumption of the ordinate scaling property! This is true, but row 2 tells indeed something more than row 1. Everything gets clearer when thinking to what we mean when we say that no improvement is possible in weighted games: we mean that, for any game with $n$ many players having weights defined by a vector $(w_1,\ldots,w_n)$, there exists a load balancing game with the same number of players and the same vector of weights with a non-smaller price of anarchy. So, the fact that this is true when considering a unitary vector of weights does not imply that this remains true for any vector, and in fact, to achieve such a stronger result, we have to assume the additionally hypothesis of abscissa scaling on the latency functions. Determining whether the same impossibility result can be obtained in weighted games by relying on the ordinate scaling property only is an interesting open problem. If this were not possible, it would follow that, in certain situations, restricting to singleton strategies is more effective in weighted games than it can be in unweighted ones. 

For the competitive ratio of (approximate) one-round walks, we have that row 1 of Figure \ref{tab4} provides the first known results for the case of selfish players. This has a number of applications. For instance, \cite{BFFM09} showed that the upper bound on the competitive ratio of exact one-round walks involving selfish players provided in \cite{CMS12} for the case of affine latencies is tight. However, the lower bounding instance provided by \cite{BFFM09} is a general congestion game, and the authors left as open problem of determining whether a load balancing games with the same performance could be achieved. Our findings show that this is possible and also tells us how to directly used to construct a load balancing instance whose performance matches that of general congestion games, thus solving this open question. Row 2 of Figure \ref{tab4} generalizes row 1 of Figure \ref{tab2}, which holds only with respect to exact one-round walks generated by cooperative players in games with affine latency functions. Row 3 of Figure \ref{tab4} implies that the upper bounds provided in \cite{CMS12} and \cite{Klimm19} for the competitive ratio of exact one-round walks in general congestion games with affine and polynomial latency functions, respectively, are tight even for load balancing games. We stress out that, even for general congestion games, no lower bound for the competitive ratio of exact one-round walks involving selfish players was known prior to this work, except for the case of unweighted games with affine latencies \cite{BFFM09}. 

Finally, row 4 of Figure \ref{tab4} generalizes results in \cite{AAGKKV95,C13,CFKKM11} which hold only with respect to exact one-round walks for games with polynomial latency functions.
Finally, the lower bounds that we obtain on the performance of exact one-round walks in load balancing games with unweighted players and identical resources improves and generalizes the previously known lower bound given in \cite{CFKKM11} which holds only for affine latency functions. 

\subsection{Related Work} 

The price of anarchy in congestion games was first considered in \cite{AAE05} and \cite{CK05} and shown to be equal to $5/2$ in unweighted games with affine latency functions. In \cite{CK05}, it is also proved that no improved bounds are possible in both symmetric unweighted games and unweighted network games; these results were improved by \cite{CddU15} which shows that the price of anarchy stays the same even in symmetric unweighted network games. Furthermore, in \cite{AAE05}, it is shown that the price of anarchy of weighted congestion games with affine latency functions is $(3+\sqrt{5})/2$. 

In \cite{CFKKM11}, it is shown that the previous bounds are tight also for load balancing games. For the special case of load balancing games on identical resources, the works of \cite{STZ07} and \cite{CFKKM11} show that the price of anarchy is $2.012067$ for unweighted games and at least $5/2$ for weighted ones. The works of  \cite{PaccagnanCFM21,VijayalakshmiS20} generalize the above results on affine unweighted games with identical resources, and provide tight bounds on the price of anarchy even for polynomial and more general latencies. In \cite{LMMR08}, it is proved that, for symmetric load balancing games, the price of anarchy drops to $4/3$ if the games are unweighted, and to $9/8$ if the games are weighted with identical resources. This last result has been generalised by \cite{GLMM06}, who provide tight bounds on the price of anarchy of symmetric weighted load balancing games with identical resources and monomial latency functions. For symmetric unweighted $k$-uniform matroid congestion games with affine latency functions, \cite{dKU16} proves that the price of anarchy is at most $28/13$ and at least $1.343$ for a sufficiently large value of $k$ (for $k=5$, it is roughly $1.3428$). Tight bounds on the price of anarchy of either weighted and unweighted congestion games with polynomial latency functions have been given by \cite{ADGMS11}. Under fairly general latency functions, \cite{F10} shows that the price of anarchy of unweighted symmetric load balancing games coincides with that of non-atomic congestion games, while \cite{BGR10} proves that assuming symmetric strategies does not lead to improved bounds in unweighted games and gives exact bounds for the case of weighted players. It also shows that, for the case of weighted players, no improvements are possible even in symmetric load balancing games with monomial and polynomial latency functions. Upper and lower bounds on the pure and mixed price of anarchy for several classes of load balancing games with polynomial latency functions are provided in \cite{GLMM06,GS07,GLMMR08}. Finally, \cite{CKS11} (resp. \cite{B12}) provides tight bounds (resp. almost tight upper bounds) on the approximate price of anarchy of unweighted (resp. weighted) congestion games under affine latency functions. \cite{Gian18} obtain tight bounds on the approximate price of anarchy of congestion games with polynomial latency functions, holding even for load balancing games; prior to \cite{Gian18}, analogue results have been obtained in the preliminary version of our work. 

The competitive ratio of exact one-round walks generated by cooperative players in load balancing games with polynomial latency functions has been first considered in \cite{AAGKKV95}, where, for the special case of affine functions, an upper bound of $3+2\sqrt{2}$ is provided for weighted players. For unweighted players, this result has been improved to $17/3$ in \cite{STZ07}, where it is also shown that, for identical resources, the upper bound drops to $2+\sqrt{5}$ in spite of a lower bound of $3.0833$. Finally, \cite{CFKKM11} show matching lower bounds of $3+2\sqrt{2}$ and $17/3$ for, respectively, weighted and unweighted players. For weighted games with polynomial latency functions, tight bounds have been given in \cite{C13}; the lower bounds, in particular, hold even for identical resources, thus improving previous results from \cite{AAGKKV95}. In \cite{CFKKM11} it is also shown that, for unweighted players and identical resources, the competitive ratio lies between $4$ and $\frac 2 3 \sqrt{21}+1$ if latency functions are affine. 

For the case of selfish players and still under affine latency functions, \cite{BFFM09,CMS12} show that the competitive ratio is $2+\sqrt{5}$ for unweighted congestion games, while, for weighted players, \cite{CMS12} gives an upper bound of $4+2\sqrt{3}$. For the more general case of polynomial latency functions, \cite{Klimm19} determines explicit and good upper bounds on the competitive ratio in unweighted and weighted congestion games; prior to \cite{Klimm19}, analogue results for weighted games with polynomial latency functions have been obtained in the preliminary version of our work \cite{BV17}. 

In \cite{BV16}, it is shown that, if latency functions are polynomials of degree $d$ and players are selfish, the $(d+1)$-th ordered Bell number is a lower-bound on the competitive ratio, even for unweighted games with identical resources. The concept of one-round walks starting from the empty state in non-atomic congestion games is analysed in \cite{V18}. 

\section{Definitions and Notation}
For two integers $0\leq k_1\leq k_2$, let $[k_1]_{k_2}:=\{k_1,k_1+1,\ldots,k_2-1, k_2\}$ and $[k_1]:=[k_1]_1$. For a real function $F:\ZP\rightarrow \R$, $[<]_s\left(F(s)\right)$ (resp. $[\leq]_s\left(F(s)\right)$) denotes an arbitrary function $G:\ZP\rightarrow \R$ such that $\lim_{s\rightarrow \infty}\frac{|G(s)|}{|F(s)|}=0$ (resp. $\limsup_{s\rightarrow \infty}\frac{|G(s)|}{|F(s)|}\leq c$ for some $c>0$).\footnote{We have not adopted the classical notation used in asymptotic analysis for two main reasons: (i) to explicate that we evaluate the asymptotic behaviour with respect to a certain variable $s$ only, while other variables/parameters are keep fixed; (ii) to avoid ambiguity with the symbol ``o", that will be often used to denote other quantities.}

\paragraph*{\bf Congestion Games.}\hspace{0.2cm}A {\em congestion game} is a tuple $$\CG=\left(\N,E,(\ell_e)_{e\in E},(w_i)_{i\in\N},({\sf\Sigma}_i)_{i\in\N}\right),$$ where $\N$ is a set of $n\geq 2$ players, $E$ is a set of resources, $\ell_e:\RPP\rightarrow\RPP$ is the (non-decreasing) latency function of resource $e\in R$, and, for each $i\in\N$, $w_i\geq  0$ is the weight of player $i$ and ${\sf\Sigma}_i\subseteq 2^R\setminus \emptyset$ is her set of strategies. Conventionally, we set $\ell_e(0)=\lim_{x\rightarrow 0^+}\ell_e(x)$ for any latency function $\ell_e$.\footnote{All the results obtained in this paper can be easily generalized to the case in which some latency function $\ell_e$ verifies $\ell_e(x)=0$ for some $x>0$, but this requires further (tedious) case analysis to cope with undefined ratios of type $c/0$.} We speak of {\em weighted} games/players when players have arbitrary weights and of {\em unweighted} games/players when $w_i=1$ for each $i\in\N$. A congestion game is {\em symmetric} if ${\sf\Sigma}_i={\sf\Sigma}$ for each $i\in\N$, i.e., if all players share the same strategy space. A {\em load balancing} game is a congestion game in which for each $i\in\N$ and $S\in{\sf\Sigma}_i$, $|S|=1$, that is, all players' strategies are singleton sets. Given a class $\mathcal{C}$ of latency functions, let ${\sf W}(\mathcal{C})$ be the class of weighted congestion games, ${\sf U}(\mathcal{C})$ be the class of unweighted congestion games, ${\sf ULB}(\mathcal{C})$ be the class of unweighted load balancing games, ${\sf WLB}(\mathcal{C})$ be the class of weighted load balancing games, and ${\sf WSLB}(\mathcal{C})$ be the class of weighted symmetric load balancing games, all having latency functions in class $\mathcal{C}$.

\paragraph*{\bf Latency Functions.}\hspace{0.2cm} A congestion game has {\em polynomial latencies of maximum degree $d$} when, for each $e\in E$, $\ell_e(x)=\sum_{h\in [d]_0} \alpha_{e,h}x^h$, with $\alpha_{e,h}\geq 0$ for each $h\in [d]_0$. When $d=0$ we speak of {\em constant latencies} and when $d=1$, we speak of {\em affine latencies}. If $\ell_e(x)=\alpha_{e,d}x^d$, we speak of {\em monomial latencies of degree $d$}, and if $d=1$, we speak of {\em linear latencies}. Let $\mathcal{P}(d)$ denote the class of polynomial latencies of maximum degree $d$. A latency function $f$ is {\em semi-convex} if $xf(x)$ is convex, and it is {\em unbounded} if $\lim_{x\rightarrow\infty}f(x)=\infty$. A congestion game has {\em identical resources} when all resources have the same latency. Given a class $\mathcal{G}$ of congestion games, let $\mathcal{C}(\mathcal{G})$ denote the class of latency functions of congestion games belonging to $\mathcal{G}$. A class $\mathcal{C}$ of latency functions is {\em closed under ordinate scaling} (resp. {\em abscissa scaling}) if, for any function $f\in\mathcal{C}$ and $\alpha> 0$, the function $g$ such that $g(x)=\alpha f(x)$ (resp. $g(x)=f(\alpha x)$) belongs to $\mathcal{C}$. 

\paragraph*{\bf Strategy Profiles and Cost Functions.}\hspace{0.2cm} A {\em strategy profile} is an $n$-tuple of strategies ${\bm\sigma}=(\sigma_1,\ldots,\sigma_n)$, that is, a state of the game in which each player $i\in\N$ is adopting strategy $\sigma_i\in{\sf\Sigma}_i$, so that ${\sf\Sigma}:=\times_{i\in\N}{\sf\Sigma}_i$ denotes the set of strategy profiles which can be realized in $\CG$. For a strategy profile $\bm\sigma$, the {\em congestion} of resource $e\in E$ in $\bm\sigma$, denoted as $k_e({\bm\sigma}):=\sum_{i\in\N:e\in\sigma_i}w_i$, is the total weight of the players using resource $e$ in $\bm\sigma$ (observe that, in unweighted games, $k_e({\bm\sigma})$ coincides with the number of users of $e$ in $\bm\sigma$). The {\em personal cost} of player $i$ in $\bm\sigma$ is defined as $cost_i({\bm\sigma})=\sum_{e\in\sigma_i}\ell_e(k_e({\bm\sigma}))$ and each player aims at minimizing it. For the sake of conciseness, when the strategy profile $\bm\sigma$ is clear from the context, we write $k_e$ in place of $k_e({\bm\sigma})$. Fix a strategy profile $\bm\sigma$ and a player $i\in\N$. We denote with ${\bm\sigma}_{-i}$ the restriction of $\bm\sigma$ to all players other than $i$; moreover, for a strategy $S\in{\sf\Sigma}_i$, we denote with $({\bm\sigma}_{-i},S)$ the strategy profile obtained from $\bm\sigma$ when player $i$ changes her strategy from $\sigma_i$ to $S$, while the strategies of all the other players are kept fixed. The quality of a strategy profile in congestion games is measured by using the {\em social function} ${\sf SUM}({\bm \sigma})=\sum_{i\in\N}w_icost_i({\bm \sigma})=\sum_{e\in E}k_e({\bm\sigma})\ell_e(k_e({\bm\sigma}))$, that is, the sum of the players' personal costs. A {\em social optimum} is a strategy profile ${\bm \sigma}^*$ minimizing $\sf SUM$. For the sake of conciseness, once a particular social optimum has been fixed, we write $o_e$ to denote the value $k_e({\bm\sigma}^*)$.

\paragraph*{\bf Solution Concepts.}\hspace{0.2cm}For any $\epsilon\geq 0$, an {\em $\epsilon$-approximate pure Nash equilibrium} is a strategy profile $\bm\sigma$ such that, for any player $i\in\N$ and strategy $S\in{\sf\Sigma}_i$, $cost_i({\bm\sigma})\leq (1+\epsilon) cost_i({\bm\sigma}_{-i},S)$. For $\epsilon=0$, we speak of an {\em (exact) pure Nash equilibrium}. We denote by ${\sf NE}_\epsilon({\sf CG})$ the set of $\epsilon$-approximate pure Nash equilibria of a congestion game $\sf CG$. 

For any $\epsilon\geq 0$, an {\em $\epsilon$-approximate one-round walk} is an online process in which players appear sequentially according to an arbitrary order and, upon arrival, each player irrevocably chooses a strategy approximatively minimizing a certain cost function. Formally, an $\epsilon$-approximate one-round walk involving selfish (or cooperative) players is an $n+1$-tuple $\bm \tau=(\bm \sigma^0,\bm \sigma^1,\bm \sigma^2,\ldots, \bm \sigma^n)$ defined as follows:
\begin{itemize}
\item[$\bullet$] assuming without loss of generality that players enter sequentially according to their indices, ${\bm \sigma}^i=(\sigma_1^i,\sigma_2^i,\ldots,\sigma_i^i,\underbrace{\emptyset,\emptyset,\ldots, \emptyset}_{n-i})$ is the strategy profile obtained when the first $i$ players have performed their strategic choice, while the remaining ones have not entered the game yet (so, we assume that each of them is playing the empty strategy);
\item[$\bullet$] if players are selfish, the $i$-th player aims at minimizing her personal cost, so that $cost_i({\bm \sigma}^i)\leq (1+\epsilon)\min_{S\in{\sf\Sigma}_i}cost_i({\bm \sigma}^{i-1},S)$; otherwise, if players are cooperative, the $i$-th player aims at minimizing the marginal increase in the social function $\sf SUM$, so that ${\sf SUM}({\bm \sigma}^i)-{\sf SUM}({\bm \sigma}^{i-1})\leq (1+\epsilon)\min_{S\in{\sf\Sigma}_i}({\sf SUM}({\bm \sigma}^{i-1},S)-{\sf SUM}({\bm \sigma}^{i-1}))$\footnote{An alternative definition for the concept of $\epsilon$-approximate one-round walk involving cooperative players can be obtained by assuming that the $i$-th player aims at minimizing the value of the social function, i.e., ${\sf SUM}({\bm \sigma}^i)\leq (1+\epsilon)\min_{S\in{\sf\Sigma}_i}{\sf SUM}({\bm \sigma}^{i-1},S)$. For $\epsilon=0$ the two definitions are equivalent; for $\epsilon>0$, the proof arguments that we will consider in this paper can be easily adapted/modified to get similar results holding under this alternative definition.}. 
\end{itemize}
For $\epsilon=0$, we speak of a(n) {\em (exact) one-round walk}. We denote by ${\sf ORW}^s_\epsilon({\sf CG})$ (resp. ${\sf ORW}^c_\epsilon({\sf CG})$) the set of strategy profiles $\sg=\bm \sigma^n$ which can be constructed by an $\epsilon$-approximate one-round walk involving selfish (resp. cooperative) players in a congestion game $\sf CG$.

\paragraph*{\bf Quality Metrics.}\hspace{0.2cm}The {\em $\epsilon$-approximate price of anarchy} of a congestion game $\sf CG$ is defined as ${\sf PoA}_\epsilon({\sf CG})=\max_{{\bm\sigma}\in{\sf NE}_\epsilon({\sf CG})}\frac{{\sf SUM}({\bm\sigma})}{{\sf SUM}({\bm\sigma}^*)}$, where ${\bm\sigma}^*$ is a social optimum for $\CG$. Similarly, the {\em competitive ratio} of $\epsilon$-approximate one-round walks generated by selfish (resp. cooperative) players, is defined as ${\sf CR}^s_\epsilon({\sf CG})=\max_{{\bm\sigma}\in{\sf ORW}^s_\epsilon({\sf CG})}\frac{{\sf SUM}({\bm\sigma})}{{\sf SUM}({\bm\sigma}^*)}$ (resp. ${\sf CR}^c_\epsilon({\sf CG})=\max_{{\bm\sigma}\in{\sf ORW}^c_\epsilon({\sf CG})}\frac{{\sf SUM}({\bm\sigma})}{{\sf SUM}({\bm\sigma}^*)}$).
Given a class of congestion games $\mathcal G$, the $\epsilon$-approximate price of anarchy of $\mathcal G$ is defined as ${\sf PoA}_{\epsilon}({\mathcal G})=\textrm{sup}_{{\sf CG}\in{\mathcal G}}{\sf PoA}_\epsilon({\sf CG})$. For the case of $\epsilon=0$, we refer to this metric simply as to the {\em price of anarchy}. The competitive ratio of $\epsilon$-approximate one-round walks of $\mathcal G$ generated by both selfish and cooperative players is defined accordingly.
Throughout the paper, we shall assume that, in any considered class of latency functions, there always exists a non-constant latency function.\footnote{We observe that, if all the latency functions are constant, the value of all quality metrics is at most $1+\epsilon$, and such upper bound is tight.}

\section{Weighted Load Balancing Games}\label{secweigh}
In this section, we show that, under mild assumptions on the latency functions, the $\epsilon$-approximate price of anarchy and the competitive ratio of $\epsilon$-approximate selfish/cooperative one-round walks of weighted congestion games cannot improve even when restricting to load balancing games. 
\subsection{Preliminary Definitions and Technical Lemmas}
We first give some definitions and two technical lemmas, whose proofs are based on the primal-dual method \cite{B12}, and are deferred to Section~\ref{missing} of the appendix (since they are quite technical and do not constitute the primary focus of this work). Lemma~\ref{lemupp} gives quasi-explicit formulas to compute upper bounds for the considered efficiency metrics, and such bounds are obtained by fixing an arbitrary congestion game, and by exploiting the dual linear program of a parametric primal linear program whose objective function returns the value of the considered efficiency metric. Then, Lemma~\ref{lemrou1} provides some parameters that will be used in the main theorems of this section to design ad-hoc load balancing instances whose performance match the upper bounds of Lemma~\ref{lemupp}.

We point out that, the upper bounds provided in Lemma~\ref{lemupp} can be equivalently derived by resorting to the smoothness framework \cite{R15} (further details are given in Section~\ref{frame} of the appendix). Prior to this work, similar results as in Lemma \ref{lemrou1} have been given in \cite{R15,BGR10,V18}, with the aim of designing parametric load balancing instances whose performance match the known upper bounds. 

Given $x\geq 1$, $k>0$, $o>0$, and a latency function $f$, let
\begin{equation}\label{def_alpha}
\beta_{\sf W}({\sf EM},k,o,f):=
\begin{cases}
-kf(k)+(1+\epsilon)of(k+o) & \text{if }{\sf EM}=\poa_\epsilon ,\\
-\int_{0}^kf(t)\text{dt}+(1+\epsilon)of(k+o) & \text{if }{\sf EM}={\sf CR}_\epsilon^s,\\
-kf(k)+(1+\epsilon)((k+o)f(k+o)-kf(k)) & \text{if }{\sf EM}={\sf CR}_\epsilon^c.
\end{cases}
\end{equation}
Given ${\sf EM}\in \{\poa_{\epsilon},{\sf CR}^s_\epsilon, {\sf CR}^c_\epsilon\}$, let 
\begin{equation}\label{def_gamma_prel}
\gamma_{\sf W}({\sf EM},x,k,o,f):=\frac{kf(k)+x\cdot \beta_{\sf W}({\sf EM},k,o,f)}{of(o)};
\end{equation}
furthermore, given a class of weighted congestion games $\mathcal{G}$, let 
\begin{equation}\label{def_gamma}
\gamma_{\sf W}({\sf EM},\mathcal{G}):=\inf_{x\geq 1}\sup_{f\in\mathcal{C}(\mathcal{G}),k> 0, o>0}\gamma_{\sf W}({\sf EM},x,k,o,f).
\end{equation}
\begin{lemma}\label{lemupp}
Let $\mathcal{G}$ be a class of weighted congestion games. For any  ${\sf EM}\in \{\poa_{\epsilon},{\sf CR}^s_\epsilon\}$ we have that ${\sf EM}(\mathcal{G})\leq \gamma_{\sf W}({\sf EM},\mathcal{G})$. This fact holds for ${\sf EM}={\sf CR}^c_\epsilon$ if the latency functions of $\mathcal{C}(\mathcal{G})$ are semi-convex.
\end{lemma}
The proof of Lemma~\ref{lemupp} is deferred to Section \ref{missing} of the appendix, and we give a brief overview of the case ${\sf EM}=\poa_{\epsilon}$. 
\begin{proof}[Sketch of the proof of Lemma \ref{lemupp}]
The maximum value of the following linear program in the variables $\alpha_e$'s is an upper bound on $\poa_\epsilon({\sf CG})$ (recall that $w_i$ is the weight of player $i$, $k_e$ and $o_e$ are the equilibrium and optimal congestions of resource $e$, respectively):\\
\begin{align}
{\sf LP1:}\quad \max\quad& \overbrace{\sum_{e\in E}\alpha_ek_e\ell_e(k_e)}^{{\sf SUM}(\bm\sigma)}\nonumber\\
s.t. \quad& \sum_{e\in \sigma_i}\alpha_e \ell_e(k_e)\leq (1+\epsilon)\sum_{e\in \sigma_i^*}\alpha_e \ell_e(k_e+w_i),\quad \forall i\in \N\label{const1main}\\
&\overbrace{\sum_{e\in E}\alpha_eo_e\ell_e(o_e)}^{{\sf SUM}(\bm\sigma^*)}= 1\label{const1bmain}\\
&\alpha_e\geq 0,\quad\forall e\in E,\nonumber
\end{align}
Indeed, by setting $\alpha_e=1$ for any $e\in E$, we have that: (i) the objective function is the social cost at the equilibrium; (ii) the constraints in  (\ref{const1main}) impose some relaxed $\epsilon$-approximate pure Nash equilibrium conditions (ensuring that each agent, at the equilibrium, does not get any benefit when deviating in favour of strategy $\sigma^*_i$); (iii) (\ref{const1bmain}) is the normalized optimal social cost (normalization is possible since there is some $o_e>0$, that implies $o_e\ell_e(o_e)>0$). We introduce further relaxations on {\sf LP1} (by replacing all the constraints in \eqref{const1main} with their weighted sum) and we get the following linear program:
\begin{align*}
{\sf LP2:}\quad \max\quad& \sum_{e\in E}\alpha_ek_e\ell_e(k_e)\nonumber\\
s.t. \quad&\sum_{e\in E}\alpha_e\beta_{\sf W}(\poa_\epsilon,k_e,o_e,\ell_e)\geq 0\nonumber\\
&\sum_{e\in E}\alpha_e o_e\ell_e(o_e)=1\\
&\alpha_e\geq 0,\quad\forall e\in E,\nonumber
\end{align*}
where $\beta_{\sf W}(\poa_\epsilon,k_e,o_e,\ell_e)$ is defined in \eqref{def_alpha}. By taking the dual of {\sf LP2}, we get the following linear program:
\begin{align*}
{\sf DLP:}\quad \min\quad & \gamma \nonumber\\
s.t.\quad & \gamma\cdot o_e\ell_e(o_e) \geq k_e\ell_e(k_e)+x\cdot \beta_{\sf W}({\sf PoA}_\epsilon,k_e,o_e,\ell_e),\quad \forall e\in E\\
&x\geq 0,\gamma\in \R\nonumber
\end{align*}
Finally, we show that there exists a feasible solution $(x,\gamma)$ of the dual that guarantees a value of at most $\gamma_{\sf W}(\poa_{\epsilon},\mathcal{G})$; by weak duality, this is an upper bound on the optimal value of {\sf LP}, and then an upper bound on the $\epsilon$-approximate price of anarchy.\qed
\end{proof}
\begin{lemma}\label{lemrou1}
Let $\mathcal{G}$ be a class of weighted congestion games. For each $M<\gamma_{\sf W}({\sf EM},\mathcal{G})$, with ${\sf EM}\in \{\poa_{\epsilon},{\sf CR}^s_\epsilon\}$, one of the following cases holds:
\begin{itemize}
\item[$\bullet$] {\bf Case 1:} there exists a non-constant latency function $f\in\mathcal{C}(\mathcal{G})$ and two real numbers $k,o>0$ such that $
M<\frac{kf(k)}{of(o)}$ and $\beta_{\sf W}({\sf EM},k,o,f)\geq 0$. 
\item[$\bullet$] {\bf Case 2:} there exist two latency functions $f_1,f_2\in \mathcal{C}(\mathcal{G})$ and four real numbers $k_1,k_2,o_1,o_2>0$ such that $M<	\frac{\alpha_1k_1f_1(k_1)+\alpha_2k_2f_2(k_2)}{\alpha_1o_1f_1(o_1)+\alpha_2o_2f_2(o_2)}$, where $\alpha_1:=\beta_{\sf W}({\sf EM},k_2,o_2,f_2)>0$ and $\alpha_2 :=-\beta_{\sf W}({\sf EM},k_1,o_1,f_1)>0$; furthermore, if ${\sf EM}=\poa_\epsilon$, $f_1$ and $f_2$ can be chosen as non-constant latency functions. 
\end{itemize}
This fact holds for ${\sf EM}={\sf CR}_\epsilon^c$ if the latency functions of $\mathcal{C}(\mathcal{G})$ are semi-convex. 
\end{lemma}
The proof of Lemma~\ref{lemrou1} is quite technical and does not represent a fundamental result of such work, thus it is deferred to Section \ref{missing} of the appendix. Anyway, we give a sketch of the proof clarifying the main steps. 
\begin{proof}[Sketch of the proof of Lemma \ref{lemrou1}]
The claim of Lemma~\ref{lemrou1} is obtained by reversing the proof arguments of Lemma~\ref{lemupp} as follows: (i) we start from the upper bound $\gamma_{\sf W}({\sf EM},\mathcal{G})$ obtained in Lemma~\ref{lemupp}; (ii) we derive a linear program $\overline{\sf DLP}$ similar as the dual program used in Lemma~\ref{lemupp}, but with two constraints only (except those that impose the non-negativity of the variables), and whose optimal value is higher than $M$; (iii) the dual of $\overline{\sf DLP}$ is similar to the linear program {\sf LP2} used in Lemma~\ref{lemupp}, but with two constraints and two variables, and has the same optimal value (higher than $M$) as in $\overline{\sf DLP}$ (by strong duality); in particular, the new linear program is explicitly defined as follows:
\begin{align*}
\overline{\sf LP}:\quad \max\quad & \alpha_1k_1 f_1(k_1)+\alpha_2k_2 f_2(k_2)\nonumber\\
s.t. \quad &\alpha_1\beta_{\sf W}({\sf EM},k_1,k_1,f_1)+\alpha_2\beta_{\sf W}({\sf EM},k_2,o_2,f_2)\geq 0\\
& \alpha_1o_1 f_1(o_1)+\alpha_2o_2 f_2(o_2)= 1\\
&\alpha_1,\alpha_2\geq 0;\nonumber
\end{align*}
(iv) finally, the claim of Lemma~\ref{lemrou1} is obtained by characterizing the optimal solution of $\overline{\sf LP}$. \qed
\end{proof}
\begin{remark}\label{rema_lemrou}
By exploiting Lemma \ref{lemrou1}, given an arbitrary choice of the input parameters $k,o,f$ (resp. $k_1,k_2,o_1,o_2,f_1,f_2$) such that $\beta_{\sf W}({\sf EM},k,o,f)\geq 0$ (resp. $\beta_{\sf W}({\sf EM},k_2,o_2,f_2)>0$ and $-\beta_{\sf W}({\sf EM},k_1,o_1,f_1)>0$), we have that $\frac{kf(k)}{of(o)}$ (resp. $\frac{\alpha_1k_1f_1(k_1)+\alpha_2k_2f_2(k_2)}{\alpha_1o_1f_1(o_1)+\alpha_2o_2f_2(o_2)}$) is a lower bound on $\gamma_{\sf W}({\sf EM},\mathcal{G})$; furthermore, by definition of $\gamma_{\sf W}({\sf EM},\mathcal{G})$, we have that $\sup_{f\in\mathcal{C}(\mathcal{G}),k,o>0}\gamma_{\sf W}({\sf EM},x,k,o,f)$ is an upper bound on $\gamma_{\sf W}({\sf EM},\mathcal{G})$ for any $x\geq 1$. Thus, if such upper bound (for a suitable choice of $x$) is equal to the former lower bound (for a suitable choice of the input parameters), we necessarily have that they are both equal to $\gamma_{\sf W}({\sf EM},\mathcal{G})$.
\end{remark}
\begin{remark}\label{remalem}
In Lemma~\ref{lemrou1}, if $\mathcal{C}(\mathcal{G})$ is closed under abscissa and ordinate scaling, we can assume without loss of generality that $o_j=1$ for each $j\in [2]$ (remove index $j$ if the first case of Lemma~\ref{lemrou1} is verified). Indeed, if it is not the case, let $\hat{f}_j\in\mathcal{C}(\mathcal{G})$ be the latency function such that $\hat{f}_j(k):=o_jf_j(o_jk)$ for any $k\geq 0$. If we consider tuple $(k_j/o_j,1,\hat{f}_j)$ in place of tuple $(k_j,o_j,f_j)$ for each $j\in [2]$, the claim of Lemma~\ref{lemrou1} follows as well. Analogously, if $\mathcal{C}(\mathcal{G})$ is closed under abscissa and ordinate scaling, the supremum appearing in the definition of $\gamma_{\sf W}({\sf EM},\mathcal{C})$ does not decrease when assuming $o=1$; thus, also Lemma \ref{lemupp} holds under this assumption. 
\end{remark}
\begin{example}
To illustrate the claims of Theorem~\ref{lemupp} and Lemma~\ref{lemrou1} with a concrete example, we consider the problem of evaluating the competitive ratio of exact one-round walks involving selfish players for affine weighted congestion games. Because of Theorem~\ref{lemupp} and Remark \ref{remalem}, we get that:
\begin{align}
&\gamma_{\sf W}({\sf CR}^s_0,\mathcal{P}(1))\nonumber\\
&=\inf_{x\geq 1}\sup_{f\in\mathcal{P}(1),k\geq 0}\left(kf(k)+x\left(f(k+1)-\int_{0}^kf(t)\text{dt}\right)\right)\nonumber\\
&=\inf_{x\geq 1}\sup_{\substack{\alpha_0,\alpha_1\geq 0\\k> 0}}\left(\sum_{d\in \{0,1\}} \alpha_{d} k^{d+1}+x\left(\sum_{d\in \{0,1\}} \alpha_{d} (k+1)^{d}-\sum_{d\in \{0,1\}} \alpha_{d} \frac{k^{d+1}}{d+1}\right)\right)\nonumber\\
&=\inf_{x\geq 1}\sup_{d\in \{0,1\},k> 0}\left(k^{d+1}+x\left((k+1)^{d}-\frac{k^{d+1}}{d+1}\right)\right)\nonumber\\
&=\min_{x>2}\sup_{k> 0}\left(k^2+x\left(k+1-\frac{k^2}{2}\right)\right)\nonumber\\
&\leq\gamma\left({\sf CR}^s_0,\overbrace{\frac{2\sqrt{3}+6}{3}}^{x},\overbrace{\frac{\sqrt{3}+3}{\sqrt{3}}}^{k},\overbrace{1}^o,f:f(t)=t\right)\label{gammaoneone}\\
&= 2\sqrt{3}+4\nonumber
\end{align}
thus obtaining the same upper bound of \cite{CMS12}. Relatively to Lemma~\ref{lemrou1}, we have that the first case is verified. In particular, by setting $k=\frac{\sqrt{3}+3}{\sqrt{3}}$, $o=1$  and $f$ in such a way that $f(t)=t$ for any $t\geq 0$ (i.e. the parameters $k,o,f$ used in (\ref{gammaoneone})), we get that the second inequality of the first case of Lemma~\ref{lemrou1} is tight (i.e., $\beta_{\sf W}({\sf CR}_\epsilon^s,k,o,f)=0$), and $\frac{kf(k)}{of(o)}$ is equal to $2\sqrt{3}+4$. Thus, by Remark \ref{rema_lemrou}, we necessarily have that $\gamma_{\sf W}({\sf CR}^s_0,\mathcal{P}(1))=2\sqrt{3}+4$.
\end{example}
\subsection{Main Theorems}
In the following theorem, we prove that, under mild assumptions on the latency functions, no improvements are possible for $\epsilon$-approximate pure Nash equilibria even when restricting to load balancing games.
\begin{theorem}\label{thm1}
Let $\mathcal{C}$ be a class of latency functions that is closed under ordinate and abscissa scaling. Then, $\gamma_{\sf W}(\poa_{\epsilon},{\sf W}(\mathcal{C}))=\poa_{\epsilon}({\sf W}(\mathcal{C}))=\poa_{\epsilon}({\sf WLB}(\mathcal{C})).$ If all latency functions of $\mathcal{C}$ (except for the constant ones) are unbounded, we get $\gamma_{\sf W}(\poa_{\epsilon},{\sf W}(\mathcal{C}))=\poa_{\epsilon}({\sf W}(\mathcal{C}))=\poa_{\epsilon}({\sf WSLB}(\mathcal{C}))$.
\end{theorem}
\begin{proof}
First of all, we consider the case in which all latency functions, except for the constant ones, are unbounded. Fix $M<\gamma_{\sf W}(\poa_{\epsilon},{\sf W}(\mathcal{C}))$. In the proof, we will show that there exists a load balancing instance $\LB\in{\sf WSLB}(\mathcal{C})$ such that $\poa_\epsilon(\LB)>M$. By Lemma~\ref{lemupp}, this will imply that $M<\poa_{\epsilon}(\LB)\leq \poa_{\epsilon}({\sf W}(\mathcal{C}))\leq \gamma_{\sf W}(\poa_{\epsilon},{\sf W}(\mathcal{C}))$, and by the arbitrariness of $M<\gamma_{\sf W}(\poa_{\epsilon},{\sf W}(\mathcal{C}))$, this fact will show the claim. We make use of a multi-graph, to represent a load balancing game together with two special strategy profiles. This multi-graph is denoted as {\em load balancing graph}, and is defined as follows: the nodes are all the resources in $E$, and each player is associated to a weighted edge $(e_1,e_2,w)$, where $\{e_1\}$ is denoted as her {\em first strategy}, $\{e_2\}$ is her {\em second strategy}, and $w$ is her weight. With a little abuse of notation, any load balancing graph will be also used to denote the  corresponding load balancing game. 

If Case 2 of Lemma~\ref{lemrou1} holds (with respect to $\mathcal{C}$), let $k_1,k_2,o_1,o_2,f_1,f_2,\alpha_1=\beta_{\sf W}({\sf PoA}_\epsilon,k_2,o_2,f_2),\alpha_2=-\beta_{\sf W}({\sf PoA}_\epsilon,k_1,o_1,f_1)$ be the parameters considered in the claim of Lemma~\ref{lemrou1}. By Lemma~\ref{lemrou1}, $f_1$ and $f_2$ can be chosen among the non-constant latencies of $\mathcal{C}$, thus, by hypothesis, they are unbounded. Furthermore, by Remark \ref{remalem}, we can assume without loss of generality that $o_1=o_2=1$. 

Let $s,n\in \NN$ be two arbitrary positive integers. Consider a load balancing graph $\LB_{s,n}(k_1,k_2,f_1,f_2)$ yielded by a directed $n$-ary tree\footnote{With a little abuse of notation, the value of $n$ considered in such load-balancing instance does not represent the total number of players, but the number of players selecting each resource when playing their first strategy.}, organized in $2s$ levels, numbered from $1$ to $2s$, and whose edges are oriented from the root to the leaves, with the addition of $n$ self-loops on the nodes of level $2s$. The weight $w_i$ of a player associated to an edge outgoing from a node at level $i\in [s]$ (resp. $i\in [2s]_{s+1}$) is equal to $(k_1/n)^i$ (resp. $(k_1/n)^s(k_2/n)^{i-s}$). For $i,j\in [2]$, define $$\theta_{i,j}:=\frac{f_i(k_i)}{(1+\epsilon)f_j(k_j+1)}$$ and $\theta_i:=\theta_{i,i}$. Each resource at level $i$ has latency 
$$g_i(x):=\begin{cases}
\theta_{1}^{i-1}f_1\left(\left(\frac{n}{k_1}\right)^{i-1}x\right) & \text{ if }i\in [s],\\
\theta_{1}^{s-1}\theta_{1,2}\theta_{2}^{i-s-1}f_2\left(\left(\frac{n}{k_1}\right)^{s}\left(\frac{n}{k_2}\right)^{i-s-1}x\right)&\text{ otherwise.}
\end{cases}
$$
See Figure \ref{fig:56} for an example. 
\begin{figure}
\centering
\includegraphics[scale=0.5]{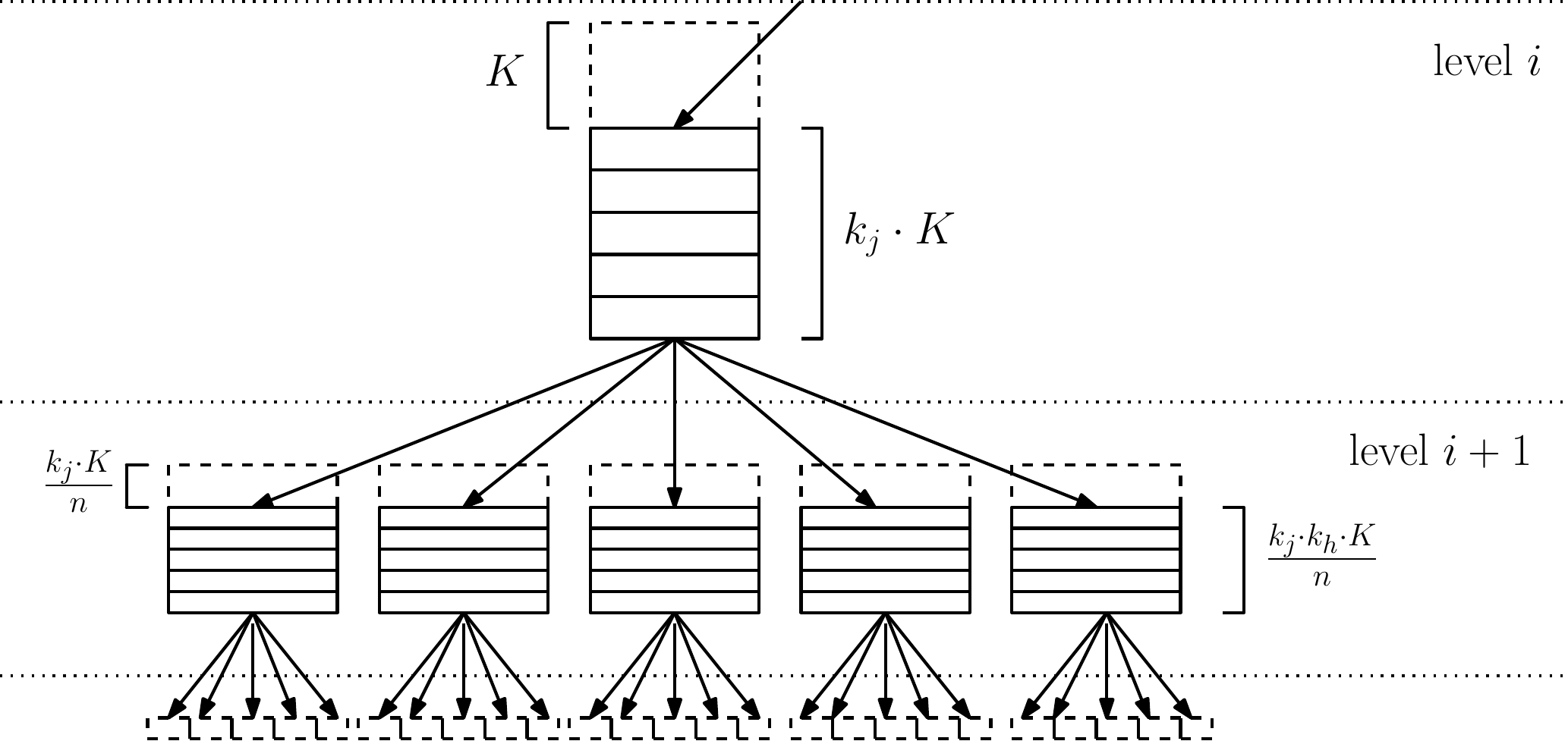}
\caption{Figure \ref{fig:56} depicts a node/resource at some level $i\in [2s-2]_2$ of the load balancing graph $\LB_{s,n}(k_1,k_2,f_1,f_2)$ with $n=5$, and all the nodes/resources at level $i+1$ connected to the considered node at level $i$. Each non-dashed (resp. dashed) rectangle on some node/resource represents a player selecting that resource in $\sg_{s,n}$ (resp. $\sg^*_{s,n}$). The congestion of each resource in $\sg_{s,n}$ or $\sg^*_{s,n}$ is represented accordingly. Let $K$ be the congestion of the resources at level $i$ in $\sg^*_{s,n}$. Then $K\cdot k_j$ is the congestion of the resources at level $i$ in $\sg_{s,n}$, $K\cdot k_j/n$ is the congestion of the resources at level $i+1$ in $\sg^*_{s,n}$, and $K\cdot k_j\cdot  k_h/n$ is the congestion of the resources at level $i+1$ in $\sg_{s,n}$, where $j=h=1$ if $i\in [s-1]$,   $j=1$ and $h=2$ if $i=s$, $j=h=2$ otherwise.}\label{fig:56}
\end{figure}
Let $\bm\sigma_{s,n}$ be the strategy profile in which all players select their first strategy. We have the following lemma:
\begin{lemma}\label{equi_claim}
For any integer $s\geq 2$, there exists $n(s)\in \NN$, such that $\bm\sigma_{s,n}$ is an $\epsilon$-approximate pure Nash equilibrium of $\LB_{s,n}(k_1,k_2,f_1,f_2)$ for any $n\geq n(s)$.
\end{lemma}
The proof of the above lemma is deferred to Section \ref{missing} of the appendix, and we only give a sketch of the proof. We consider an arbitrary player whose first strategy is a resource from a generic level $i$. Since the game is symmetric, we have to check deviations to all possible strategies. We have that: (i) if $i\in [2s-1]$ and the player deviates in favour of a resource from level $j=i+1$, her cost decreases exactly of a factor $1+\epsilon$; (ii) if $i\in [2s]$ and the player deviates in favour of a resource from level $j\leq i$, her cost does not decrease; (iii) if $i\in [2s-2]$ and the player deviates in favour of a resource from level $j>i+1$, for any sufficiently large $n$, her cost does not decrease. In any case, the cost of the considered player, when she deviates, cannot decrease of a factor higher than $1+\epsilon$,  thus $\sg_{s,n}$ is a pure Nash equilibrium. In the appendix, we prove separately all the three cases.

In the remainder of the proof, with the aim of estimating the approximate price of anarchy of $\LB_{s,n}(k_1,k_2,f_1,f_2)$, we compare the social value of the $\epsilon$-pure Nash equilibrium $\bm\sigma_{s,n}$, with that of the strategy profile $\bm\sigma^*_{s,n}$ in which all players select their second strategy. By exploiting the definition of $\alpha_1,\alpha_2$ given in Lemma \ref{lemrou1}, we necessarily get $k_1\theta_{1}>1$ (from $\alpha_2>0$) and $k_2\theta_{2}<1$ (from $\alpha_1>0$). We have that
\begin{align}
&{\sf SUM}(\bm\sigma_{s,n})\nonumber\\
&=\sum_{i=1}^{2s}n^{i-1}w_i\cdot n \cdot g_i(n\cdot w_i)\nonumber\\
&=\sum_{i=1}^s n^{i}\left(\frac{k_1}{n}\right)^ig_i\left(n\left(\frac{k_1}{n}\right)^i\right)\nonumber\\
&\ \ \ \ \ \ \ +\sum_{i=s+1}^{2s}n^{i}\left(\frac{k_1}{n}\right)^s\left(\frac{k_2}{n}\right)^{i-s}g_i\left(n\left(\frac{k_1}{n}\right)^s\left(\frac{k_2}{n}\right)^{i-s}\right)\nonumber\\
&=\sum_{i=1}^sk_1^i\theta_{1}^{i-1}f_1(k_1)+\sum_{i=s+1}^{2s}k_1^sk_2^{i-s}\theta_{1}^{s-1}\theta_{1,2}\theta_{2}^{i-s-1}f_2(k_2)\nonumber\\
&=\sum_{i=1}^s(k_1\theta_{1})^{i-1}k_1f_1(k_1)+\sum_{i=s+1}^{2s}(k_1\theta_{1})^{s-1}k_1\theta_{1,2}(k_2\theta_{2})^{i-s-1}k_2f_2(k_2)\nonumber\\
&=\sum_{i=0}^{s-1}(k_1\theta_{1})^{i}k_1f_1(k_1)+(k_1\theta_{1})^{s-1}k_1\theta_{1,2}\sum_{i=0}^{s-1}(k_2\theta_{2})^{i}k_2f_2(k_2)\label{conteq}\\
&=\left(\frac{(k_1\theta_{1})^s-1}{k_1\theta_{1}-1}\right)k_1f_1(k_1)+(k_1\theta_{1})^{s-1}k_1\theta_{1,2}\left(\frac{1-(k_2\theta_{2})^s}{1-k_2\theta_{2}}\right)k_2f_2(k_2).\label{equeq}
\end{align}
Now, observe that in $\bm\sigma^*_{s,n}$ there are no players at level 1, there is exactly one player of weight $(k_1/n)^{i-1}$ (resp. $(k_1/n)^{s}(k_2/n)^{i-s-1})$)  on each resource at level $i\in [s+1]_2$ (resp. $i\in [2s-1]_{s+2}$),  and there are $n+1$ players on each resource at level $2s$. We get\footnote{The definition of $[<]_s(F(s))$ has been given in the preliminaries, and substitutes the usual asymptotic notation.}
\begin{align}
&{\sf SUM}(\bm\sigma^*_{s,n})\nonumber\\
&=\sum_{i=2}^{2s-1}n^{i-1}w_{i-1}\cdot g_i(w_{i-1})+n^{2s-1}(n\cdot w_{2s}+w_{2s-1})g_{2s}(n\cdot w_{2s}+w_{2s-1})\nonumber\\
&=\sum_{i=2}^{s+1}\left(n^{i-1}\left(\frac{k_1}{n}  \right)^{i-1}\right)g_i\left(\left(\frac{k_1}{n}\right)^{i-1}\right)\nonumber\\
&\ \ \ \ \ \ \ +\sum_{i=s+2}^{2s-1}n^{i-1}\left(\frac{k_1}{n}\right)^s\left(\frac{k_2}{n}\right)^{i-s-1}g_i\left(\left(\frac{k_1}{n}\right)^s\left(\frac{k_2}{n}\right)^{i-s-1}\right)\nonumber\\
&\ \ \ \ \ \ \ +n^{2s-1}(n\cdot w_{2s}+w_{2s-1})g_{2s}(n\cdot w_{2s}+w_{2s-1})\nonumber\\
&=\sum_{i=2}^{s}k_1^{i-1}\theta_{1}^{i-1}f_1(1)+\sum_{i=s+1}^{2s-1}k_1^sk_2^{i-s-1}\theta_{1}^{s-1}\theta_{1,2}\theta_{2}^{i-s-1}f_2(1)\nonumber\\
&\ \ \ \ \ \ \ +k_1^sk_2^{s-1}(k_2+1)\theta_{1}^{s-1}\theta_{1,2}\theta_{2}^{s-1}f_2(k_2+1)\label{contopt}\\
&=\sum_{i=1}^{s}k_1^{i-1}\theta_{1}^{i-1}f_1(1)+\sum_{i=s+1}^{2s-1}k_1^sk_2^{i-s-1}\theta_{1}^{s-1}\theta_{1,2}\theta_{2}^{i-s-1}f_2(1)\nonumber\\
&\ \ \ \ \ \ \ \overbrace{-f_1(1)+k_1^sk_2^{s-1}(k_2+1)\theta_{1}^{s-1}\theta_{1,2}\theta_{2}^{s-1}f_2(k_2+1)}^{[<]_s\left((k_1\theta_{1})^{s-1}\right)}\nonumber\\
&=\sum_{i=0}^{s-1}(k_1\theta_{1})^{i}f_1(1)+(k_1\theta_{1})^{s-1}k_1\theta_{1,2}\sum_{i=0}^{s-1}(k_2\theta_{2})^{i}f_2(1)+[<]_s\left((k_1\theta_{1})^{s-1}\right)\nonumber\\
&=\frac{(k_1\theta_{1})^s-1}{k_1\theta_{1}-1}f_1(1)+(k_1\theta_{1})^{s-1}k_1\theta_{1,2}\frac{1-(k_2\theta_{2})^s}{1-k_2\theta_{2}}f_2(1)
+[<]_s\left((k_1\theta_{1})^{s-1}\right).\label{opteq}
\end{align}
By using (\ref{equeq}) and (\ref{opteq}), inequalities $k_1\theta_{1}>1$ and $k_2\theta_{2}<1$, and the fact that $\sg_{s,n}$ is a pure Nash equilibrium of $\LB_{s,n}(k_1,k_2,f_1,f_2)$ for any $s$ and any sufficiently large $n$ (by Lemma \ref{equi_claim}), we get
\begin{align}
&\sup_{s,n\in \NN}\poa_\epsilon(\LB_{s,n}(k_1,k_2,f_1,f_2))\nonumber\\
&\geq \lim_{s\rightarrow \infty}\lim_{n\rightarrow \infty}\frac{{\sf SUM}(\bm\sigma_{s,n})}{{\sf SUM}(\bm\sigma^*_{s,n})}\nonumber\\
&=\lim_{s\rightarrow \infty}\lim_{n\rightarrow \infty}\frac{{\sf SUM}(\bm\sigma_{s,n}) /(k_1\theta_{1})^{s-1}}{{\sf SUM}(\bm\sigma^*_{s,n})/(k_1\theta_{1})^{s-1}}\nonumber\\
&=\frac{\frac{k_1\theta_{1}}{k_1\theta_{1}-1}k_1f_1(k_1)+k_1\theta_{1,2}\left(\frac{1}{1-k_2\theta_{2}}\right)k_2f_2(k_2)}
{\frac{k_1\theta_{1}}{k_1\theta_{1}-1}f_1(1)+k_1\theta_{1,2}\left(\frac{1}{1-k_2\theta_{2}}\right)f_2(1)}\nonumber\\
&=\frac{\left(\frac{k_1f_1(k_1)}{k_1f_1(k_1)-(1+\epsilon)f_1(k_1+1)}\right)k_1f_1(k_1)+\left(\frac{k_1f_1(k_1)}{(1+\epsilon)f_2(k_2+1)-k_2f_2(k_2)}\right)k_2f_2(k_2)}{\left(\frac{k_1f_1(k_1)}{k_1f_1(k_1)-(1+\epsilon)f_1(k_1+1)}\right)f_1(1)+\left(\frac{k_1f_1(k_1)}{(1+\epsilon)f_2(k_2+1)-k_2f_2(k_2)}\right)f_2(1)}\nonumber\\
&=\frac{\left(\frac{k_1f_1(k_1)}{\alpha_2}\right)k_1f_1(k_1)+\left(\frac{k_1f_1(k_1)}{\alpha_1}\right)k_2f_2(k_2)}{\left(\frac{k_1f_1(k_1)}{\alpha_2}\right)f_1(1)+\left(\frac{k_1f_1(k_1)}{\alpha_1}\right)f_2(1)}\nonumber\\
&=\frac{\alpha_1k_1f_1(k_1)+\alpha_2 k_2f_2(k_2)}{\alpha_1f_1(1)+\alpha_2f_2(1)}\nonumber\\
&>M,\label{finale}
\end{align}
where (\ref{finale}) comes from Lemma~\ref{lemrou1}. Thus, by \eqref{finale}, we can choose $s$ and $n$ in such a way that $\poa_\epsilon(\LB_{s,n}(k_1,k_2,f_1,f_2))>M$, and this shows the claim when Case 2 of Lemma \ref{lemrou1} holds. 

Now, suppose that Case 1 of Lemma~\ref{lemrou1} holds, and let $k,o,f$ be the parameters considered in that case; as done previously, by Remark \ref{remalem}, we can assume without loss of generality that $o=1$. The load balancing instance we consider here is $\LB_{s,n}(k,k,f,f)$, and the strategy profile $\sg_{s,n}$ and $\sg^*_{s,n}$ are defined as in the previous part of the proof. To evaluate ${\sf SUM}({\bm\sigma}_{s,n})$, observe that all equalities up to (\ref{conteq}) hold. Therefore, by continuing from (\ref{conteq}), we get 
\begin{equation*}
{\sf SUM}({\bm\sigma}_{s,n})=2skf(k)
\end{equation*}
 if $\beta_{\sf W}(\poa_\epsilon,k,o,f)=0$ (since, in such case, we have $k\theta_{1}=k\theta_2=1$) and
\begin{equation*}
{\sf SUM}(\bm\sigma_{s,n})=\frac{1-\left(\frac{kf(k)}{(1+\epsilon)f(k+1)}\right)^{2s}}{1-\frac{kf(k)}{(1+\epsilon)f(k+1)}}kf(k)
\end{equation*}
 if $\beta_{\sf W}(\poa_\epsilon,k,o,f)>0$ (since now we have $k\theta_{1}=k\theta_2<1$). Analogously, to evaluate ${\sf SUM}(\bm\sigma^*_{s,n})$, observe that all equalities up to (\ref{contopt}) hold. Therefore, by continuing from (\ref{contopt}) we get $${\sf SUM}(\bm\sigma^*_{s,n})=2sf(1)+[\leq ]_s(1)$$ if $\beta_{\sf W}(\poa_\epsilon,k,o,f)=0$, and
\begin{equation*}
{\sf SUM}(\bm\sigma^*_{s,n})=\frac{1-\left(\frac{kf(k)}{(1+\epsilon)f(k+1)}\right)^{2s}}{1-\frac{kf(k)}{(1+\epsilon)f(k+1)}}f(1)-f(1)+[<]_s(1)
\end{equation*}
if $\beta_{\sf W}(\poa_\epsilon,k,o,f)>0$. Thus, again, we get $$
\sup_{s,n\in \NN}\poa_\epsilon(\LB_{s,n}(k,k,f,f))\geq \lim_{s\rightarrow \infty}\lim_{n\rightarrow \infty}\frac{{\sf SUM}(\bm\sigma_{s,n})}{{\sf SUM}(\bm\sigma^*_{s,n})}\geq \frac{kf(k)}{f(1)}>M,$$ where the last inequality follows from Case~1 of Lemma~\ref{lemrou1}. 

By exploiting the above inequality, we can choose $s$ and $n$ in such a way that $\poa_\epsilon(\LB_{s,n}(k,k,f,f))>M$, and this shows the claim when Case 1 of Lemma~\ref{lemrou1} holds. We conclude that the approximate price of anarchy does not improve when restricting to symmetric load balancing games, if all non-constant latency functions are unbounded. 

If there is some non-constant latency function in $\mathcal{C}(\mathcal{G})$ that is not unbounded, we can consider a general (asymmetric) load balancing game based on the load balancing graph $\LB_{s,n}(k_1,k_2,f_1,f_2)$ defined above with $n=1$ (i.e., the load balancing graph is a directed path), but where each player can select among her first and second strategy only, and based on the parameters $k_1,k_2,f_1,f_2$ or $k,f$ derived according to Lemma \ref{lemrou1}. Let $\sg_{s,1}$ and $\sg_{s,1}^*$ be the strategy profiles defined above (with $n=1$). By using a similar proof as in Lemma~\ref{equi_claim}, we have that $\sg_{s,1}$ is an $\epsilon$-approximate pure Nash equilibrium. Indeed, the unique deviation that must be analysed is when each player located at level $i$ in $\sg_{s,1}$ deviates in favour of a resource at level $i+1$ (case (i) in the proof of the lemma), and showing that such deviation does not give any benefit to that player does not require the unboundedness of the latency functions. Finally, we can reuse the last part of the above proof to show that, for sufficiently large $s$, $\poa_\epsilon(\LB_{s,1}(k_1,k_2,f_1,f_2))>M$. Thus, by the arbitrariness of $M$, we have $\gamma_{\sf W}(\poa_{\epsilon},{\sf W}(\mathcal{C}))=\poa_{\epsilon}({\sf W}(\mathcal{C}))=\poa_{\epsilon}({\sf WLB}(\mathcal{C}))$ which shows the claim.\qed
\end{proof}
Now, we prove that no improvements are possible for approximate one-round walks when restricting to load balancing games.

\begin{theorem}\label{thm2}
Let $\mathcal{C}$ be a class of latency functions that is closed under ordinate and abscissa scaling. Then, $\gamma_{\sf W}({\sf CR}^s_{\epsilon},{\sf W}(\mathcal{C}))={\sf CR}^s_{\epsilon}({\sf W}(\mathcal{C}))={\sf CR}^s_{\epsilon}({\sf WLB}(\mathcal{C}))$. If all functions in $\mathcal{C}$ are also semi-convex, we have that $\gamma_{\sf W}({\sf CR}^c_{\epsilon},{\sf W}(\mathcal{C}))={\sf CR}^c_{\epsilon}({\sf W}(\mathcal{C}))={\sf CR}^c_{\epsilon}({\sf WLB}(\mathcal{C})).$
\end{theorem}
\begin{proof}
Fix $M<\gamma_{\sf W}({\sf EM},{\sf W}(\mathcal{C}))$, where ${\sf EM}\in \{{\sf CR}^s_{\epsilon},{\sf CR}^c_{\epsilon}\}$ is the considered efficiency metric. As in Theorem \ref{thm1}, by Lemma \ref{lemupp}, it is sufficient showing that there exists a load balancing game $\LB\in {\sf WLB}(\mathcal{C})$ such its competitive ratio is at least $M$. 

We start with the case of selfish players. Let $k_1,k_2,f_1,f_2,\alpha_1,\alpha_2$ (resp. $k,o,f$) be the parameters considered in Case 2 (resp. Case 1) of  Lemma~\ref{lemrou1}, but related to ${\EM}={\sf CR}_{\epsilon}^s$. For simplicity, in the following we set $k_1=k_2=k$, $o_1=o_2=o$, and $f_1=f_2=f$, if Case 1 of Lemma \ref{lemrou1} holds. By Remark \ref{remalem}, we can assume without loss of generality that $o_1=o_2=1$.  

As in Theorem \ref{thm1}, we resort to a load-graph representation. In particular, we extend the load balancing graph $\LB_{s,n}(k_1,k_2,f_1,f_2)$ defined in the proof of Theorem \ref{thm1} as follows. Denote as $i(v)$ the level of resource $v$. For each node $u$ in the load balancing graph, consider an arbitrary enumeration of all the $n$ outgoing edges of $u$. Since each node has a unique incoming edge, we denote by $h(v)\in [n]$ the label/position associated to the unique edge entering $v$ in the given ordering, so that, for any node $u$, all the $n$ edges of type $(u,v)$ are associated to a label $h(v)\in [n]$. 

For $i,j\in [2]$ and $h\in [n]$, define $$\theta_{i,j}(h):=\frac{f_i\left(\frac{h k_i}{n}\right)}{(1+\epsilon)f_j(k_j+1)}$$ and $\theta_{i}(h):=\theta_{i,i}(h)$. Resource $v$ has latency function
$$g_v(x):=\begin{cases}
f_1(x) & \text{ if }i(v)=1,\\
\underbrace{\theta_1(h(v)) A_u}_{A_v} f_1\left(\left(\frac{n}{k_1}\right)^{i(v)-1}x\right)&\text{ if }i(v)\in [s]_2,\\
\underbrace{\theta_{1,2}(h(v)) A_u}_{A_v}  f_2\left(\left(\frac{n}{k_1}\right)^{s}x\right) & \text{ if }i(v)=s+1,\\
\underbrace{\theta_2(h(v)) A_u}_{A_v}  f_2\left(\left(\frac{n}{k_1}\right)^{s}\left(\frac{n}{k_2}\right)^{i(v)-s-1}x\right)&  \text{ if }i(v)\in [2s]_{s+2},
\end{cases}
$$
where $(u,v)$ denotes the unique incoming edge of $v$ and $A_v$ is recursively defined on the basis of $A_u$ by setting $A_v=1$ for $i(v)=1$, i.e., for $v$ being the root of the tree. The weights of all players are defined as in Theorem~\ref{thm1}. 

Consider the online process $\bm \tau_{s,n}$ in which players enter the game in non-increasing order of level (with respect to their first strategy) and, within the same level, players are processed in non-decreasing order of position; equivalently, if two players $i_u$ and $i_v$ have, as their first strategies, some resources $u$ and $v$ respectively, such that either $i(u)>i(v)$, or $i(u)=i(v)\wedge h(u)<h(v)$, then player $i_u$ is processed before player $i_v$ in the online process $\bm \tau_{s,n}$. Let ${\bm\sigma}_{s,n}$ be strategy profile obtained at the end of the process, i.e., the strategy profile in which each player selects her first strategy. We have the following lemma:
\begin{lemma}\label{oneround_claim}
The online process $\bm \tau_{s,n}$ is an $\epsilon$-approximate one-round walk.
\end{lemma}

The above lemma can be shown as follows. Let $z$ be an arbitrary player whose first strategy is at level $i\in [2s-1]$; let $u,v$ be the first and second strategy of $z$, respectively, and let $h:=h(v)$ be the position/label associated to resource $v$. By construction of the online process $\bm \tau_{s,n}$, we have that, when $z$ enters the game, there are $h-1$ players of weight $w_i$ already assigned to resource $u$, and $n$ players of weight $w_{i+1}$ assigned to resource $v$. Let $\sg^u_{s,n}$ be the partial strategy profile obtained when $z$ is processed according to $\bm\tau_{s,n}$ (i.e., only players preceding $z$ have been already assigned), and $z$ is assigned to her first strategy $u$; let $\sg^{v}_{s,n}$ be the partial strategy profile in which $z$, instead, is assigned to her second strategy $v$. As our choice of player $z$ has been arbitrary, it is sufficient proving that $cost_z(\sg^u_{s,n})\leq (1+\epsilon)cost_z(\sg^v_{s,n})$ to show that $\bm\tau_{s,n}$ is an $\epsilon$-approximate one-round walk. By using the recursive definitions of the latency functions $g_i$ and $g_{i+1}$ (involving the quantities $A_u$ and $A_v$, respectively), and by using a similar approach as in case (i) of Lemma \ref{equi_claim}, one can easily show that $cost_z(\sg^u_{s,n})=g_u(h\cdot w_i)=(1+\epsilon)g_v(n\cdot w_{i+1}+w_i)=(1+\epsilon)cost_z(\sg^v_{s,n})$, thus showing the claim of Lemma \ref{oneround_claim}. 

Now, let ${\bm \sigma}^*_{s,n}$ be the strategy profile in which all players select their second strategy. In the remainder of the proof, with the aim of estimating the competitive ratio of $\LB_{s,n}(k_1,k_2,f_1,f_2)$, we compare the social value of the strategy profile $\sg_{s,n}$ resulting from the $\epsilon$-approximate one-round walk $\bm \tau_{s,n}$ with that of strategy profile $\bm\sigma^*_{s,n}$. For $i,j\in [2]$, let $\xi_{i,j}$ denote quantity $\frac{\int_{t=0}^k f_i(t)\text{dt}}{(1+\epsilon)f_j(k_j+1)}$, let $\xi_i:=\xi_{i,i}$, and let ${\sf SUM}_i({\bm\sigma_{s,n}})$ denote the weighted total latency in $\sg_{s,n}$ of all the resources at level $i\in [2s]$. We have that 
\begin{align}
{\sf SUM}_i({\bm\sigma_{s,n}})
&=\sum_{v\text{ in level }i}k_v(\sg_{s,n})g_{v}\left(k_v(\bm\sigma_{s,n})\right)\nonumber\\
&=\sum_{(h_1,h_2,\ldots, h_{i-1})\in [n]^{i-1}}n\cdot \overbrace{\left(\frac{k_1}{n}\right)^i}^{w_i}\left(\prod_{t=1}^{i-1}\theta_1(h_t)\right)f_1(k_1)\nonumber\\
&=\frac{k_1^{i}}{n^{i-1}}\left(\sum_{h=1}^n\theta_1(h)\right)^{i-1}f_1(k_1)\nonumber\\
&=\left(\sum_{h=1}^n\theta_1(h)\frac{k_1}{n}\right)^{i-1}k_1f_1(k_1)\nonumber\\
&=\left(\sum_{h=1}^n\frac{f_i\left(\frac{h k_1}{n}\right)}{(1+\epsilon)f_j(k_1+1)}\left(\frac{k_1}{n}\right)\right)^{i-1}k_1f_1(k_1)\nonumber\\
&\overrightarrow{n\rightarrow \infty} \left(\frac{\int_{t=0}^k f_1(t)\text{dt}}{(1+\epsilon)f_1(k_1+1)}\right)^{i-1}k_1f_1(k_1)\nonumber\\
&=\xi_1^{i-1}k_1f_1(k_1)\label{form1}
\end{align}
if $i\in [s]$, and (by using similar arguments)
\begin{align}
&{\sf SUM}_i({\bm\sigma_{s,n}})\nonumber\\
&=\left(\sum_{h=1}^n\theta_1(h)\frac{k_1}{n}\right)^{s-1}\left(\sum_{h=1}^n\theta_{1,2}(h)\frac{k_1}{n}\right)\left(\sum_{h=1}^n\theta_{2}(h)\frac{k_2}{n}\right)^{i-1-s}k_2f_2(k_2)\nonumber\\
&\overrightarrow{n\rightarrow \infty}\ \ \ \xi_1^{s-1}\cdot \xi_{1,2}\cdot \xi_2^{i-1-s}k_2f_2(k_2),\label{form2}
\end{align}
if $i\in [2s]_{s+1}$. 
Now, for any fixed $s\geq 2$, we can compute $\lim_{n\rightarrow\infty}{\sf SUM}({\bm\sigma_{s,n}})$. If Case 2 of Lemma \ref{lemrou1} holds, we have that $\xi_1>0$ and $\xi_2<0$. Thus, by using similar arguments as in Theorem \ref{thm1} (in particular, as in equality \eqref{equeq}), and by using \eqref{form1} and \eqref{form2}, we get 
\begin{align*}
\lim_{n\rightarrow\infty}{\sf SUM}({\bm\sigma_{s,n}})
&=\lim_{n\rightarrow\infty}\sum_{i=1}^{2s}{\sf SUM}_i({\bm\sigma_{s,n}})\\
&=\sum_{i=0}^{s-1}\xi_1^{i}k_1f_1(k_1)+\xi_1^{s-1}\cdot \xi_{1,2}\sum_{i=0}^{s-1}\xi_2^{i}k_2f_2(k_2)\\
&=\left(\frac{\xi_1^s-1}{\xi_1-1}\right)k_1f_1(k_1)+\xi_1^{s-1}\xi_{1,2}\left(\frac{1-\xi_2^s}{1-\xi_2}\right)k_2f_2(k_2);
\end{align*}
analogously, if Case 1 of Lemma \ref{lemrou1} holds, we have that $\lim_{n\rightarrow\infty}{\sf SUM}({\bm\sigma_{s,n}})=2s kf(k)$ if $\beta_{\sf W}({\sf CR}_\epsilon^s,k,o,f)=0$ (as $\xi_1=\xi_2=1$), and $\lim_{n\rightarrow\infty}{\sf SUM}({\bm\sigma_{s,n}})=\frac{1-\xi_1^s}{1-\xi_1}kf(k)$ if $\beta_{\sf W}({\sf CR}_\epsilon^s,k,o,f)>0$ (as $\xi_1=\xi_2<1$). 

By using analogous steps as in (\ref{form1}) and (\ref{form2}), one can show that the weighted total latency ${\sf SUM}_i(\sg_{s,n}^*)$ in $\sg^*_{s,n}$ of resources at level $i\in [2s-1]_2$ can be obtained by replacing each $k_jf_j(k_j)$ with $f_j(1)$ in (\ref{form1}) and (\ref{form2}) and, analogously to Theorem~\ref{thm1}, one can also compute ${\sf SUM}(\sg_{s,n}^*)$, that is equal to $\sum_{i=2}^{2s-1}{\sf SUM}_i(\sg_{s,n}^*)+[<]_s\left(\sum_{i=2}^{2s-1}{\sf SUM}_i(\sg_{s,n}^*)\right)$ (as in Theorem~\ref{thm1}, the contribution of resources at level $1$ and $2s$ to the social cost is not significant). At this point, by using the same proof arguments of Theorem~\ref{thm1} (in particular, as in \eqref{finale}), we get  
\begin{align*}
\lim_{s\rightarrow \infty}\lim_{n\rightarrow \infty}{\sf CR}_\epsilon^s(\LB_{s,n}(k_1,k_2,f_1,f_2))
&\geq \lim_{s\rightarrow \infty}\lim_{n\rightarrow \infty}\frac{{\sf SUM}({\bm\sigma_{s,n}})}{{\sf SUM}(\sg_{s,n}^*)}\\
&\geq \frac{\alpha_1k_1f_1(k_1)+\alpha_2k_2f_2(k_2)}{\alpha_1f_1(1)+\alpha_2f_2(1)}\\
&>M,
\end{align*}
where the last inequality comes from Lemma \ref{lemrou1}. Thus, the claim holds for the case of selfish players. 

For cooperative players, let $k_1,k_2,o_1,o_2,f_1,f_2,\alpha_1,\alpha_2$ as in Lemma~\ref{lemrou1} (remove index $j\in [2]$ and set $\alpha=1$ if Case 1  of Lemma~\ref{lemrou1} is verified). By Remark~\ref{remalem}, we assume without loss of generality that $o_1=o_2=1$. Consider a load balancing graph $\LB_{s,n}(k_1,k_2,f_1,f_2)$ as that defined above, but with $n=1$ and $\theta_{i,j}(1):=\frac{f_i(k_i)}{(1+\epsilon)((k_j+1)f_j(k_j+1)-k_jf_j(k_j))}$. Let $\sg_{s,1}$, $\bm \tau_{s,1}$, and $\sg^*_{s,1}$ be defined as in the case of selfish players (with $n=1$). Similarly as in the case of selfish players, one can show that $\bm \tau_{s,1}$ is an $\epsilon$-approximate one-round walk (involving cooperative players) that generates strategy profile $\sg_{s,1}$, and we can show again that 
\begin{align*}
\lim_{s\rightarrow \infty}\lim_{n\rightarrow \infty}{\sf CR}_\epsilon^c(\LB_{s,n}(k_1,k_2,f_1,f_2))
&\geq \lim_{s\rightarrow \infty}\lim_{n\rightarrow \infty}\frac{{\sf SUM}({\bm\sigma_{s,n}})}{{\sf SUM}(\sg_{s,n}^*)}\\
&\geq \frac{\alpha_1k_1f_1(k_1)+\alpha_2k_2f_2(k_2)}{\alpha_1f_1(1)+\alpha_2f_2(1)}\\
&>M;
\end{align*}
this shows that ${\sf CR}_\epsilon^c(\LB_{s,1}(k_1,k_2,f_1,f_2))>M$ for a sufficiently large $s$, and the claim follows. \qed
\end{proof}
\begin{remark}\label{remalowweig}
Fix a metric ${\sf EM}\in\{{\sf PoA}_\epsilon,{\sf CR}^s_{\epsilon},{\sf CR}^c_{\epsilon}\}$, a class of latency functions $\mathcal{C}$, and let $(k_1,k_2,o_1,o_2,f_1,f_2)$ (resp. $(k,o,f)$) be a tuple such that values $\alpha_1,\alpha_2$ considered in Case~2 of Lemma~\ref{lemrou1} are positive (resp. such that the inequality of Case 1 corresponding to the considered metric {\sf EM} is satisfied, i.e., $\beta_{\sf W}({\sf EM},k,o,f)\geq 0$). By Remark \ref{remalem}, we assume without loss of generality that $o_1=o_2=1$ (resp. $o=1$). By inspecting the proofs of Theorem~\ref{thm1} and \ref{thm2}, we observe that the set of load balancing instances of type $\LB_{s,n}(k_1,k_2,f_1,f_2)$ (resp. $\LB_{s,n}(k,k,f,f)$) guarantees a performance of at least $\frac{\alpha_1k_1 f_1(k_1)+\alpha_2 k_2 f_2(k_2)}{\alpha_1f(1)+\alpha_2f(2)}$ (resp. $\frac{k f(k)}{f(1)}$) under the considered metric ${\sf EM}$. Thus, in the case we are not able to quantify the exact value of $\gamma_{\sf W}({\sf EM}, {\sf W}(\mathcal{C}))$, we can still find a tuple $(k_1,k_2,f_1,f_2)$ (resp. $(k,f)$) that leads to good (and possibly tight) lower bounds; furthermore, the tightness can be shown by resorting to the characterization provided in Remark \ref{rema_lemrou}.
\end{remark}

\subsection{Application to Polynomial Latency Functions}\label{wplf}
Consider the class $\mathcal{P}(d)$ of polynomials with non-negative coefficients and maximum degree $d$.  By Lemma~\ref{lemupp}, $\gamma_{\sf W}({\sf EM},{\sf W}(\mathcal{P}(d)))$ is an upper bound on the efficiency metrics ${\sf EM}\in\{{\sf PoA}_\epsilon,{\sf CR}^s_{\epsilon},{\sf CR}^c_{\epsilon}\}$ for weighted congestion games and load balancing games with polynomial latency functions of maximum degree $d$. 

By exploiting the definition of $\gamma_{\sf W}({\sf EM},{\sf W}(\mathcal{P}(d)))$ given at the beginning of this section, and by applying to these definitions similar arguments as those used in \cite{ADGMS11,CMS12} we can compute the exact value of $\gamma_{\sf W}({\sf EM},{\sf W}(\mathcal{P}(d)))$ for any ${\sf EM}\in\{{\sf PoA}_\epsilon,{\sf CR}^s_{\epsilon},{\sf CR}^c_{\epsilon}\}$\footnote{In \cite{ADGMS11,CMS12}, the upper bounds on the performance of congestion games have been represented in terms of smoothness inequalities (see Section \ref{frame} of the appendix), but can be easily rewritten in terms of the upper bound $\gamma_{\sf W}({\sf EM},{\sf W}(\mathcal{P}(d)))$. Thus, since the exact quantification of $\gamma_{\sf W}({\sf EM},{\sf W}(\mathcal{P}(d)))$ relies on results and/or proof arguments already appeared, we omit the mathematical analysis and we give the final numerical values. Furthermore, we point out that the main focus of this subsection is not computing $\gamma_{\sf W}({\sf EM},{\sf W}(\mathcal{P}(d)))$, but applying Theorem \ref{thm1} and \ref{thm2} to provide tight bounds on the performance of load balancing games with polynomial latency functions. }. In particular, we have that $\gamma_{\sf W}({\sf EM},{\sf W}(\mathcal{P}(d)))$ is equal to $\left(\varphi_{\epsilon,d}\right)^{d}$, where $\varphi_{\epsilon,d}$ is defined as the unique solution $k\geq 0$ of equation $-k^{d+1}+(1+\epsilon)(k+1)^d=0$ if ${\sf EM}=\poa_\epsilon$, of equation $-\frac{k^{d+1}}{d+1}+(1+\epsilon)(k+1)^d=0$ if ${\sf EM}={\sf CR}^s_{\epsilon}$, and of equation $-(2+\epsilon)k^{d+1}+(1+\epsilon)(k+1)^{d+1}=0$ if ${\sf EM}={\sf CR}^c_{\epsilon}$.

Some of the above upper bounds have already appeared in works based on the study of congestion games. $\gamma_{\sf W}({\sf PoA}_{\epsilon},{\sf W}(\mathcal{P}(d)))$ has been already evaluated in \cite{ADGMS11,CKS11}. The bounds related to the $\epsilon$-approximate one-walks generated by  selfish players generalize the results obtained in \cite{CMS12}, in which the same bounds have been shown for affine latency functions and $\epsilon=0$, only; for more general polynomial latency functions and $\epsilon=0$, the bounds related to the $\epsilon$-approximate one-walks have been re-obtained and shown in more detail in \cite{Klimm19}, subsequently to the preliminary version of our work. 

Finally, for $\epsilon=0$ and general polynomial latencies, the bounds related to cooperative players can be equivalently derived from the analysis provided in \cite{C08}, in which the greedy algorithm has been applied to minimize the $L_p$ norm in load balancing problems. 

A way to interpret the above tight bounds is provided by Remark~\ref{rema_lemrou}. In particular, one can show that the exact value of $\gamma_{\sf W}({\sf EM},{\sf W}(\mathcal{P}(d)))$ is attained by $\frac{k f(k)}{f(1)}$, where $f$ is the monomial function defined as $f(t)=t^d$ for any $t\geq 0$, and $k$ is the unique value satisfying constraint $\beta_{\sf W}({\sf EM},k,1,f)\geq 0$ at equality.

Since the class of polynomial latency functions $\mathcal{P}(d)$ satisfies the hypothesis of Theorem~\ref{thm1} and \ref{thm2}, as a corollary we have that the values $\gamma_{\sf W}({\sf EM},{\sf W}(\mathcal{P}(d)))$ considered above are tight bounds on the performance of load balancing instances, and even of symmetric load balancing instances if the efficiency metric is the $\epsilon$-price of anarchy. Thus, for polynomial latency functions, the performance does not improve when assuming singleton strategies. Furthermore, our findings can be used to close the gap between upper and lower bounds on the competitive ratio of exact one-round walks generated by selfish players, for congestion games with affine latency functions. Indeed,  \cite{CMS12} showed that $\gamma_{\sf W}({\sf CR}^s_{0},{\sf W}(\mathcal{P}(1)))=2\sqrt{3}+4\approx 7.464$ is an upper bound, and \cite{BFFM09} provided a lower bound of $2+\sqrt{5}\approx 4.236$ (holding even for unweighted congestion games); closing the above gap has been left as an open question. Our results directly imply that the upper bound $\approx 7.464$ provided by \cite{CMS12} is tight, even for load balancing games. 
See Table \ref{figurab1} for some numerical comparisons. 
\begin{figure}[!t]
\begin{center}
\begin{tabular}{|c|c|c|c|}
  \hline
  $d$ & ${\sf PoA}_0({\sf W}(\mathcal{P}(d)))$ \cite{ADGMS11,BGR10}  & ${\sf CR}_0^s({\sf W}(\mathcal{P}(d)))$ & ${\sf CR}_0^c({\sf W}(\mathcal{P}(d)))$ \cite{C08} \\\hline
  1 & 2.618 \cite{AAE05,CFKKM11} & 7.464 \cite{CMS12} & 5.828 \cite{CFKKM11}  \\
  2 & 9.909 &  90.3 & 56.94 \\
  3 & 47.82 & 1,521 & 780.2 \\
  4 & 277 & 32,896 & 13,755 \\
  5 & 1,858 & 868,567 & 296,476\\
  6 & 14,099  & 27,089,557 & 7,553,550 \\
  7 & 118,926  & 974,588,649 & 222,082,591\\
  8 & 1,101,126  & 39,729,739,895 & 7,400,694,480\\
  $\vdots$ &  $\vdots$ &  $\vdots$ &  $\vdots$ \\
  $\infty$ & $(\Theta(d/\log(d)))^{d+1}$  & $(\Theta(d))^{d+1}$ & $(\Theta(d))^{d+1}$\\
  \hline
\end{tabular}

\caption{The price of anarchy and the competitive ratio of exact one-round walks (generated by selfish or cooperative players) in weighted congestion games or load balancing games with polynomial latency functions of maximum degree $d$. Our findings provide tight lower bounds for the competitive ratio under selfish players, holding even for load balancing games.}\label{figurab1}
\end{center}
\end{figure}

\section{Unweighted Load Balancing Games}\label{secunweigh}
In this section, we show that, under mild assumptions on the latency functions, the $\epsilon$-approximate price of anarchy and the competitive ratio of selfish or cooperative $\epsilon$-approximate one-round walks in unweighted congestion games cannot improve even when restricting to load balancing games. 
\subsection{Preliminary Definitions and Technical Lemmas}
Similarly as in Section \ref{secweigh}, we start by providing some definitions and two technical lemmas (Lemmas \ref{lemupp2} and \ref{lemrou2}), which are variants of Lemmas~\ref{lemupp} and \ref{lemrou1}, but applied to unweighted games. We omit the proofs of such lemmas, since they can be obtained by using the same proof arguments  as in Lemmas~\ref{lemupp} and \ref{lemrou1}, with the only difference that the weights are unitary and the congestions are integers. For completeness, we give a sketch of the proofs in the appendix. 

Given $x\geq 1$, $k\in \Z$, $o\in \NN$, and a latency function $f$, let 
\begin{equation}\label{def_alpha_unw}
\beta_{\sf U}({\sf EM},k,o,f):=
\begin{cases}
-kf(k)+(1+\epsilon)of(k+1) & \text{if }{\sf EM}=\poa_\epsilon ,\\
-\sum_{h=1}^kf(h)+(1+\epsilon)of(k+1) & \text{if }{\sf EM}={\sf CR}_\epsilon^s,\\
-kf(k)+(1+\epsilon)((k+1)f(k+1)-kf(k)) & \text{if }{\sf EM}={\sf CR}_\epsilon^c.
\end{cases}
\end{equation}
Given ${\sf EM}\in \{\poa_{\epsilon},{\sf CR}^s_\epsilon, {\sf CR}^c_\epsilon\}$, let 
\begin{equation}\label{def_gamma_prel_unw}
\gamma_{\sf U}({\sf EM},x,k,o,f):=\frac{kf(k)+x\cdot \beta_{\sf U}({\sf EM},k,o,f)}{of(o)};
\end{equation}
furthermore, given a class of unweighted congestion games $\mathcal{G}$, let 
\begin{equation}\label{def_gamma_unw}
\gamma_{\sf U}({\sf EM},\mathcal{G}):=\inf_{x\geq 1}\sup_{f\in\mathcal{C}(\mathcal{G}),k\in \Z, o\in \NN}\gamma_{\sf U}({\sf EM},x,k,o,f).
\end{equation}
\begin{lemma}\label{lemupp2}
Let $\mathcal{G}$ be a class of unweighted congestion games. For any  ${\sf EM}\in \{\poa_{\epsilon},{\sf CR}^s_\epsilon\}$,  we have that ${\sf EM}(\mathcal{G})\leq \gamma_{\sf U}({\sf EM},\mathcal{G})$. This fact holds for ${\sf EM}={\sf CR}^c_\epsilon$ if the latency functions of $\mathcal{C}(\mathcal{G})$ are semi-convex.
\end{lemma}
\begin{lemma}\label{lemrou2}
Let $\mathcal{G}$ be a class of unweighted congestion games. For each $M<\gamma_{\sf U}({\sf EM},\mathcal{G})$, with ${\sf EM}\in \{\poa_{\epsilon},{\sf CR}^s_\epsilon\}$, one of the following cases holds:
\begin{itemize}
\item[$\bullet$] {\bf Case 1:} there exist a latency function $f\in\mathcal{C}(\mathcal{G})$ and two integers $k,o>0$ such that $
M<\frac{kf(k)}{of(o)}$ and $\beta_{\sf U}({\sf EM},k,o,f)\geq 0$. 
\item[$\bullet$] {\bf Case 2:} there exist two latency functions $f_1,f_2\in \mathcal{C}(\mathcal{G})$ and four integers $k_1>0,k_2\geq 0, o_1>0,o_2>0$ such that $M<\frac{\alpha_1k_1f_1(k_1)+\alpha_2k_2f_2(k_2)}{\alpha_1o_1f_1(o_1)+\alpha_2o_2f_2(o_2)}$, where $\alpha_1:=\beta_{\sf U}({\sf EM},k_2,o_2,f_2)>0$ and $\alpha_2 :=-\beta_{\sf U}({\sf EM},k_1,o_1,f_1)>0$.
\end{itemize}
This fact holds for ${\sf EM}={\sf CR}_\epsilon^c$ if the latency functions of $\mathcal{C}(\mathcal{G})$ are semi-convex. 
\end{lemma}
\begin{remark}\label{rema_lemrou_unw}
By exploiting Lemma \ref{lemupp2}, given an arbitrary choice of the input parameters $k,o,f$ (resp. $k_1,k_2,o_1,o_2,f_1,f_2$) such that $\beta_{\sf U}({\sf EM},k,o,f)\geq 0$ (resp. $\beta_{\sf U}({\sf EM},k_2,o_2,f_2)>0$ and $-\beta_{\sf U}({\sf EM},k_1,o_1,f_1)>0$), we have that $\frac{kf(k)}{of(o)}$ (resp. $\frac{\alpha_1k_1f_1(k_1)+\alpha_2k_2f_2(k_2)}{\alpha_1o_1f_1(o_1)+\alpha_2o_2f_2(o_2)}$) is a lower bound on $\gamma_{\sf U}({\sf EM},\mathcal{G})$; furthermore, by definition of $\gamma_{\sf U}({\sf EM},\mathcal{G})$, we have that $\sup_{f\in\mathcal{C}(\mathcal{G}),k\in \Z, o\in \NN}\gamma_{\sf U}({\sf EM},x,k,o,f)$ is an upper bound on $\gamma_{\sf U}({\sf EM},\mathcal{G})$ for any $x\geq 1$. Thus, if such upper bound (for a suitable choice of $x$) is equal to the former lower bound (for a suitable choice of the input parameters), we necessarily have that they are both equal to $\gamma_{\sf U}({\sf EM},\mathcal{G})$.
\end{remark}
\subsection{Main Theorems}
In the following theorem, we prove that, under mild assumptions on the latency functions, no improvements are possible for $\epsilon$-approximate pure Nash equilibria when restricting to load balancing games.
\begin{theorem}\label{thm3}
Let $\mathcal{C}$ be a class of latency functions that is closed under ordinate scaling. Then $\gamma_{\sf U}({\sf PoA}_{\epsilon},{\sf U}(\mathcal{C}))=\poa_{\epsilon}({\sf U}(\mathcal{C}))=\poa_{\epsilon}({\sf ULB}(\mathcal{C}))$.
\end{theorem}
\begin{proof}
Fix $M<\gamma_{\sf U}(\poa_{\epsilon},{\sf U}(\mathcal{C}))$. Similarly as in Theorem \ref{thm1}, we will show that there exists a load balancing instance $\LB\in{\sf ULB}(\mathcal{C})$ such that $\poa_\epsilon(\LB)>M$. By Lemma~\ref{lemupp2}, this will imply that $M<\poa_{\epsilon}(\LB)\leq \poa_{\epsilon}({\sf U}(\mathcal{C}))\leq \gamma_{\sf U}(\poa_{\epsilon},{\sf U}(\mathcal{C}))$, and by the arbitrariness of $M<\gamma_{\sf U}(\poa_{\epsilon},{\sf U}(\mathcal{C}))$, this fact will show the claim. 

Suppose that Case 2 of Lemma~\ref{lemrou2} is verified, and let $k_1,k_2,o_1,o_2,$ $f_1,f_2,\alpha_1,\alpha_2$ be defined as in Lemma~\ref{lemrou2}. As in Theorem \ref{thm1}, we resort to a load balancing graph representation. Consider a load balancing game defined by a multi-partite directed graph $\LB_s(k_1,k_2,o_1,o_2,f_1,f_2)$ as load balancing graph, organized in $2s$ levels, numbered from $1$ to $2s$, and defined as follows. For each $i\in [s]$ (resp. $i\in [2s]_{s+1}$), there are $o_1^{s-i}k_1^{i-1}o_2^{s}$ (resp. $o_2^{2s-i}k_2^{i-s-1}k_1^s$) nodes/resources. Edges can only connect nodes of consecutive levels, except for nodes at level $2s$, each of which has $k_2$ self-loops. The out-degree of each node at level $i\in [s]$ (resp. $i\in [2s]_{s+1}$) is $k_1$ (resp. $k_2$), and the in-degree of each node at level $i\in [s]_2$ (resp. $i\in [2s]_{s+1}$ without considering self-loops) is $o_1$ (resp. $o_2$); observe that this configuration can be realized since the total number of nodes at level $i\in [s-1]$ (resp. $i=s$, resp. $i\in [2s-1]_{s+1}$) multiplied by $k_1$ (resp. $k_1$, resp. $k_2$) is equal to the number of nodes at level $i+1$ multiplied by $o_1$ (resp. $o_2$, resp. $o_2$). See the example depicted in Figure \ref{fig:2}. For $i,j\in [2]$, define $\theta_{i,j}:=\frac{f_i(k_i)}{(1+\epsilon)f_j(k_j+1)}$ and $\theta_i:=\theta_{i,i}$. Each resource at level $i$ has latency function $g_i(x):=\theta_{1}^{i-1}f_1\left(x\right)$  if $i\in [s]$, and $g_i(x):= \theta_{1}^{s-1}\theta_{1,2}\theta_{2}^{i-s-1}f_2\left(x\right)$ otherwise. We assume that each player can only choose among her first and her second strategy, and let $\bm\sigma_s$ and $\bm\sigma^{*}_s$ be the strategy profiles in which all players select their first and second strategy, respectively. Analogously to Theorem~\ref{thm1}, one can show that $\bm\sigma_s$ is an $\epsilon$-approximate pure Nash equilibrium.

In the remainder of the proof, with the aim at estimating the approximate price of anarchy of $\LB_{s}(k_1,k_2,o_1,o_2,f_1,f_2)$, we compare the social value of the $\epsilon$-pure Nash equilibrium $\bm\sigma_{s}$, with that of the strategy profile $\bm\sigma^*_{s}$ in which all players select their second strategy. By exploiting similar arguments as in Theorem \ref{thm1}, we get
\begin{align}
&{\sf SUM}(\bm\sigma_s)\nonumber\\
&=\sum_{i=1}^s\left(o_1^{s-i}k_1^{i-1}o_2^{s}\right)k_1g_i(k_1)+\sum_{i=s+1}^{2s}\left(o_2^{2s-i}k_2^{i-s-1}k_1^s\right)k_2g_i(k_2)\nonumber\\
&=\sum_{i=1}^s\left(o_1^{s-i}k_1^{i-1}o_2^{s}\right)k_1\theta_{1}^{i-1}f_1(k_1)\nonumber\\
&\ \ \ \ \ \ \ +\sum_{i=s+1}^{2s}\left(o_2^{2s-i}k_2^{i-s-1}k_1^s\right)k_2\theta_{1}^{s-1}\theta_{1,2}\theta_{2}^{i-s-1}f_2(k_2)\nonumber\\
&=\sum_{i=1}^so_1^{s-1}o_2^s \left(\frac{k_1\theta_{1}}{o_1}\right)^{i-1}k_1f_1(k_1)\nonumber\\
&\ \ \ \ \ \ \ +\sum_{i=s+1}^{2s}o_1^{s-1}o_2^s\left(\frac{k_1\theta_{1}}{o_1}\right)^{s-1}\frac{k_1\theta_{1,2}}{o_2}\left(\frac{k_2\theta_{2}}{o_2}\right)^{i-s-1}k_2f_2(k_2)\nonumber\\
&=o_1^{s-1}o_2^s\sum_{i=0}^{s-1}\left(\frac{k_1\theta_{1}}{o_1}\right)^{i}k_1f_1(k_1)\nonumber\\
&\ \ \ \ \ \ \ +o_1^{s-1}o_2^s\left(\frac{k_1\theta_{1}}{o_1}\right)^{s-1}\frac{k_1\theta_{1,2}}{o_2}\sum_{i=0}^{s-1}\left(\frac{k_2\theta_{2}}{o_2}\right)^{i}k_2f_2(k_2)\label{conteq2}\\
&=o_1^{s-1}o_2^s\left(\frac{\left(\frac{k_1\theta_{1}}{o_1}\right)^s-1}{\frac{k_1\theta_{1}}{o_1}-1}\right)k_1f_1(k_1)\nonumber\\
&\ \ \ \ \ \ \ +o_1^{s-1}o_2^s\left(\frac{k_1\theta_{1}}{o_1}\right)^{s-1}\frac{k_1\theta_{1,2}}{o_2}\left(\frac{1-\left(\frac{k_2\theta_{2}}{o_2}\right)^s}{1-\frac{k_2\theta_{2}}{o_2}}\right)k_2f_2(k_2),\label{equeq2}
\end{align}
and
\begin{align}
&{\sf SUM}(\bm\sigma^*_s)\nonumber\\
&=\sum_{i=2}^{s}\left(o_1^{s-i}k_1^{i-1}o_2^{s}\right)o_1g_i(o_1)+\sum_{i=s+1}^{2s-1}\left(o_2^{2s-i}k_2^{i-s-1}k_1^s\right)o_2g_i(o_2)\nonumber\\
&\ \ \ \ \ \ \ +\left(k_2^{s-1}k_1^s\right)(k_2+o_e)g_{2s}(k_2+o_2)\nonumber\\
&=\sum_{i=2}^{s}\left(o_1^{s-i}k_1^{i-1}o_2^{s}\right)o_1\theta_{1}^{i-1}f_1(o_1)\nonumber\\
&\ \ \ \ \ \ \ +\sum_{i=s+1}^{2s-1}\left(o_2^{2s-i}k_2^{i-s-1}k_1^s\right)o_2\theta_{1}^{s-1}\theta_{1,2}\theta_{2}^{i-s-1}f_2(o_2)\nonumber\\
&\ \ \ \ \ \ \ +\left(k_2^{s-1}k_1^s\right)(k_2+o_e)\theta_{1}^{s-1}\theta_{1,2}\theta_{2}^{s-1}f_2(k_2+o_2)\label{contopt2}\\
&=o_1^{s-1}o_2^s\sum_{i=1}^{s}\left(\frac{k_1\theta_{1}}{o_1}\right)^{i-1}o_1f_1(o_1)-\overbrace{\left(o_1^{s-1}o_2^{s}\right)o_1f_1(o_1)}^{[<]_s\left(o_1^{s-1}o_2^s\left(\frac{k_1\theta_{1}}{o_1}\right)^{s-1}\right)}\nonumber\\
&\ \ \ \ \ \ \ +o_1^{s-1}o_2^s\sum_{i=s+1}^{2s-1}\left(\frac{k_1\theta_{1}}{o_1}\right)^{s-1}
\frac{k_1\theta_{1,2}}{o_2}\left(\frac{k_2\theta_{2}}{o_2}\right)^{i-s-1}o_2f_2(o_2)\nonumber\\
&\ \ \ \ \ \ \ +\overbrace{o_1^{s-1}o_2^s\left(\frac{k_1\theta_{1}}{o_1}\right)^{s-1}\frac{k_1\theta_{1,2}}{o_2}\left(\frac{k_2\theta_{2}}{o_2}\right)^{s-1}(k_2+o_e)f_2(k_2+o_e)}^{[<]_s\left(o_1^{s-1}o_2^s\left(\frac{k_1\theta_{1}}{o_1}\right)^{s-1}\right)}\nonumber\\
&=o_1^{s-1}o_2^s\sum_{i=0}^{s-1}\left(\frac{k_1\theta_{1}}{o_1}\right)^{i}o_1f_1(o_1)\nonumber\\
&\ \ \ \ \ \ \ +o_1^{s-1}o_2^s\left(\frac{k_1\theta_{1}}{o_1}\right)^{s-1}\frac{k_1\theta_{1,2}}{o_2}\sum_{i=0}^{s-1}\left(\frac{k_2\theta_{2}}{o_2}\right)^{i}o_2f_2(o_2)\nonumber\\
&\ \ \ \ \ \ \ +[<]_s\left(o_1^{s-1}o_2^s\left(\frac{k_1\theta_{1}}{o_1}\right)^{s-1}\right)\label{opteq2}\\
&=o_1^{s-1}o_2^s\left(\frac{\left(\frac{k_1\theta_{1}}{o_1}\right)^s-1}{\frac{k_1\theta_{1}}{o_1}-1}\right)o_1f_1(o_1)\nonumber\\
&\ \ \ \ \ \ \ +o_1^{s-1}o_2^s\left(\frac{k_1\theta_{1}}{o_1}\right)^{s-1}\frac{k_1\theta_{1,2}}{o_2}\left(\frac{1-\left(\frac{k_2\theta_{2}}{o_2}\right)^s}{1-\frac{k_2\theta_{2}}{o_2}}\right)o_2f_2(o_2)\nonumber\\
&\ \ \ \ \ \ \ +[<]_s\left(o_1^{s-1}o_2^s\left(\frac{k_1\theta_{1}}{o_1}\right)^{s-1}\right).\label{optzero1}
\end{align}
By using equalities (\ref{equeq2}) and (\ref{optzero1}), and the fact that $\frac{k_1\theta_{1}}{o_1}>1$ and $\frac{k_2\theta_{2}}{o_2}<1$ (since, by Lemma~\ref{lemrou2}, we have $\alpha_1,\alpha_2>0$) we get 
\begin{align}
&\sup_{s\in \NN}\poa_\epsilon(\LB_s(k_1,k_2,o_1,o_2,f_1,f_2))\nonumber\\
&\geq \lim_{s\rightarrow \infty}\frac{{\sf SUM}(\bm\sigma_s)}{{\sf SUM}(\bm\sigma^*_s)}\nonumber\\
&=\lim_{s\rightarrow \infty}\frac{{\sf SUM}(\bm\sigma_s) \left(o_1^{s-1}o_2^s\left(\frac{k_1\theta_{1}}{o_1}\right)^{s-1}\right)}{{\sf SUM}(\bm\sigma^*_s) \left(o_1^{s-1}o_2^s\left(\frac{k_1\theta_{1}}{o_1}\right)^{s-1}\right)}\nonumber\\
&=\frac{\left(\frac{\frac{k_1\theta_{1}}{o_1}}{\frac{k_1\theta_{1}}{o_1}-1}\right)k_1f_1(k_1)+\frac{k_1\theta_{1,2}}{o_2}
\left(\frac{1}{1-\frac{k_2\theta_{2}}{o_2}}\right)k_2f_2(k_2)}
{\left(\frac{\frac{k_1\theta_{1}}{o_1}}{{\frac{k_1\theta_{1}}{o_1}-1}}\right)o_1f_1(o_1)+\frac{k_1\theta_{1,2}}{o_2}\left(\frac{1}{1-\frac{k_2\theta_{2}}{o_2}}\right)o_2f_2(o_2)}\nonumber\\
&=\frac{\left(\frac{k_1f_1(k_1)}{\alpha_2}\right)k_1f_1(k_1)+\left(\frac{k_1f_1(k_1)}{\alpha_1}\right)k_2f_2(k_2)}{\left(\frac{k_1f_1(k_1)}{\alpha_2}\right)o_1f_1(o_1)+\left(\frac{k_1f_1(k_1)}{\alpha_1}\right)o_2f_2(o_2)}\nonumber\\
&=\frac{\alpha_1k_1f_1(k_1)+\alpha_2 k_2f_2(k_2)}{\alpha_1o_1f_1(o_1)+\alpha_2o_2f_2(o_2)}\nonumber\\
&>M\label{finaleunw}
\end{align}
where (\ref{finaleunw}) comes from Lemma~\ref{lemrou2}. Then, there exists an integer $s$ such that $\poa_\epsilon(\LB_s(k_1,k_2,o_1,o_2,f_1,f_2))>M$, and this shows the claim. 

Finally, suppose that Case 1 of Lemma~\ref{lemrou2} holds and let $k,o,f$ be the parameters defined for that case. Consider the load balancing graph $\LB_s(k,k,o,o,f,f)$. Observe that, analogously to the previous cases, $\bm\sigma_s$ is an $\epsilon$-approximate pure Nash equilibrium. Indeed, it suffices considering the previous proofs with $k_1=k_2=k$. To evaluate ${\sf SUM}(\bm\sigma_s)$, observe that all the equalities up to (\ref{conteq2}) hold. Therefore, by continuing from (\ref{conteq2}) we get $${\sf SUM}(\bm\sigma_s)=o_1^{s-1}o_2^s 2skf(k)$$ if $\beta_{\sf U}(\poa_\epsilon,k,o,f)=0$ (as in such case we get $\frac{k_1\theta_{1}}{o_1}=\frac{k_2\theta_{2}}{o_2}=1$), and
$${\sf SUM}(\bm\sigma_s)=o_1^{s-1}o_2^s\left(\frac{1-\left(\frac{kf(k)}{(1+\epsilon)of(k+1)}\right)^{2s}}{1-\frac{kf(k)}{(1+\epsilon)of(k+1)}}\right)kf(k)$$ if $\beta_{\sf U}(\poa_\epsilon,k,o,f)>0$ (as in such case we get $\frac{k_1\theta_{1}}{o_1}=\frac{k_2\theta_{2}}{o_2}<1$). To evaluate ${\sf SUM}(\bm\sigma^*_s)$, observe that all the equalities up to (\ref{contopt2}) hold. Therefore, by continuing from (\ref{contopt2}), we get $${\sf SUM}(\bm\sigma^*_s)=o_1^{s-1}o_2^s(2sof(o)+[\leq]_s(1))$$  if $\beta_{\sf U}(\poa_\epsilon,k,o,f)=0$, and
$${\sf SUM}(\bm\sigma^*_s)=o_1^{s-1}o_2^s\left(\left(\frac{1-\left(\frac{kf(k)}{(1+\epsilon)of(k+1)}\right)^{2s}}{1-\frac{kf(k)}{(1+\epsilon)of(k+1)}}\right)of(o)-of(o)+[<]_s(1)\right)$$ if $\beta_{\sf U}(\poa_\epsilon,k,o,f)>0.$
Then:
\begin{equation*}
\sup_{s\in \NN}\poa_\epsilon(\LB_s(k,k,o,o,f,f))\geq \lim_{s\rightarrow \infty}\frac{{\sf SUM}(\bm\sigma_s)}{{\sf SUM}(\bm\sigma^*_s)}\geq \frac{kf(k)}{of(o)}>M,
\end{equation*}
where the last inequality follows from Lemma \ref{lemrou2}. Thus, there exists an integer $s$ such that $\poa_\epsilon(\LB_s(k_1,k_2,o_1,o_2,f_1,f_2))>M$, and this shows the claim if the first Case of Lemma \ref{lemrou2} holds. \qed
\end{proof}

\begin{figure}[h]
\centering
\includegraphics[scale=0.65]{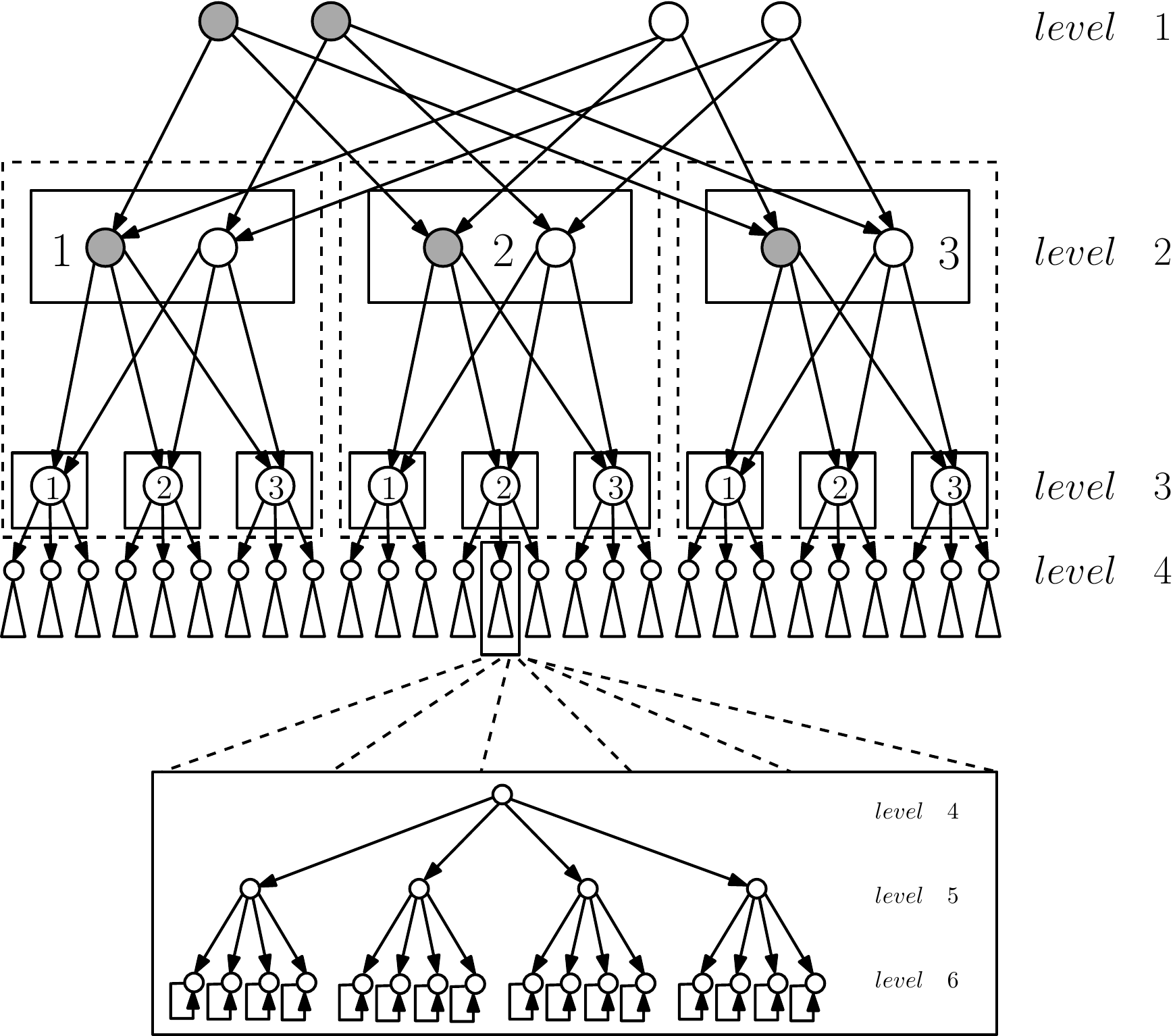}
\caption{The load balancing graph described in the proof of Theorems \ref{thm3} and \ref{thm4}, with $s=3$, $k_1=3$, $o_1=2$, $k_2=4$ and $o_2=1$. We also describe the partitioning and labeling structures used in the proof of Theorem~\ref{thm4}. }\label{fig:2}
\end{figure}

Then, we prove a similar limitation for approximate one-round walks.

\begin{theorem}\label{thm4}
Let $\mathcal{C}$ be a class of latency functions that is closed under ordinate scaling. Then $\gamma_{\sf U}({\sf CR}^s_{\epsilon},{\sf U}(\mathcal{C}))={\sf CR}^s_{\epsilon}({\sf U}(\mathcal{C}))={\sf CR}^s_{\epsilon}({\sf ULB}(\mathcal{C}))$. If the functions of $\mathcal{C}$ are semi-convex, we have that $\gamma_{\sf U}({\sf CR}^c_{\epsilon},{\sf U}(\mathcal{C}))={\sf CR}^c_{\epsilon}({\sf U}(\mathcal{C}))={\sf CR}^c_{\epsilon}({\sf ULB}(\mathcal{C}))$.
\end{theorem}
\begin{proof}
Fix $M<\gamma_{\sf U}({\sf EM},{\sf U}(\mathcal{C}))$, where ${\sf EM}\in \{{\sf CR}^s_{\epsilon},{\sf CR}^c_{\epsilon}\}$ is the considered efficiency metric. As in Theorem \ref{thm3}, by Lemma \ref{lemupp2}, it is sufficient showing that there exists a load balancing game $\LB\in {\sf ULB}(\mathcal{C})$ such that its competitive ratio is at least $M$. 

Let us start with the case of selfish players. Let $k_1,k_2,o_1,o_2,f_1,f_2,\alpha_1,\alpha_2$ be the parameters defined as in Lemma~\ref{lemrou2}, for which we remove index $j\in [2]$ if the first case of Lemma~\ref{lemrou2} is verified. Define $j(i)=1$ if $i\in [s]$ and $j(i)=2$ otherwise. We extend the load balancing graph ${\sf LB}(k_1,k_2,o_1,o_2,f_1,f_2)$ used in the proof of Theorem~\ref{thm3} according to the following recursive procedure.
\begin{itemize}
\item[$\bullet$] {\bf Base Case:} partition the resources of the first level (resp. second level) of the load balancing graph in $o_{j(2)}$ (resp. $k_{j(1)}$) groups of equal size, and add edges from the first level to the second one in such a way that each resource in the first level has exactly $k_{j(1)}$ outgoing edges, each ending in a different group of the second level, and each resource in the second level has exactly $o_{j(2)}$ incoming edges, each coming from a different group of the first level; number the groups of the second level from $1$ to $k_{j(1)}$ and label each resource with the number associated to the group it belongs to. For an illustrating example see Figure \ref{fig:2}, where resources belonging to different groups at level $1$ are represented with different colors, resources belonging to different groups at level $2$ belong to different rectangles and they are labeled with the number of the rectangle they belong to.
\item[$\bullet$] {\bf Inductive Case:} as inductive hypothesis, suppose that resources at level $i\in [2s-1]$ have been partitioned into $m(i)$ groups of equal size and labeled with values from $1$ to $k_{j(i-1)}$, where each label is assigned to $m(i)/k_{j(i-1)}$ distinct groups, and that all the edges from level $i-1$ to level $i$ have been added. Partition resources at level $i+1$ in a temporary partition of $m(i)$ groups of equal size, and consider a bijective correspondence between groups at level $i$ and groups at level $i+1$ (in Figure \ref{fig:2}, groups at levels $2$ and $3$ which are in bijective correspondence have been depicted in the same dashed square). Partition each group at level $i$ into $o_{j(i+1)}$ subgroups of equal size, and the corresponding group at level $i+1$ (according to the bijective correspondence considered above) into $k_{j(i)}$ subgroups of equal size, and add edges from each group at level $i$ to the corresponding group at level $i+1$ in the same way as described in the base case, i.e., each resource in the group at level $i$ has exactly $k_{j(i)}$ outgoing edges, each ending in a different subgroup of the corresponding group at level $i+1$, and each resource in this group has exactly $o_{j(i+1)}$ incoming edges, each coming from a different subgroup of the considered group at level $i$. For each group at level $i+1$, number its subgroups with values from $1$ to $k_{j(i)}$ and label each resource with the number associated to the subgroup which the considered resource belongs to. The final partitioning of nodes at level $i+1$ is made of all the $m(i+1):=m(i)k_{j(i)}$ subgroups considered above, these groups have equal size, they are labeled with values from $1$ to $k_{j(i)}$, each label is assigned to $m(i+1)/k_{j(i-1)}=m(i)$ distinct groups, and all the edges from level $i$ to level $i+1$ have been added. Thus the inductive step is well-defined and allows to construct the structure of the load balancing graphs, starting from the lowest levels, to the higher ones.

For instance, Figure \ref{fig:2} depicts the recursive procedure used to construct a load balancing graph. In Figure \ref{fig:2}, each arbitrary dashed square includes a group at level $2$ and a group at level $3$ which are in bijective correspondence. Analogously to the base case, resources belonging to different subgroups of the considered group at level $2$ (resp. at level $3$)  are represented with different colors (resp. belong to different squares and are labeled with the number of the square they belong to).
\end{itemize}
Again, all the players can select their first and second strategies, only. Let $i(v)$ be  the level of resource $v$ and $h(v)\in [k_{j(i(v))}]$ be the label assigned to resource $v$ by the previous recursive procedure. 

For $i,j\in [2]$ and $h\in k_i$, define $\theta_{i,j}(h):=\frac{f_i(h)}{(1+\epsilon)f_j(k_j+1)}$ and $\theta_i(h):=\theta_{i,i}(h)$. Resource $v$ has latency function $g_v(x):=f_1(x)$ if $i(v)=1$,
$g_v(x):=\underbrace{\theta_1(h(v)) A_u}_{A_v} f_1(x)$ if $i(v)\in [s]_2$,
$g_v(x):=\underbrace{\theta_{1,2}(h(v)) A_u}_{A_v}  f_2(x)$ if $i(v)=s+1$, and
$g_v(x):=\underbrace{\theta_2(h(v)) A_u}_{A_v}  f_2(x)$ otherwise,
where $(u,v)$ is an arbitrary incoming edge of $v$ and $A_v$ is recursively defined on the basis of $A_u$ by setting $A_v=1$ for $i(v)=1$.
By using the recursive structure of the load balancing graph, one can prove that $A_u=A_{u'}$ if $(u,v)$ and $(u',v)$ are both edges of the load balancing graph, so that the definition of $A_v$ (and then $g_v$) is independent of the particular incoming edge of $v$.

Consider the online process $\bm\tau_s$ in which: (i) players enter the game in non-increasing order of level (with respect to their first strategy) and, within the same level, players are processed in non-decreasing order of position defined by labelling function $h$; (ii) each player selects her first strategy. Similarly as in Lemma \ref{oneround_claim}, one can show that $\bm \tau_s$ is an $\epsilon$-approximate one-round walk generated by selfish players. To show this, let $z$ be an arbitrary player entering the game according to $\bm \tau_s$, and let $u$ and $v$ be her first and second strategy, respectively. One can show that, when $z$ is entering the game, there are $h(v)-1$ players already assigned to $u$, and $k_{j(i(v))}$ players assigned to $v$; thus, similarly as in Lemma \ref{oneround_claim}, and by using the recursive definition of the latency functions, one can show that the individual cost of player $z$ is exactly $(1+\epsilon)$ times the cost resulting from a deviation in favour of her second strategy. 

Now, let $\sg$ and ${\bm \sigma}^*_s$ be the strategy profile generated by the $\epsilon$-approximate one-round walk $\bm \tau_s$ and the optimal strategy profile, respectively.  In the remainder of the proof, with the aim of estimating the competitive ratio of $\LB_{s,n}(k_1,k_2,o_1,o_2,f_1,f_2)$, we compare the social value of the strategy profile $\sg_s$ with that of $\sg^*_s$. Let ${\sf SUM}_i({\bm\sigma_s})$ denote the total latency of resources under strategy profile $\sg$. By using similar approaches as in Theorem \ref{thm2} (since the metric is the competitive ratio) and Theorem \ref{thm3} (since the game is unweighted), one can prove that 
\begin{align*}
{\sf SUM}_i({\bm\sigma_s})
&=\sum_{v\text{ in level }i}k_v(\bm\sigma_s)g_{v}\left(k_v(\bm\sigma_s)\right)\\
&=\left(o_1^{s-i}o_2^{s}\right)\sum_{(h_1,h_2,\ldots, h_{i-1})\in [k_1]^{i-1}}\left(\prod_{t=1}^{i-1}\theta_1(h_t)\right)k_1f_1(k_1)\\
&=\left(o_1^{s-1}o_2^{s}\right)\left(\frac{\sum_{h=1}^{k_1}\theta_1(h)}{o_1}\right)^{i-1}k_1f_1(k_1)\\
&=\left(o_1^{s-1}o_2^{s}\right)\left(\frac{\sum_{h=1}^{k_1}f_1(h)}{o_1f_1(k_1+1)}\right)^{i-1}k_1f_1(k_1)
\end{align*}
for any $i\in [s]$, and 
\begin{align*}
{\sf SUM}_i({\bm\sigma_s})
&=\left(o_1^{s-1}o_2^{s}\right)\left(\frac{\sum_{h=1}^{k_1}\theta_1(h)}{o_1}\right)^{s-1}\left(\frac{\sum_{h=1}^{k_1}\theta_{1,2}(h)}{o_2}\right)\\
&\ \ \ \ \ \ \ \left(\frac{\sum_{h=1}^{k_2}\theta_{2}(h)}{o_2}\right)^{i-1-s}k_2f_2(k_2)\\
&=\left(o_1^{s-1}o_2^{s}\right)\left(\frac{\sum_{h=1}^{k_1}f_1(h)}{o_1f_1(k_1+1)}\right)^{s-1}\left(\frac{\sum_{h=1}^{k_1}f_1(h)}{o_2f_2(k_2+2)}\right)\\
&\ \ \ \ \ \ \ \left(\frac{\sum_{h=1}^{k_2}f_2(h)}{o_2f_2(k_2+2)}\right)^{i-1-s}k_2f_2(k_2)
\end{align*}
Thus, we can compute ${\sf SUM}({\bm\sigma_s})=\sum_{i=1}^{2s}{\sf SUM}_i({\bm\sigma_s})$ and, analogously, we can compute ${\sf SUM}({\bm\sigma}_s^*)$. At this point, by using the same proof arguments of Theorems~\ref{thm2} and \ref{thm3}, we can prove that
\begin{align}
\lim_{s\in \NN}{\sf CR}_\epsilon^s(\LB_{s,n}(k_1,k_2,o_1,o_2,f_1,f_2))
&\geq \lim_{s\rightarrow \infty}\frac{{\sf SUM}({\bm\sigma_s})}{{\sf SUM}({\bm\sigma_s}^*)}\nonumber\\
&\geq \frac{\alpha_1k_1f_1(k_1)+\alpha_2k_2f_2(k_2)}{\alpha_1o_1f_1(o_1)+\alpha_2o_2f_2(o_2)}\nonumber\\
&>M,\label{unwfin4}
\end{align}
where the last inequality comes from Lemma \ref{lemrou2}, and this shows the claim in the case of selfish players. 

For the case of cooperative players, it suffices considering the same load balancing graph with $\theta_{i,j}(h)=\frac{h_if_i\left(h_i\right)-(h_i-1)f_i(h_i-1)}{(1+\epsilon)((k_j+1)f_j(k_j+1)-k_jf_j(k_j))}$. Let $\sg_s$ and $\sg^*_s$ be the strategy profiles in which each player chooses her first and second strategy, respectively. By using the same arguments as in the previous proof, one can show that inequality \eqref{unwfin4} holds for the case of cooperative players, too.\qed 
\end{proof}
\begin{remark}\label{remalowunw}
Fix a metric ${\sf EM}\in\{{\sf PoA}_\epsilon,{\sf CR}^s_{\epsilon},{\sf CR}^c_{\epsilon}\}$, a class of latency functions $\mathcal{C}$, and let $(k_1,k_2,o_1,o_2,f_1,f_2)$ (resp. $(k,o,f)$) be a tuple such that values $\alpha_1,\alpha_2$ considered in Case~2 of Lemma~\ref{lemrou2} are positive (resp. such that the inequality of Case 1 corresponding to the considered metric {\sf EM} is satisfied, i.e., $\beta_{\sf U}({\sf EM},k,o,f)\geq 0$). As in Remark~\ref{remalowweig}, by inspecting the proofs of Theorems~\ref{thm3} and \ref{thm4}, we can construct a load balancing instance parametrized by the above tuple whose performance is at least $\frac{\alpha_1k_1 f_1(k_1)+\alpha_2 k_2 f_2(k_2)}{\alpha_1o_1 f_1(o_1)+\alpha_2 o_2 f_2(o_2)}$ (resp. $\frac{k f(k)}{o f(o)}$). Thus, in case we are not able to quantify the exact value of $\gamma_{\sf U}({\sf EM}, {\sf U}(\mathcal{C}))$, we can still find a choice of the above parameters that leads to good (and possibly tight) lower bounds; the tightness can be shown by resorting to the characterization provided in Remark \ref{rema_lemrou_unw}. 
\end{remark}
\subsection{Application to Polynomial Latency Functions}
Let $\gamma_{\sf U}({\sf EM},{\sf U}(\mathcal{P}(d)))$ with ${\sf EM}\in\{{\sf PoA}_\epsilon,{\sf CR}^s_{\epsilon},{\sf CR}^c_{\epsilon}\}$ be the quantity defined at the beginning of the section, where $\mathcal{P}(d)$ is the class of polynomial latency functions of maximum degree $d$. By Lemma \ref{lemupp2}, $\gamma_{\sf U}({\sf EM},{\sf U}(\mathcal{P}(d)))$ constitutes an upper bound on the efficiency metrics ${\sf EM}\in\{{\sf PoA}_\epsilon,{\sf CR}^s_{\epsilon},{\sf CR}^c_{\epsilon}\}$ for unweighted congestion games and load balancing games with polynomial latency functions of maximum degree $d$.

$\gamma_{\sf U}({\sf PoA}_{\epsilon},{\sf U}(\mathcal{P}(d)))$ has been already evaluated in \cite{ADGMS11,CKS11}. $\gamma_{\sf U}({\sf CR}^c_{\epsilon},{\sf U}(\mathcal{P}(d)))$ has been evaluated in \cite{CFKKM11} for exact pure Nash equilibria and affine functions;\footnote{As in Subsection \ref{wplf}, the value of $\gamma_{\sf U}({\sf EM},{\sf U}(\mathcal{P}(d)))$ can be easily derived from equivalent results appearing in the cited works, thus the analysis leading to such quantification has been omitted.} by using similar arguments as in \cite{ADGMS11}, one can easily compute the exact value of $\gamma_{\sf U}({\sf CR}^c_{\epsilon},{\sf U}(\mathcal{P}(d)))$ for any $\epsilon>0$. 

The above values of ${\sf EM}\in\{{\sf PoA}_\epsilon,{\sf CR}^c_{\epsilon}\}$ can be also represented by resorting to the characterization provided by Remark \ref{rema_lemrou_unw}. In particular, let $k$ be the unique real solution of equation $\beta_{\sf U}({\sf EM},k,1,f)=0$, where $f$ is the monomial function defined as $f(t)=t^d$. If $k$ is not integer we have that $\gamma_{\sf U}({\sf EM},{\sf U}(\mathcal{P}(d)))$ is equal to $\frac{\alpha_1k_1f_1(k_1)+\alpha_2k_2f_2(k_2)}{\alpha_1o_1f_1(o_1)+\alpha_2o_2f_2(o_2)}$, where $(k_1,k_2,o_1,o_2):=(\lfloor k\rfloor +1,\lfloor k\rfloor,1,1)$, $f_1:=f_2:=f$, and $\alpha_1$ and $\alpha_2$ are defined as in Case 2 of Lemma \ref{lemrou2} with respect to parameters $k_1,k_2,o_1,o_2,f_1,f_2$. Instead, if $k$ is integer, we have that $\gamma_{\sf U}({\sf EM},{\sf U}(\mathcal{P}(d)))$ is equal to $\frac{kf(k)}{f(1)}$, i.e., Case 1 of Lemma \ref{lemrou2} is satisfied by tuple $(k,1,f)$. 

Since the class of polynomial latency functions $\mathcal{P}(d)$ satisfies the hypothesis of Theorem~\ref{thm3} and \ref{thm4}, as a corollary we have that the values of $\gamma_{\sf U}({\sf EM},{\sf U}(\mathcal{P}(d)))$ considered above are tight upper bounds on the performance of load balancing games, thus matching the performance of general congestion games; this generalizes/improves some results of \cite{CFKKM11,GS07}, in which the considered lower bounds hold for load balancing games under some restrictions (e.g., $d=1$ for the competitive ratio, and $\epsilon=0$ for the price of anarchy). 

Regarding the $\epsilon$-approximate one-round walks generated by selfish players, we have that the exact value of $\gamma_{\sf U}({\sf CR}^s_{\epsilon},{\sf U}(\mathcal{P}(d)))$ has been provided in \cite{B12} for $d\in [3]$ and $\epsilon=0$, thus, by Theorem \ref{thm4}, such value is tight even for load balancing games; even for the case of affine latencies (i.e., $d=1$), this fact also improves a result of \cite{BFFM09}, in which a tight lower bound has been provided for general congestion games only, and the case of singleton strategies was left open.

For more general polynomial latency functions with maximum degree $d>3$ and for any $\epsilon\geq 0$, an upper bound on $\gamma_{\sf U}({\sf CR}^s_{\epsilon},{\sf U}(\mathcal{P}(d)))$ can be trivially obtained by reusing the upper bound shown in Subsection \ref{wplf}, that holds for more general weighted games; for exact one-round walks, a better upper bound has been recently provided in \cite{Klimm19}. By Remark \ref{remalowunw}, we can get good lower bounds on the performance of $\epsilon$-approximate one-round walks, having a similar representation as in the cases of ${\sf EM}\in\{{\sf PoA}_\epsilon,{\sf CR}^c_{\epsilon}\}$. In particular, let $k$ be the highest non-negative integer such that $\beta_{\sf U}({\sf CR}_\epsilon^s,k,1,f)\geq 0$, where $f$ is the monomial function defined as $f(t)=t^d$; if $\beta_{\sf U}({\sf CR}_\epsilon^s,k,1,f)>0$ (resp. $\beta_{\sf U}({\sf CR}_\epsilon^s,k,1,f)=0$), we can consider the lower bounding instances defined in Theorem \ref{thm4}, parametrized by $(k+1,k,1,1,f,f)$ (resp. $(k,k,1,1,f,f)$); by Remark \ref{remalowunw}, the performance of such load balancing instances is equal to $\frac{\alpha_1k_1f_1(k_1)+\alpha_2k_2f_2(k_2)}{\alpha_1o_1f_1(o_1)+\alpha_2o_2f_2(o_2)}$ (resp. $\frac{k f(k)}{o f(o)}$), where $\alpha_1$ and $\alpha$ are defined as in Lemma \ref{lemrou2}. Such load balancing instances  improve the lower bounds obtained in \cite{BV16}, where it is showed that the $\epsilon$-approximate competitive ratio of general unweighted congestion games is at least equal to the $(d+1)$-th Geometric polynomial evaluated in $(1+\epsilon)$; furthermore, for $d>3$, we conjecture that the performance of the obtained load balancing instances match  the exact value of $\gamma_{\sf U}({\sf CR}^s_{\epsilon},{\sf U}(\mathcal{P}(d)))$, and in such case, they would match the competitive ratio of general congestion games with polynomial latency functions of maximum degree $d$. 
See Table \ref{figurab2} for some numerical comparisons. 
\begin{figure}[!t]
\begin{center}
\begin{tabular}{|c|c|c|c|}
  \hline
  $d$ & ${\sf PoA}_0({\sf U}(\mathcal{P}(d)))$ \cite{ADGMS11,GS07}  & ${\sf CR}_0^s({\sf U}(\mathcal{P}(d)))$ & ${\sf CR}_0^c({\sf U}(\mathcal{P}(d)))$ \\\hline
  1 & 2.5 \cite{CK05,AAE05,CFKKM11} &  4.236 \cite{CMS12,BFFM09} & 5.66 \cite{CFKKM11} \\
  2 & 9.583 & 37.58 \cite{B12} & 55.46 \\
  3 & 41.54 & 527.3 \cite{B12} & 755.2 \\
  4 & 267.6 & 9,387 (L.B.) & 13,170\\
  5 & 1,514  & 201,401 (L.B.) & 289,648 \\
  6 & 12,345  & 5,276,150 (L.B.) & 7,174,495\\
  7 & 98,734  & 151,192,413 (L.B.)& 220,349,064\\
  8 & 802,603  & 5,287,749,084 (L.B.) & 7,022,463,077\\
  $\vdots$ &  $\vdots$ &  $\vdots$ &  $\vdots$ \\
  $\infty$ & $(\Theta(d/\log(d)))^{d+1}$  & $(\Theta(d))^{d+1}$ & $(\Theta(d))^{d+1}$\\
  \hline
\end{tabular}
\caption{The competitive ratio of exact one-round walks generated by cooperative players in unweighted load balancing games
with polynomial latency functions of maximum degree $d$. For $d\geq 2$, our findings provide tight lower bounds on the competitive ratio involving cooperative players, holding even for load balancing games. For $d\in [3]$, we show that the existing upper bounds on the competitive ratio involving selfish players are tight even for load balancing games; instead, for $d>3$, we provide almost tight lower bounds (see Table \ref{figurab1} for the upper bounds). We observe that, when considering the solution concept of one-round walk, the selfish behaviour induce better outcomes than the cooperative one (in contrast to what happens for weighted games).}\label{figurab2}
\end{center}
\end{figure}

\section{The Case of Identical Resources}\label{sec_ident}
\subsection{The Approximate Price of Anarchy of Weighted Games}
In this section, we characterize the approximate price of anarchy of weighted symmetric load balancing games with identical resources having semi-convex latency functions. Given $\epsilon\geq 0$,  $\lambda\in (0,1)$, a latency function $f$, and $x>0$, let 
\begin{align*}
&[x]_{\epsilon,f}:=\inf\{t\geq 0: f(x)\leq (1+\epsilon)f(x/2+t)\},\\
&\gamma_{\epsilon,f}(x,\lambda):=\frac{\lambda xf(x)+(1-\lambda)[x]_{\epsilon,f} f([x]_{\epsilon,f})}{(\lambda x+(1-\lambda)[x]_{\epsilon,f}) f(\lambda x+(1-\lambda)[x]_{\epsilon,f})}.
\end{align*}
In Theorem \ref{thm5} and \ref{thm6}, we provide respectively upper and lower bounds on the $\epsilon$-approximate price of anarchy (which are tight, under mild assumptions), and we generalize the results obtained in \cite{LMMR08,GLMM06}, holding for the restricted case of affine and monomial latency functions with $\epsilon=0$. The high-level idea of the proofs, which are given in the appendix, is showing that the price of anarchy of class ${\sf WSLB}(\{f\})$, where $f$ is a  semi-convex latency function, is given by instances verifying the following conditions: (i) the number $m$ of resources and the total weight $W$ of players tends to infinite, in such a way that the ratio $W/m$ tends to $\lambda x+(1-\lambda)[x]_{\epsilon,f}$, for some $x>0$ and $\lambda\in (0,1)$; (ii) there is an optimal strategy profile in which all the resource congestions are equal (i.e., equal to $W/m$); (iii) at the worst-case equilibrium there are two groups $A$ and $B$ of resources, such that $|A|/m$ and $|B|/m$ tend respectively to $\lambda$ and $1-\lambda$, as $m$ tends to infinite; (iv) at the worst-case equilibrium, each resource of group $A$ has congestion $x>0$, and each resource of group $B$ has congestion $[x]_{\epsilon,f}\leq x$, that can be equivalently defined as the minimum possible congestion guaranteeing the equilibrium conditions (while keeping the congestions in $A$ equal to $x$). According to the above properties, the equilibrium and the optimal cost increase (with respect to $m$) as $\lambda m x f(x)+(1-\lambda)m [x]_{\epsilon,f} f([x]_{\epsilon,f})$ and $(\lambda x+(1-\lambda)[x]_{\epsilon,f})m$, respectively. Thus, we get that the approximate price of anarchy is equal to the the highest value of $\frac{\lambda xf(x)+(1-\lambda)[x]_{\epsilon,f} f([x]_{\epsilon,f})}{(\lambda x+(1-\lambda)[x]_{\epsilon,f}) f(\lambda x+(1-\lambda)[x]_{\epsilon,f})}$, over $x>0$ and $\lambda\in (0,1)$.
\begin{theorem}\label{thm5}
Let $f:\RP\rightarrow \RP$ be a semi-convex latency function, and let $\epsilon\geq 0$. We have that 
\begin{equation}\label{eq-upp}
 \poa_{\epsilon}({\sf WSLB}(\{f\}))\leq\sup_{x>0}\max_{\lambda\in (0,1)}\gamma_{\epsilon,f}(x,\lambda).
\end{equation}
\end{theorem}
\begin{theorem}[Lower Bound]\label{thm6}
For any $\epsilon\geq 0$ and for any $x\geq 0$, let $\lambda^*(x)\in \arg\max_{\lambda\in (0,1)}\gamma_{\epsilon,f}(x,\lambda)$. If (i) $\lambda^*(x)\leq \frac{1}{2}$, and (ii) either $\epsilon=0$, or $\lambda^*(x) x+(1-\lambda^*(x))[x]_{\epsilon,f}-x/2\geq 0$, for any $x>0$, then
\begin{equation}\label{fine}
\poa_{\epsilon}({\sf WSLB}(\{f\}))\geq \sup_{x>0}\max_{\lambda\in (0,1)}\gamma_{\epsilon,f}(x,\lambda).
\end{equation}
\end{theorem}
By exploiting (\ref{fine}), one can obtain tight bounds on the price of anarchy of weighted symmetric load balancing games with identical resources having polynomial latency functions. The same tight bounds have been given in \cite{LMMR08,GLMM06} for monomial latency functions. In the following corollary of Theorems~\ref{thm5} and \ref{thm6} we show that the same bounds hold for more general polynomial latency functions. 
\begin{corollary}\label{thmpoly}
Let $\mathcal{P}(d)$ be the class of polynomial latency functions of maximum degree $d$. Then,
$\poa_0({\sf WSLB}(\mathcal{P}(d)))=\frac{d^d(2^{d+1}-1)^{d+1}}{2^d(d+1)^{d+1}(2^d-1)^d}\in \Theta\left(\frac{2^d}{d}\right)$.
\end{corollary}
In Table \ref{figura1}, we show a comparison between the cases of general and identical resources with respect to the price of anarchy for games with polynomial latency functions.
\begin{figure}[!t]
\begin{center}\footnotesize
\begin{tabular}{|c|c|c||c|c|c|}
  \hline
  $d$ & Identical& General & $d$ & Identical & General\\\hline
  1 & 1.125 & 2.618 & 6 & 7.544 & 14,099\\
  2 & 1.412 & 9.909 & 7 & 12.866 & 118,926\\
  3 & 1.946 & 47.82 & 8 & 22.478 & 1,101,126\\
  4 & 2.895 & 277 & $\vdots$ & $\vdots$ & $\vdots$\\
  5 & 4.571 & 1,858 & $\infty$ & $\Theta\left(\frac{2^d}{d}\right)$ & $\left(\Theta\left(\frac{d}{\log d}\right)\right)^{d+1}$\\
  \hline
\end{tabular}

\caption{The price of anarchy of weighted symmetric load balancing games
with polynomial latency functions of maximum degree $d$: a comparison
between the cases of identical and general
resources.}\label{figura1}
\end{center}
\end{figure}
As in Corollary \ref{thmpoly}, the general upper bound provided in Theorem \ref{thm5} is well-suited to derive (possibly non-tight) upper bounds on the $\epsilon$-approximate price of anarchy of symmetric polynomial load balancing games with identical resources for any $\epsilon>0$. 
\subsection{Lower Bounds for Exact One-Round Walks}
The following construction gives a class of lower bounds for exact one-round walks generated by selfish/cooperative players in unweighted load balancing games with identical resources having latency function $f$. Fix $n\in\mathbb{N}$ and a sequence of integers $1=o_1\leq o_2\leq \ldots\leq  o_n$.  Let $E=E_0\supset E_1\supset E_2 \supset\ldots \supset E_n\supset E_{n+1}=\emptyset$ be a sequence of sets of resources such that $(|E_{i-1}|-|E_{i}|) o_i=|E_{i}|$ (observe that such a sequence exists). For any $i\in [n]$, we have $|E_i|$ players of type $i$ whose set of strategies is $E_{i-1}$. Suppose that players enter the game in non-decreasing order with respect to their type. One can easily prove that the strategy profile ${\bm \sigma}$ in which each player of type $i$ selects a different resource $e\in E_i$ is a possible outcome for an exact one-round walk generated by selfish/cooperative players. Consider the strategy profile in which, for any resource $e\in E_{i-1}\setminus E_i$, there are exactly $o_i$ players of type $i$ selecting $e$. We get
\begin{equation}\label{equalid}
{\sf CR}^s_0(\{f\})\geq \frac{{\sf SUM}({\bm\sigma})}{{\sf SUM}({\bm\sigma}^*)}=\frac{\sum_{i=1}^n(|E_i|-|E_{i+1}|)if(i)}{\sum_{i=1}^n(|E_{i-1}|-|E_i|)o_if(o_i)}.
\end{equation}
For linear latency functions, by using $n=10^{13}$ and $o_i=\left\lfloor \frac{ 44411}{100000}i+1+\left\lfloor\frac{\sqrt{i}}{7}\right\rfloor\right\rfloor$ in (\ref{equalid}), we get a lower bound of at least $4.0009$ which improves the currently known lower bound of 4 given in \cite{CFKKM11}. We conjecture that a tight class of lower bounding instances for linear and more general polynomial latency functions is given by the union of all the instances described above, over all values of $n\in\mathbb{N}$ and all sequences $(o_i)_{i\in [n]}$.

\section{Open Problems and Research Directions}
We have investigated how the combinatorial structure of the players' strategy space impacts on the efficiency of some decentralized solutions in congestion games by focusing on the simplest possible situation: that of singleton strategies. All of our negative results clearly carry over to more general structures, such as in matroid congestion games and in network congestion games.

Our work leaves two main open problems. The first is to understand whether better performance are possible for approximate one-round walks in weighted symmetric load balancing games (we conjecture this is not the case), while the second is to give upper bounds on the performance of one-round walks in weighted and unweighted load balancing games with identical resources. Relatively to further research directions, we believe that the modus-operandi considered in this work can be efficiently used to find tight lower bounds in several variants of congestion games or scheduling games, and with respect to different solution concepts. For instance, after the appearance of the conference version of this work, similar ideas have been efficiently reused in \cite{BMV18,CGV17,V18,BiloV20,BMMV20,BMPV20,BV20sagt} to derive tight lower bounds on the efficiency of some variants of congestion games and load balancing games. 

\bibliographystyle{spmpsci}      
\bibliography{Bibliography}

\appendix
\section{Connection with The Smoothness Framework}\label{frame}
The smoothness analysis \cite{R15} has been introduced to give almost tight upper bounds on the price of anarchy of several games, and applies successfully to the case of congestion games. In particular, an atomic congestion game $\CG$ is {\em $(\lambda,\mu)$-smooth} if, for any pair of strategy profiles $\sg,\sg^*$ of $\CG$, the following condition holds:
\begin{align}
\sum_{i\in \N}cost_i(\sg_{-i},\sg_i^*)\leq \lambda\cdot {\sf SUM}(\sg^*)+\mu\cdot {\sf SUM}(\sg).
\end{align}
A class $\mathcal{G}$ of congestions games is $(\lambda,\mu)$-smooth if any game $\CG\in \mathcal{G}$ is $(\lambda,\mu)$-smooth. The following results give a characterization of the  price of anarchy of atomic congestion games with respect to the social function ${\sf SUM}$.
\begin{theorem}[\cite{R15}]\label{thmsmoot}
(i) Given a $(\lambda,\mu)$-smooth class $\mathcal{G}$ of (weighted or unweighted) congestion games with $\mu\in [0,1]$ and $\lambda\geq 0$, then $\poa(\mathcal{G})\leq \frac{\lambda}{1-\mu}$. (ii) If $\mathcal{G}={\sf U}(\mathcal{C})$ for some class $\mathcal{C}$ of latency functions, we have a tight characterization of the price of anarchy: 
\begin{align}
&\poa(\mathcal{G})\nonumber\\
&=\inf\left\{\frac{\lambda}{1-\mu}:\mathcal{G}\text{ is $(\lambda,\mu)$-smooth}\right\}\nonumber\\
&=\inf_{\mu\in [0,1),\lambda\geq 0}\left\{\frac{\lambda}{1-\mu}:of(k+1)\leq \lambda of(o)+\mu kf(k), \forall f\in\mathcal{C},k\in\Z,o\in\Z\right\}.\label{formsmoothunw}
\end{align}
\end{theorem}
One can easily see that, by setting $\gamma:=\frac{\lambda}{1-\mu}$ and $x:=\frac{1}{1-\mu}$ in the upper bound of Theorem~\ref{thmsmoot}, one can reobtain the upper bounds for (exact) Nash equilibria provided in Lemma~\ref{lemupp} and Lemma~\ref{lemupp2} for weighted and unweighted games, respectively. Furthermore, the approach considered in \cite{R15} to show the tight bound provided in (\ref{formsmoothunw}), resorts to a technical lemma similar to Lemma~\ref{lemrou2}. 

A tight characterization of the price of anarchy for weighted congestion games holds under mild assumptions on the considered latency functions. 
\begin{theorem}[\cite{BGR10}]\label{thmsmoot2}
Given a class of weighted congestion games $\mathcal{G}={\sf W}(\mathcal{C})$, where $\mathcal{C}$ is closed under abscissa scaling, the following characterization of the price of anarchy holds:
\begin{align}
&\poa({\sf W}(\mathcal{C}))\nonumber\\
&=\inf\left\{\frac{\lambda}{1-\mu}:{\sf W}(\mathcal{C})\text{ is $(\lambda,\mu)$-smooth}\right\}\nonumber\\
&=\inf_{\mu\in [0,1),\lambda\geq 0}\left\{\frac{\lambda}{1-\mu}:of(k+o)\leq \lambda of(o)+\mu kf(k), \forall f\in\mathcal{C},k\geq 0,o\geq 0\right\}\label{formsmoothw}
\end{align}
\end{theorem}
The approach considered in \cite{BGR10} to show the tight bound provided in (\ref{formsmoothw}) resorts to a technical lemma similar to Lemma~\ref{lemrou1}. 
\section{Missing Proofs from Section~\ref{secweigh}}\label{missing}
\subsection{Proof of Lemma~\ref{lemupp}.}
We prove the lemma for ${\sf EM}=\poa_{\epsilon}$ by using the primal-dual approach  \cite{B12}, and we give a sketch for the other cases, since the proof is analogue. Let ${\sf CG}\in\mathcal{G}$, and let $\bm\sigma$ and $\bm\sigma^*$ be an $\epsilon$-approximate pure Nash equilibrium and an optimal strategy profile, respectively. 
We have that the maximum value of the following linear program in the variables $\alpha_e$'s is an upper bound on $\poa_\epsilon({\sf CG})$ (recall that $w_i$ is the weight of player $i$, $k_e$ and $o_e$ are the equilibrium and optimal congestion of resource $e$, respectively):\\
\begin{align}
{\sf LP1:}\quad \max\quad& \overbrace{\sum_{e\in E}\alpha_ek_e\ell_e(k_e)}^{{\sf SUM}(\bm\sigma)}\nonumber\\
s.t. \quad& \sum_{e\in \sigma_i}\alpha_e \ell_e(k_e)\leq (1+\epsilon)\sum_{e\in \sigma_i^*}\alpha_e \ell_e(k_e+w_i),\quad \forall i\in \N\label{const1}\\
&\overbrace{\sum_{e\in E}\alpha_eo_e\ell_e(o_e)}^{{\sf SUM}(\bm\sigma^*)}= 1\label{const1b}\\
&\alpha_e\geq 0,\quad\forall e\in E.\nonumber
\end{align}
Indeed, by setting $\alpha_e=1$ for any $e\in E$, we have that: (i) the objective function is the social cost at the equilibrium; (ii) the constraints in  (\ref{const1}) impose some relaxed $\epsilon$-approximate pure Nash equilibrium conditions (ensuring that each agent, at the equilibrium, does not get any benefit when deviating in favour of strategy $\sigma^*_i$); (iii) (\ref{const1b}) is the normalized optimal social cost (normalization is possible since there is some $o_e>0$, that implies $o_e\ell_e(o_e)>0$). 

By upper bounding all the $w_i$'s in (\ref{const1}) with $o_e$ (this is possible since $w_i\leq o_e$ for any $i\in N$ and $e\in E$ such that $e\in \sigma^*_i$) and by summing all the resulting constraints once scaled by a factor $w_i$, we obtain the following relaxation of ${\sf LP1}$:
\begin{align}
{\sf LP 2:}\quad \max\quad& \sum_{e\in E}\alpha_ek_e\ell_e(k_e)\nonumber\\
s.t. \quad &\sum_{e\in E}\alpha_e\beta_{\sf W}({\sf \poa_\epsilon},k_e,o_e,\ell_e)\geq 0\label{const2}\\
&\sum_{e\in E}\alpha_e o_e\ell_e(o_e)= 1\label{const2b}\\
&\alpha_e\geq 0,\quad\forall e\in E,\nonumber
\end{align}
where $\beta_{\sf W}({\sf \poa_\epsilon},k_e,o_e,\ell_e)$ is defined as in \eqref{def_alpha}. Indeed:
\begin{equation}
\sum_{i\in \N}w_i\sum_{e\in \sigma_i}\alpha_e\ell_e(k_e)=\sum_{e\in E}\sum_{i\in N:e\in \sigma_i}\alpha_ew_i\ell_e(k_e)=\sum_{e\in E}\alpha_ek_e\ell_e(k_e)\nonumber
\end{equation}
and 
\begin{align}
\sum_{i\in \N}w_i\sum_{e\in \sigma_i}\alpha_e\ell_e(k_e+w_i)
&\leq \sum_{i\in \N}w_i\sum_{e\in \sigma_i}\alpha_e\ell_e(k_e+o_e)\nonumber\\
&= \sum_{e\in E}\sum_{i\in N:e\in \sigma_i^*}\alpha_ew_i\ell_e(k_e+o_e)\nonumber\\
&=\sum_{e\in E}\alpha_eo_e\ell_e(k_e+o_e).\nonumber
\end{align}
Then, we get inequality $\sum_{e\in E}\alpha_e k_e\ell_e(k_e)\leq (1+\epsilon)\sum_{e\in E}\alpha_e o_e\ell_e(k_e+o_e)$, that is equivalent to constraint  \eqref{const2}. 
Consider the dual program of ${\sf LP2}$ in the dual variables $x$ and $\gamma$ (respectively associated to the dual constraints (\ref{const2}) and (\ref{const2b})):
\begin{align}
{\sf DLP:}\quad \min\quad & \gamma \nonumber\\
s.t.\quad & \gamma\cdot o_e\ell_e(o_e) \geq k_e\ell_e(k_e)+x\cdot \beta_{\sf W}({\sf PoA}_\epsilon,k_e,o_e,\ell_e),\quad \forall e\in E\label{constdual}\\
&x\geq 0,\gamma\in \R.\nonumber
\end{align}
In ${\sf DLP}$, one can assume that $x\geq 1$, otherwise, if $x<1$ and $o_e=0$ for some $e\in E$, the dual constraint in (\ref{constdual}) related to $e$ is not feasible. Now, one can assume that $o_e>0$, since the constraint (\ref{constdual}) corresponding to resource $e$ is always verified if $o_e=0$, thus we can remove such constraint. Let $\gamma^*$ be the optimal value of ${\sf DLP}$. By the Weak Duality Theorem, we get $\gamma^*\geq \poa_\epsilon(\mathcal{G})$. By (\ref{constdual}) and the optimality of $\gamma^*$, we get
\begin{equation}
\gamma^*=\min_{x\geq 1}\max_{e\in E}\gamma_{\sf W}(\poa_\epsilon,x,k_e,o_e,\ell_e)\leq \inf_{x\geq 1}\sup_{\substack{f\in\mathcal{C}(\mathcal{G}),\\k> 0, o>0}}\gamma_{\sf W}({\sf EM},x,k,o,f)\footnote{By exploiting the continuity of function $-kf(k)$ with respect to $k\geq 0$, we have that the supremum (appearing in the inequalities) can be equivalently taken over $k>0$ instead of $k\geq 0$, without decreasing its value.}=\gamma_{\sf W}(\poa_{\epsilon},\mathcal{G}),
\end{equation}
thus showing the claim for ${\sf EM}=\poa_{\epsilon}.$

For ${\sf EM}={\sf CR}^s_{\epsilon}$, one can use a similar proof as in the previous case. Let $\bm \tau=({\bm \sigma}^0,{\bm \sigma}^1,\ldots,{\bm \sigma}^n)$ be the $\epsilon$-approximate one-round walk involving selfish players that returns a strategy profile ${\bm\sigma}={\bm \sigma}^n$. We  consider the linear program {\sf LP1}, but with constraint 
\begin{equation}\label{oneroundineq}
\sum_{e\in \sigma_i}\alpha_e \ell_e(k_e({\bm \sigma }^i))\leq (1+\epsilon)\sum_{e\in \sigma_i^*}\alpha_e \ell_e(k_e({\bm \sigma}^{i-1})+w_i)
\end{equation}
 in place of (\ref{const1}), as constraint \eqref{oneroundineq} is necessarily satisfied by the greedy choice of each selfish player $i\in \N$ (when $\alpha_e=1$ for any $e\in E$); thus, by exploiting the same arguments as in the previous case, the maximum value of {\sf LP1} is an upper bound on ${\sf CR}^s_{\epsilon}(\CG)$. We have that 
\begin{align}
\sum_{i\in \N}w_i\sum_{e\in \sigma_i}\alpha_e \ell_e(k_e({\bm \sigma }^i))
&=\sum_{e\in E}\sum_{i:e\in \sigma_i}w_i\alpha_e \ell_e(k_e({\bm \sigma }^i))\nonumber \\
&=\sum_{e\in E}\sum_{i:e\in \sigma_i}\overbrace{\left(\int_{t=0}^{k_e}\chi_{[k_e({\bm \sigma }^{i-1}),k_e({\bm \sigma }^{i})]}(t)\text{dt}\right)}^{w_i}\alpha_e \ell_e(k_e({\bm \sigma }^i))\nonumber \\
& = \sum_{e\in E}\alpha_e\int_{t=0}^{k_e}\left(\sum_{i:e\in \sigma_i}\ell_e(k_e({\bm \sigma }^i))\chi_{[k_e({\bm \sigma }^{i-1}),k_e({\bm \sigma }^{i})]}(t)\right)\text{dt}\nonumber \\
& = \sum_{e\in E}\alpha_e\int_{t=0}^{k_e}\left(\sum_{i\in \N}\ell_e(k_e({\bm \sigma }^i))\chi_{[k_e({\bm \sigma }^{i-1}),k_e({\bm \sigma }^{i})]}(t)\right)\text{dt}\label{oneround_ineqq}\\
&\geq \sum_{e\in E}\alpha_e\int_{t=0}^{k_e}\left(\sum_{i\in \N}\ell_e(t)\chi_{[k_e({\bm \sigma }^{i-1}),k_e({\bm \sigma }^{i})]}(t)\right)\text{dt}\nonumber \\
&= \sum_{e\in E}\alpha_e \int_{t=0}^{k_e}\ell_e(t)\text{dt},\label{oneround_ineqq1}
\end{align}
where $\chi$ denotes the indicator function (i.e., $\chi_A(t)=1$ if $t\in A$, $\chi_A(t)=0$ otherwise), and \eqref{oneround_ineqq} holds since $\chi_{[k_e({\bm \sigma }^{i-1}),k_e({\bm \sigma }^{i})]}(t)$ is null if $e\notin \sigma_i$; furthermore, we have that 
\begin{align}
\sum_{i\in \N}w_i\sum_{e\in \sigma_i^*}\alpha_e \ell_e(k_e(\sg^{i-1})+w_i)&\leq  \sum_{i\in \N}w_i\sum_{e\in \sigma_i^*}\alpha_e \ell_e(k_e+o_e)=\sum_{e\in E}\alpha_e o_e\ell_e(k_e+o_e).\label{oneround_ineqq2}
\end{align}
Thus, by \eqref{oneroundineq}, \eqref{oneround_ineqq1}, and \eqref{oneround_ineqq2}, we get 
\begin{align*}
\sum_{e\in E}\alpha_e \int_{t=0}^{k_e}\ell_e(t)\text{dt}
&\leq \sum_{i\in \N}w_i\sum_{e\in \sigma_i}\alpha_e \ell_e(k_e({\bm \sigma }^i))\\
&\leq \sum_{i\in \N}w_i(1+\epsilon)\sum_{e\in \sigma_i^*}\alpha_e \ell_e(k_e(\sg^{i-1})+w_i)\\
&\leq \sum_{e\in E}\alpha_e (1+\epsilon)o_e\ell_e(k_e+o_e).
\end{align*}
Such inequality is equivalent to constraint (\ref{const2}), but with ${\sf CR}_\epsilon^s$ in place of $\poa_\epsilon$; thus, we can continue the analysis from a linear program analogue to {\sf LP2}, and by proceeding as in the case of the approximate price of anarchy, we can show the claim for $\EM={\sf CR}^s_{\epsilon}$. 

Finally, we consider the case $\EM={\sf CR}^c_{\epsilon}$. Let $\bm \tau=({\bm \sigma}^0,{\bm \sigma}^1,\ldots,{\bm \sigma}^n)$ be the $\epsilon$-approximate one-round walk involving cooperative players that returns a strategy profile ${\bm\sigma}={\bm \sigma}^n$. Let $\hat{\ell}_e$ denote the function such that $\hat{\ell}_e(x)=x\ell_e(x)$ for any $x\geq 0$. By hypothesis, $\hat{\ell}_e$ is a convex function for any $e\in E$. Again, we can consider the linear program {\sf LP1}, but with constraint 
\begin{equation}\label{coopineq}
\sum_{e\in E}\alpha_e(\hat{\ell}_e(k_e(\sg^i))-\hat{\ell}_e(k_e(\sg^{i-1})))\leq (1+\epsilon)\sum_{e\in E}\alpha_e(\hat{\ell}_e(k_e(\sg^{i-1}_{-i},\sigma_i^*))-\hat{\ell}_e(k_e(\sg^{i-1})))
\end{equation} 
in place of (\ref{const1}), as \eqref{coopineq} is necessarily satisfied by the greedy choice of each cooperative player $i\in \N$ (when $\alpha_e=1$ for any $e\in E$); thus, the maximum value of {\sf LP1} is an upper bound on ${\sf CR}^c_{\epsilon}(\CG)$. We have that
\begin{align}
&\sum_{e\in E}\alpha_e k_e\ell_e(k_e)\nonumber\\
&=\sum_{e\in E}\hat{\ell}_e(k_e)\nonumber\\
&=\sum_{e\in E}\alpha_e\sum_{i\in \N}(\hat{\ell}_e(k_e(\sg^i))-\hat{\ell}_e(k_e(\sg^{i-1})))\nonumber\\
&=\sum_{i\in \N}\sum_{e\in E}\alpha_e(\hat{\ell}_e(k_e(\sg^i))-\hat{\ell}_e(k_e(\sg^{i-1})))\nonumber\\
&\leq (1+\epsilon)\sum_{e\in E}\alpha_e(\hat{\ell}_e(k_e(\sg^{i-1}_{-i},\sigma_i^*))-\hat{\ell}_e(k_e(\sg^{i-1})))\label{coop2}\\
&\leq\sum_{i\in \N}(1+\epsilon)\sum_{e\in E}\alpha_e(\hat{\ell}_e(k_e(\sg^{i-1})+w_i\cdot \chi_{\sigma_i^*}(e))-\hat{\ell}_e(k_e(\sg^{i-1})))\nonumber\\
&=\sum_{i\in \N}(1+\epsilon)\sum_{e\in \sigma_i^*}\alpha_e(\hat{\ell}_e(k_e(\sg^{i-1})+w_i)-\hat{\ell}_e(k_e(\sg^{i-1})))\nonumber\\
&=\sum_{e\in E}\alpha_e\sum_{i\in \N:e\in \sigma^*_i}(1+\epsilon)\left(\hat{\ell}_e(k_e({\bm \sigma}^{i-1})+w_i)-\hat{\ell}_e(k_e({\bm \sigma}^{i-1}))\right)\nonumber\\
&\leq \sum_{e\in E}\alpha_e\sum_{i\in \N:e\in \sigma^*_i}(1+\epsilon)\left(\hat{\ell}_e(k_e({\bm \sigma})+w_i)-\hat{\ell}_e(k_e({\bm \sigma}))\right)\label{coopconv-}\\
&\leq \sum_{e\in E}\alpha_e\sum_{i\in \N:e\in \sigma^*_i}(1+\epsilon)\left(\hat{\ell}_e\left(k_e(\sg)+\sum_{h\in [i]:e\in\sigma_h^*}w_h\right)-\hat{\ell}_e\left(k_e(\sg)+\sum_{h\in[i-1]:e\in\sigma_h^*}w_h\right)\right)\label{coopconv}\\
&=\sum_{e\in E}\alpha_e(1+\epsilon)\left(\hat{\ell}_e\left(k_e(\sg)+\sum_{h\in \N:e\in\sigma_h^*}w_h\right)-\hat{\ell}_e\left(k_e(\sg)\right)\right)\label{coopconv+}\\
&=\sum_{e\in E}\alpha_e(1+\epsilon)\left((k_e+o_e)\ell_e(k_e+o_e)-k_e\ell_e(k_e)\right),\nonumber
\end{align}
where (\ref{coop2}) comes from (\ref{coopineq}), (\ref{coopconv-}) and (\ref{coopconv}) holds because of the convexity of the functions $\hat{\ell}_e$'s (as $F(x+y+z)-F(x+y)\geq F(y+z)-F(y)$ holds for any convex function $F$, and for any $x,y,z\in \RP$), and (\ref{coopconv+}) holds since the second sum in (\ref{coopconv}) is telescoping. Thus, inequality $$\sum_{e\in E}\alpha_e k_e\ell_e(k_e)\leq \sum_{e\in E}\alpha_e (1+\epsilon)((k_e+o_e)\ell_e(k_e+o_e)-k_e\ell_e(k_e))$$ holds. Again, this inequality is equivalent to constraint \eqref{const2}, but with ${\sf CR}_\epsilon^c$ in place of $\poa_\epsilon$; thus, we can continue the analysis from a linear program analogue to {\sf LP2}, and by proceeding as in the case of the approximate price of anarchy, we can show the claim for $\EM={\sf CR}^c_{\epsilon}$.
\subsection{Proof of Lemma~\ref{lemrou1}.}
Fix $M<\gamma_{\sf W}({\sf EM},\mathcal{G})$, where ${\sf EM}\in \{\poa_{\epsilon},{\sf CR}^s_\epsilon,{\sf CR}^c_\epsilon\}$ is the considered efficiency metric. 

First of all, we assume that, for any latency function $f\in\mathcal{C}(\mathcal{G})$ and for any $o>0$, there exists $k>0$ such that $\beta_{\sf W}({\sf EM},k',o,f)<0$ for any $k'>k$. Indeed, if this is not the case, there exist $o>0$, a latency function $f$, and a sufficiently large $k> 0$, such that Case 1  of the claim is verified by tuple $(k,o,f)$, where $f$ can be chosen as any non-constant function. 

Given ${\sf EM}\in \{{\sf CR}^s_\epsilon,{\sf CR}^c_\epsilon\}$, define
\begin{align*}
&U_{\geq }:=\left\{(k,o,f):k> 0, o>0, f\in\mathcal{C}(\mathcal{G}),\beta_{\sf W}({\sf EM},k,o,f)\geq 0\right\},\\
&U_{<}:=\left\{(k,o,f):k> 0, o>0, f\in\mathcal{C}(\mathcal{G}),\beta_{\sf W}({\sf EM},k,o,f)< 0\right\},\\
&\gamma_\geq(x):=\sup_{(k,o,f)\in U_\geq}\gamma_{\sf W}({\sf EM},x,k,o,f)\in \RP\cup\{\infty\}\\
&\gamma_<(x):=\sup_{(k,o,f)\in U_<}\gamma_{\sf W}({\sf EM},x,k,o,f)\in \R\cup\{\infty\}.
\end{align*}
One can easily observe that $U_{\geq 0}$ is non-empty, and because of our previous assumption, $U_<$ is non-empty, too. If ${\sf EM}=\poa_\epsilon$, we consider the same definitions $\gamma_\geq(x)$ and $\gamma_<(x)$, but we assume that the supremum is restricted to the set of non-constant functions only. By definition of $\gamma_{\sf W}({\sf EM},\mathcal{G})$, we have that 
\begin{equation}\label{form_appe0}
M<\gamma_{\sf W}({\sf EM},\mathcal{G})\leq \inf_{x\geq 1}\max\{\gamma_\geq(x),\gamma_<(x)\}=\inf_{x> 0}\max\{\gamma_\geq(x),\gamma_<(x)\}\footnote{For $x<1$, one can easily show that $\gamma_\geq(x)=\infty$. Thus, even if the infimum appearing in the inequality is calculated over $x>0$, its value does not decrease (with respect to the infimum over $x\geq 1$).}
\end{equation} 
for any $\{{\sf CR}^s_\epsilon,{\sf CR}^c_\epsilon\}$. Furthermore, we have that $\gamma_{\sf W}({\sf EM},\mathcal{G})\leq \inf_{x> 0}\max\{\gamma_\geq(x),\gamma_<(x)\}$ holds even for ${\sf EM}=\poa_\epsilon$. Indeed, for any $x> 0$, one can easily show that quantity $\gamma_{\sf W}(\poa_\epsilon,x,k,o,f)$ is maximized by a non-constant latency function~$f$. Thus, even if the supremum in $\gamma_<(x)$ and $\gamma_\geq(x)$ is calculated over the non-constant latency functions, inequality \eqref{form_appe0} holds as well. 

We observe that $\gamma_{\sf W}({\sf EM},x,k,o,f)$ is non-decreasing (resp. non-increasing) in $x\geq 0$ if $(k,o,f)\in U_\geq$ (resp. $(k,o,f)\in U_<$). Thus, we have the following remark:
\begin{remark}\label{remaincgam}
$\gamma_\geq(x)$ and $\gamma_<(x)$ are non-decreasing and non-increasing in $x\geq 0$, respectively.\qed
\end{remark}
Let $x^*:=\max\{0,\sup\{x> 0:\gamma_\geq(x)< \gamma_<(x)\}\}\in \RP\cup\{\infty\}.$ By exploiting \eqref{form_appe0}, the definition of $x^*$, and Remark \ref{remaincgam}, we get 
\begin{equation}\label{form_appe}
M<\gamma_{\sf W}({\sf EM},\mathcal{G})\leq \inf_{x> 0}\max\{\gamma_\geq(x),\gamma_<(x)\}\leq \min\left\{\lim_{x\rightarrow x^{*-}}\gamma_<(x),\lim_{x\rightarrow x^{*+}}\gamma_\geq (x)\right\},
\end{equation}
in which the left-hand (resp. right-hand) limit is not considered if $x^*=0$ (resp. $x^*=\infty$). In the following, we distinguish between the cases $x^*\in (0,\infty)$, $x^*=0$, and $x^*=\infty$. 

We first assume that $x^*\in (0,\infty)$, and to show the claim in such a case, we will proceed as follows:
\begin{itemize}
\item[$\bullet$] {\bf Step 1:} We show that there exist two tuples $(k_1,o_1,f_1)\in U_<$ and $(k_2,o_2,f_2)\in U_\geq$ such that $\overline{\gamma}:=\gamma_{\sf W}({\sf EM},\overline{x},k_1,o_1,f_1)=\gamma_{\sf W}({\sf EM},\overline{x},k_2,o_2,f_2)>~M$ for some $\overline{x}>0$.
\item[$\bullet$] {\bf Step 2:} We show that a dual program as {\sf DLP} in (the proof of) Lemma~\ref{lemupp}, having two constraints only (except to those imposing the non-negativity of the variables), associated, respectively, to tuples $(k_1,o_1,f_1)$ and $(k_2,o_2,f_2)$, has an optimal value higher than $M$. 
\item[$\bullet$] {\bf Step 3:}  By taking the dual of the above dual program, we obtain a primal program as ${\sf LP2}$ in the proof of Lemma~\ref{lemupp}, but with two constraints and two variables only, and whose optimal value is higher than $M$  (by the Strong Duality Theorem); the claim of Lemma~\ref{lemrou1} is obtained by characterizing the optimal solution of this primal program. 
\end{itemize} 
\paragraph*{Step 1:$\ $}To prove Step 1, we start with the following lemma:
\begin{lemma}\label{lem_appe}
We have that $\lim_{x\rightarrow \hat{x}^-}\gamma_\geq(x)=\gamma_\geq (\hat{x})=\lim_{x\rightarrow \hat{x}^+}\gamma_\geq(x)$ for any $\hat{x}>0$.\footnote{Equivalently, Lemma \ref{lem_appe} states that $\gamma_{\geq}(\hat{x})$, as function from $\RPP$ to $\RP\cup\{\infty\}$, is continuous in $\hat{x}>0$.}
\end{lemma}
\begin{proof}
Fix $\hat{x}>0$. If $\lim_{x\rightarrow \hat{x}^-}\gamma_\geq(x)=\infty$, we also get $\gamma_\geq(\hat{x})=\lim_{x\rightarrow \hat{x}+}\gamma_\geq(x)=~\infty$ (as $\gamma_\geq (x)$ is non-decreasing in $x$) and the claim follows. If $\lim_{x\rightarrow \hat{x}^-}\gamma_\geq(x)<\infty$, we necessarily have that there exists $c\geq 0$ such that $\frac{\beta_{\sf W}(\EM,k,o,f)}{of(o)}\leq c$ for any $(k,o,f)\in U_\geq$. Indeed, if it were not the case, we would have that $$\gamma_\geq(x)=\sup_{(k,o,f)\in U_\geq}\left(\frac{kf(k)}{of(o)}+\frac{x\beta_{\sf W}(\EM,k,o,f)}{of(o)}\right)\geq x\sup_{(k,o,f)\in U_\geq}\frac{\beta_{\sf W}(\EM,k,o,f)}{of(o)}=\infty,$$ for any $x>0$, and then we would be in case $\lim_{x\rightarrow \hat{x}^-}\gamma_\geq(x)=\infty$. Thus, for any sufficiently small $\delta>0$, we have that 
\begin{align*}
&\overbrace{\gamma_\geq(\hat{x}+\delta)}^{\geq \gamma_\geq(\hat{x}-\delta)}\\
&=\sup_{(k,o,f)\in U_\geq}\left(\frac{kf(k)+(\hat{x}-\delta)\beta_{\sf W}(\EM,k,o,f)}{of(o)}+\overbrace{\frac{2\delta\beta_{\sf W}(\EM,k,o,f)}{of(o)}}^{\geq 0}\right)\\
&\leq \sup_{(k,o,f)\in U_\geq}\left(\frac{kf(k)+(\hat{x}-\delta)\beta_{\sf W}(\EM,k,o,f)}{of(o)}\right)+2\delta\sup_{(k,o,f)\in U_\geq}\frac{\beta_{\sf W}(\EM,k,o,f)}{of(o)}\\
&= \gamma_\geq(\hat{x}-\delta)+2\delta\sup_{(k,o,f)\in U_\geq}\frac{\beta_{\sf W}(\EM,k,o,f)}{of(o)}\\
&\leq \gamma_\geq(\hat{x}-\delta)+2\delta c,
\end{align*}
that implies
\begin{equation}\label{disblem}
0\leq \gamma_\geq(\hat{x}+\delta)-\gamma_\geq(\hat{x}-\delta)\leq 2\delta c.
\end{equation}
By \eqref{disblem}, and by the arbitrariness of $\delta>0$, we necessarily have that $\lim_{x\rightarrow \hat{x}^-}\gamma_\geq(x)=\gamma_\geq (\hat{x})=\lim_{x\rightarrow \hat{x}^+}\gamma_\geq(x)$ (where the first equality holds since $\gamma_\geq$ is non-decreasing), and this shows the claim. \qed
\end{proof}
By applying Lemma \ref{lem_appe} to \eqref{form_appe}, we get that there exists a sufficiently small $\delta>0$ such that $M<\gamma_{\sf W}({\sf EM},\mathcal{G})\leq \gamma_\geq (x)$ for any $x>0$ with $|x-x^*|\leq \delta$; furthermore, by exploiting the definition of $x^*$, we have that $\gamma_<(x^*-\delta)>\gamma_\geq(x^*-\delta)$ and $\gamma_<(x^*+\delta/2)\leq \gamma_\geq (x^*+\delta/2)$. These facts, easily translate in the existence of two tuples $(k_1,o_1,f_1)\in U_<$ and $(k_2,o_2,f_2)\in U_\geq$ such that $\gamma_{\sf W}({\sf EM},x^*-\delta,k_1,o_1,f_1)>\gamma_{\sf W}({\sf EM},x^*-\delta,k_2,o_2,f_2)>M$, $\gamma_{\sf W}({\sf EM},x^*+\delta,k_1,o_1,f_1)< \gamma_{\sf W}({\sf EM},x^*+\delta,k_2,o_2,f_2)$,  and $\gamma_{\sf W}({\sf EM},x^*+\delta,k_2,o_2,f_2)>M$. Thus, we get $\overline{\gamma}:=\gamma_{\sf W}({\sf EM},\overline{x},k_1,o_1,f_1)=\gamma_{\sf W}({\sf EM},\overline{x},k_2,o_2,f_2)>~M$ for some $\overline{x}\in (x^*-\delta,x^*+\delta)$, and this concludes Step 1. 

\paragraph*{ Step 2:$\ $}
Consider a modification of the dual program {\sf DLP} of Lemma~\ref{lemupp}, in which there are two resources $e_1$ and $e_2$ only, whose characteristics are expressed respectively by the tuples $(k_1,o_1,f_1)$ and $(k_2,o_2,f_2)$ determined in Step 1. Let $\overline{{\sf DLP}}$ denote the new linear program in variables $x,\gamma$, that is explicitly defined as follows: 
\begin{align}
\overline{\sf DLP}:\quad \min\quad & \gamma \nonumber\\
s.t.\quad & \gamma o_1f_1(o_1) \geq k_1 f_1(k_1)+x\cdot \beta_{\sf W}({\sf EM},k_1,o_1,f_1),\nonumber\\
& \gamma o_2f_2(o_2) \geq k_2f_2(k_2)+x\cdot \beta_{\sf W}({\sf EM},k_2,o_2,f_2),\nonumber\\
&x\geq 0,\gamma \in \R.\nonumber
\end{align}
By exploiting the claim shown in Step 1, we get that $(\overline{x},\overline{\gamma})$ is the optimal solution of $\overline{\sf DLP}$, and since $\overline{\gamma}>M$, we conclude Step 2. 

\paragraph*{ Step 3:$\ $} The dual of $\overline{\sf DLP}$ is a linear program $\overline{\sf LP}$ similar to the program {\sf LP2} defined in the proof of Lemma \ref{lemupp}, but with resources $e_1$ and $e_2$ only. In particular, $\overline{\sf LP}$ is a linear program in variables $\alpha_1,\alpha_2$ explicitly defined as follows:
\begin{align}
\overline{\sf LP}:\quad \max\quad & \alpha_1k_1 f_1(k_1)+\alpha_2k_2 f_2(k_2)\nonumber\\
s.t. \quad &\alpha_1\beta_{\sf W}({\sf EM},k_1,k_1,f_1)+\alpha_2\beta_{\sf W}({\sf EM},k_2,o_2,f_2)\geq 0\label{const2appe}\\
& \alpha_1o_1 f_1(o_1)+\alpha_2o_2 f_2(o_2)= 1\label{const2bappe}\\
&\alpha_1,\alpha_2\geq 0\nonumber
\end{align}

By the Strong Duality Theorem, the value of the optimal solution of $\overline{\sf DLP}$ is equal to that of $\overline{\sf LP}$, that is higher than $M$. Furthermore, as the optimal solution $(\overline{x},\overline{\gamma})$ of $\overline{\sf DLP}$ is such that $\overline{x}>0$, by the complementary slackness conditions, constraint (\ref{const2appe}) is necessarily tight. Thus, the following optimization problem in variables $\alpha_1,\alpha_2$ is equivalent to $\overline{\sf LP}$:
\begin{align}
\overline{\sf OP}:\quad\max\quad & \frac{\alpha_1k_1 f_1(k_1)+\alpha_2k_2 f_2(k_2)}{\alpha_1o_1 f_1(o_1)+\alpha_2o_2 f_2(o_2)}\nonumber\\
s.t. \quad &\alpha_1\beta_{\sf W}({\sf EM},k_1,f_1,f_1)+\alpha_2\beta_{\sf W}({\sf EM},k_2,f_2,f_2)=0,\nonumber\\
&\alpha_1,\alpha_2\geq 0,\nonumber
\end{align}
and we observe that its optimal solution is obtained by setting
\begin{align*}
&\alpha_1:=\beta_{\sf W}({\sf EM},k_2,o_2,f_2)>0\\
&\alpha_2:=-\beta_{\sf W}({\sf EM},k_1,o_1,f_1)\geq 0.
\end{align*}
If $\alpha_2=0$, we are in Case 1 of the lemma, and if $\alpha_2>0$ we are in Case 2. Thus, the claim follows if $x^*\in (0,\infty)$

Now, we consider the case $x^*=0$, and we resort to similar proof arguments as in the above case. By our initial assumptions, we have that, for any latency function $f\in\mathcal{C}(\mathcal{G})$ and for any $o>0$, there exists $k>0$ such that $\beta_{\sf W}({\sf EM},k',o,f)<0$ for any $k'>k$. This shows that $\gamma_<(0)=\infty$. By applying Lemma \ref{lem_appe} to \eqref{form_appe}, we have that there exists a sufficiently small $\delta>0$ such that $\gamma_\geq (x)>M$ for any $x\in (0,\delta]$, and by exploiting the definition of $x^*$, we have that $\gamma_\geq (\delta/2)\geq \gamma_<(\delta/2)$. These facts translate into the existence of two tuples $(k_1,o_1,f_1)\in U_<$ and $(k_2,o_2,f_2)\in U_\geq$ such that $\gamma_{\sf W}({\sf EM},0,k_1,o_1,f_1)>\gamma_{\sf W}({\sf EM},0,k_2,o_2,f_2)>M$, $\gamma_{\sf W}({\sf EM},\delta,k_1,o_1,f_1)<\gamma_{\sf W}({\sf EM},\delta,k_2,o_2,f_2)$,  and $\gamma_{\sf W}({\sf EM},\delta,k_2,o_2,f_2)>M$. Thus, as in Step~1, we get $\overline{\gamma}:=\gamma_{\sf W}({\sf EM},\overline{x},k_1,o_1,f_1)=\gamma_{\sf W}({\sf EM},\overline{x},k_2,o_2,f_2)>~M$ for some $\overline{x}\in (0,\delta)$, and by proceeding as in Steps 2 and 3 of the above case, we can show the claim when $x^*=0$.

Finally, we consider the case $x^*=\infty$, i.e., when $\gamma_\geq(x)< \gamma_<(x)$ for any $x>0$. We observe that there exists a sufficiently large $\hat{x}>0$ such that $\gamma_\geq(\hat{x})>M$. Thus, there are two tuples $(k_1,o_1,f_1)\in U_<$ and $(k_2,o_2,f_2)\in U_\geq$ such that $\gamma_{\sf W}({\sf EM},\hat{x},k_1,o_1,f_1)>\gamma_{\sf W}({\sf EM},\hat{x},k_2,o_2,f_2)>M$. Then, there exists $\overline{x}>\hat{x}$ such that $\overline{\gamma}:=\gamma_{\sf W}({\sf EM},\overline{x},k_1,o_1,f_1)=\gamma_{\sf W}({\sf EM},\overline{x},k_2,o_2,f_2)>M$, and by proceeding as in Steps 2 and 3 of the above cases, we can show the claim when $x^*=\infty$.

%

\subsection{Proof of Lemma~\ref{equi_claim}.}
We consider an arbitrary player $z$ whose first strategy is a resource from a level $i$. Since the game is symmetric, we will show that: (i) if $i\in [2s-1]$ and the player deviates in favour of a resource from level $j=i+1$, her cost decreases exactly of a factor $1+\epsilon$; (ii) if $i\in [2s]$ and the player deviates in favour of a resource from level $j\leq i$, her cost does not decrease; (iii) if $i\in [2s-2]$ and the player deviates in favour of a resource from level $j>i+1$, for any sufficiently large $n$, her cost does not decrease. In the following, we prove separately the three different cases. Let ${\bm\sigma}^{j}_{s,n}$ denote the strategy profile obtained when player $z$, from level $i$, deviates in favour of a different resource at level $j$. 
\begin{description}
\item[(i)] Fix $j:=i+1$. If $i\in [s-1]$, we get 
\begin{align*}
cost_z(\bm\sigma_{s,n})
&=g_i(n\cdot w_i)\\
&=g_i\left(n\left(\frac{k_1}{n}\right)^i\right)\\
&=\theta_{1}^{i-1}f_1\left(\left(\frac{n}{k_1}\right)^{i-1}n\left(\frac{k_1}{n}\right)^i\right)\\
&=\theta_{1}^{i-1}f_1(k_1)\\
&=(1+\epsilon)\theta_{1}^{i}f_1(k_1+1)\\
&=(1+\epsilon)\theta_{1}^{i}f_1\left(\left(\frac{n}{k_1}\right)^{i}\left(n\left(\frac{k_1}{n}\right)^{i+1}+\left(\frac{k_1}{n}\right)^i\right)\right)\\
&=(1+\epsilon)g_{i+1}\left(n\left(\frac{k_1}{n}\right)^{i+1}+\left(\frac{k_1}{n}\right)^i\right)\\
&=g_{i+1}(n\cdot w_{i+1}+w_i)\\
&=(1+\epsilon)cost_z(\bm\sigma^{i+1}_{s,n});
\end{align*}
if $i=s$, we get 
\begin{align*}
cost_z(\bm\sigma_{s,n})
&=g_s\left(n\left(\frac{k_1}{n}\right)^s\right)\\
&=\theta_{1}^{s-1}f_1\left(\left(\frac{n}{k_1}\right)^{s-1}n\left(\frac{k_1}{n}\right)^s\right)\\
&=\theta_{1}^{s-1}f_1(k_1)\\
&=(1+\epsilon)\theta_{1}^{s-1}\theta_{1,2}f_2(k_2+1)\\
&=(1+\epsilon)\theta_{1}^{s-1}\theta_{1,2} f_2\left(\left(\frac{n}{k_1}\right)^{s}\left(n\left(\frac{k_1}{n}\right)^{s}\left(\frac{k_2}{n}\right)+\left(\frac{k_1}{n}\right)^s\right)\right)\\
&=(1+\epsilon)g_{s+1}\left(n\left(\frac{k_1}{n}\right)^{s}\left(\frac{k_2}{n}\right)+\left(\frac{k_1}{n}\right)^s\right)\\
&=(1+\epsilon)cost_z(\bm\sigma^{s+1}_{s,n});
\end{align*}
if $i\in [2s-1]_{s+1}$, we get 
\begin{align*}
&cost_z(\bm\sigma_{s,n})\\
&=g_i\left(n\left(\frac{k_1}{n}\right)^s\left(\frac{k_2}{n}\right)^{i-s}\right)
\\
&=\theta_{1}^{s-1}\theta_{1,2}\theta_{2}^{i-s-1}
f_2\left(\left(\frac{n}{k_1}\right)^{s}\left(\frac{n}{k_2}\right)^{i-s-1}n\left(\frac{k_1}{n}\right)^s\left(\frac{k_2}{n}\right)^{i-s}\right)\\
&=\theta_{1}^{s-1}\theta_{1,2}\theta_{2}^{i-s-1}f_2(k_2)\nonumber\\
&=(1+\epsilon)\theta_{1}^{s-1}\theta_{1,2}
\theta_{2}^{i-s}f_2(k_2+1)\\
&=(1+\epsilon)\theta_{1}^{s-1}\theta_{1,2}\theta_{2}^{i-s} f_2\nonumber\\
&\ \ \ \ \ \ \ \left(\left(\frac{n}{k_1}\right)^{s}\left(\frac{n}{k_2}\right)^{i-s}\left(n\left(\frac{k_1}{n}\right)^{s}\left(\frac{k_2}{n}\right)^{i+1-s}+\left(\frac{k_1}{n}\right)^s\left(\frac{k_2}{n}\right)^{i-s}\right)\right)\nonumber\\
&=(1+\epsilon)g_{i+1}\left(n\left(\frac{k_1}{n}\right)^{s}\left(\frac{k_2}{n}\right)^{i+1-s}+\left(\frac{k_1}{n}\right)^s\left(\frac{k_2}{n}\right)^{i-s}\right)\\
&=(1+\epsilon)cost_z(\bm\sigma^{i,i+1}_{s,n}).\nonumber
\end{align*}
\item[(ii)] Fix $j\leq i$. If $e$ is at level $i$ and $e'$ is at level $j$, we easily have that $g_e(k_e)\leq g_{e'}(k_e')$. Therefore $cost_z(\bm\sigma_{s,n})\leq (1+\epsilon)cost_z(\bm\sigma^{j}_{s,n})$.
\item[(iii)] Let $n$ be a sufficiently large integer such that
\begin{align}
&\max\{f_1(k_1),f_2(k_2)\}\nonumber\\
&\leq (1+\epsilon)\theta_1^{s-1}\theta_{1,2}\theta_2^{s-1}\min\left\{f_1\left(k_1+\frac{n}{\max\{k_1,k_2\}}\right),f_2\left(k_2+\frac{n}{\max\{k_1,k_2\}}\right)\right\}.\nonumber
\end{align}
Observe that $n$ exists since $f_1$ and $f_2$ are non-decreasing and unbounded. Informally, this value of $n$ is chosen in such a way that, when a player from level $i$ deviates in favor of a resource from the last level, her cost does not decrease, and one can easily see that this fact guarantees that $cost_z(\bm\sigma_{s,n})\leq (1+\epsilon)cost_z(\bm\sigma^{j}_{s,n})$ for any $j>i+1$ (and independently from index $i$). 
\end{description}
We conclude that $\bm\sigma_{s,n}$ is an $\epsilon$-approximate pure Nash equilibrium. 
\section{Missing Proofs from Section \ref{secunweigh}}
\subsection{Proof of Lemma \ref{lemupp2} (sketch).}
Let ${\sf EM}\in \{\poa_{\epsilon},{\sf CR}^s_\epsilon,{\sf CR}^c_\epsilon\}$. Let $\CG\in {\sf U}(\mathcal{G})$, and let $\sg$ and $\sg^*$ be an outcome satisfying the considered solution concept (e.g., approximate pure Nash equilibrium or approximate one-round walk), and an optimal strategy profile of $\CG$, respectively. As in Lemma \ref{lemupp}, we have that the optimal value of the following linear program in variables $\alpha_e$'s is an upper bound on the ${\sf EM}(\mathcal{G})$:
\begin{align}
{\sf LP:}\quad \max\quad& \sum_{e\in E}\alpha_ek_e\ell_e(k_e)\nonumber\\
s.t. \quad &\sum_{e\in E}\alpha_e\beta_{\sf U}({\sf \poa_\epsilon},k_e,o_e,\ell_e)\geq 0\label{const2unw}\\
&\sum_{e\in E}\alpha_e o_e\ell_e(o_e)= 1\label{const2bunw}\\
&\alpha_e\geq 0,\quad\forall e\in E,\nonumber
\end{align}
where $\beta_{\sf U}({\sf \poa_\epsilon},k_e,o_e,\ell_e)$ is defined as in \eqref{def_alpha_unw}. As in the weighted case analysed in Lemma~\ref{lemupp},  the objective function coincides with the social cost of outcome $\sg$, and constraint \eqref{const2bunw} is the normalized optimal social cost. Furthermore, we show that \eqref{const2bunw} is a constraint that is necessarily satisfied by the considered solution concept when $\alpha_e=1$ for any $e\in E$:
\begin{itemize}
\item[$\bullet$] If ${\sf EM}=\poa_\epsilon$, we have that $\sg$ is an $\epsilon$-approximate pure Nash equilibrium, thus we get $$\sum_{e\in \sigma_i}\alpha_e\ell_e(k_e)\leq \sum_{e\in \sigma_i^*}\alpha_e\ell_e(k_e+1)$$ for any $i\in \N$ (when $\alpha_e=1$ for any $e\in E$). By summing such inequality over all players $i\in \N$, we get constraint \eqref{const2unw}.
\item[$\bullet$] If ${\sf EM}={\sf CR}_\epsilon^s$, we have that $\sg$ is the strategy profile generated by some $\epsilon$-approximate one-round walk $\bm \tau=({\bm \sigma}^0,{\bm \sigma}^1,\ldots,{\bm \sigma}^n)$ involving selfish players. Thus we get $$\sum_{e\in \sigma_i}\alpha_e\ell_e(k_e(\sg^i))\leq (1+\epsilon)\sum_{e\in \sigma_i^*}\alpha_e\ell_e(k_e(\sg^{i-1})+1)$$ for any $i\in \N$ (when $\alpha_e=1$ for any $e\in E$). We have that $$\sum_{i\in \N}\sum_{e\in \sigma_i}\alpha_e\ell_e(k_e(\sg^i))=\sum_{e\in E}\alpha_e\sum_{i\in \N:e\in \sigma_i}\ell_e(k_e(\sg^i))=\sum_{e\in E}\alpha_e\sum_{h=1}^{k_e}\ell_e(h)$$ and $$\sum_{i\in \N}\sum_{e\in \sigma_i^*}\alpha_e\ell_e(k_e(\sg^{i-1})+1)\leq \sum_{i\in \N}\sum_{e\in \sigma^*_i}\alpha_e\ell_e(k_e+1)= \sum_{e\in E}\alpha_e o_e \ell_e(k_e+1).$$ Thus, by using the above inequalities we get 
\begin{align*}
\sum_{e\in E}\alpha_e\sum_{h=1}^{k_e}\ell_e(j)
&\leq \sum_{i\in \N}\sum_{e\in \sigma_i}\alpha_e\ell_e(k_e(\sg^i))\\
&\leq \sum_{i\in \N}(1+\epsilon)\sum_{e\in \sigma_i^*}\alpha_e\ell_e(k_e(\sg^{i-1})+1)\\
&\leq \sum_{e\in E}\alpha_e (1+\epsilon)o_e \ell_e(k_e+1),
\end{align*}
and this is equivalent to constraint \eqref{const2unw}. 
\item[$\bullet$] If ${\sf EM}={\sf CR}_\epsilon^c$, we have that $\sg$ is the strategy profile generated by some $\epsilon$-approximate one-round walk $\bm \tau=({\bm \sigma}^0,{\bm \sigma}^1,\ldots,{\bm \sigma}^n)$ involving cooperative players. Thus, denoting $x\ell_e(x)$ with $\hat{\ell}_e(x)$, we get $$\sum_{e\in E}\alpha_e(\hat{\ell}_e(k_e(\sg^i))-\hat{\ell}_e(k_e(\sg^{i-1})))\leq (1+\epsilon)\sum_{e\in E}\alpha_e(\hat{\ell}_e(k_e(\sg^{i-1})+ \chi_{\sigma_i^*}(e))-\hat{\ell}_e(k_e(\sg^{i-1})))$$ for any $i\in \N$ (when $\alpha_e=1$ for any $e\in E$), where $\chi$ denotes the indicator function. We get
\begin{align}
&\sum_{i\in \N}(1+\epsilon)\sum_{e\in E}\alpha_e(\hat{\ell}_e(k_e(\sg^{i-1})+ \chi_{\sigma_i^*}(e))-\hat{\ell}_e(k_e(\sg^{i-1})))\nonumber\\
&= \sum_{i\in \N}(1+\epsilon)\sum_{e\in \sigma_i^*}\alpha_e(\hat{\ell}_e(k_e(\sg^{i-1})+1)-\hat{\ell}_e(k_e(\sg^{i-1})))\nonumber\\
&\leq \sum_{i\in \N}(1+\epsilon)\sum_{e\in \sigma_i^*}\alpha_e(\hat{\ell}_e(k_e+1)-\hat{\ell}_e(k_e))\label{unwconv}\\
&=\sum_{e\in E}\alpha_e(1+\epsilon)o_e((k_e+1)\ell_e(k_e+1)-\ell_e(k_e)),\nonumber
\end{align}
where \eqref{unwconv} follows from the semi-convexity of the latency functions (i.e., the convexity of the functions $\hat{\ell}_e$'s). We conclude that
\begin{align*}
\sum_{e\in E}\alpha_e k_e\ell_e(k_e)
&=\sum_{e\in E}\alpha_e\sum_{i\in \N}(\hat{\ell}_e(k_e(\sg^i))-\hat{\ell}_e(k_e(\sg^{i-1})))\\
&=\sum_{i\in \N}\sum_{e\in E}\alpha_e(\hat{\ell}_e(k_e(\sg^i))-\hat{\ell}_e(k_e(\sg^{i-1})))\\
&\leq \sum_{i\in \N}(1+\epsilon)\sum_{e\in E}\alpha_e(\hat{\ell}_e(k_e(\sg^{i-1})+ \chi_{\sigma_i^*}(e))-\hat{\ell}_e(k_e(\sg^{i-1})))\\
&\leq \sum_{e\in E}\alpha_e(1+\epsilon)o_e((k_e+1)\ell_e(k_e+1)-\ell_e(k_e)),
\end{align*}
and such inequality is equivalent to constraint \eqref{const2unw}. 
\end{itemize}
We have shown that the optimal value of {\sf LP} is an upper bound on ${\sf EM}(\mathcal{G})$. Now, similarly as in the proof of Lemma \ref{lemupp2}, we proceed by analysing the dual of {\sf LP}, that is defined as follows:
\begin{align*}
{\sf DLP:}\quad \min\quad & \gamma \nonumber\\
s.t.\quad & \gamma\cdot o_e\ell_e(o_e) \geq k_e\ell_e(k_e)+x\cdot \beta_{\sf U}({\sf EM},k_e,o_e,\ell_e),\quad \forall e\in E\\
&x\geq 0,\gamma\in \R.\nonumber
\end{align*}
Analogously to Lemma \ref{lemupp}, we can show that maximum value of {\sf DLP} is at most $\gamma_{\sf U}({\sf EM},\mathcal{G})$, and by the Weak Duality Theorem, this is an upper bound on the optimal value of {\sf LP}, thus showing that $\gamma_{\sf U}({\sf EM},\mathcal{G})\geq {\sf EM}(\mathcal{G})$. 
\subsection{Proof of Lemma \ref{lemrou2} (sketch).}
As the claim of Lemma \ref{lemrou1} is obtained by reversing the proof arguments of Lemma \ref{lemupp}, the claim of Lemma~\ref{lemrou2} is analogously derived from Lemma~\ref{lemupp2}. In particular, we can show the claim of Lemma \ref{lemrou2} by resorting to the following steps: (i) we start from the upper bound $\gamma_{\sf U}({\sf EM},\mathcal{G})$ obtained in Lemma~\ref{lemupp2}; (ii) we derive a linear program $\overline{\sf DLP}$ similar as the dual program used in Lemma~\ref{lemupp2}, but with two constraints only (except for those imposing the non-negativity of the variables), and whose optimal value is higher than $M$; (iii) the dual of $\overline{\sf DLP}$ is similar to the linear program {\sf LP} used in Lemma~\ref{lemupp2}, but with two constraints and two variables only, and has the same optimal value (higher than $M$) as in $\overline{\sf DLP}$ (by the Strong Duality Theorem); (iv) finally, the claim of Lemma~\ref{lemrou2} is obtained by characterizing the optimal solution of $\overline{\sf LP}$. 
\section{Missing Proofs from Section \ref{sec_ident}}
\subsection{Proof of Theorem \ref{thm5}.}
Let $\LB(W,m)\in {\sf WSLB}(\{f\})$ be an arbitrary weighted load balancing game with identical resources such that $W$ is the total weight of players (i.e., $W=\sum_{i\in \sf N}w_i$), and the set of resources $E$ is equal to $[m]$ for some integer $m\geq 2$. Let $\bm\sigma$ be an $\epsilon$-approximate pure Nash equilibrium and $\bm\sigma^*$ be an optimal strategy profile. 
Let $C_1(W,m):=\left\{\bm x\in \RP^m:\sum_{e\in [m]} x_e=W\right\}$. 
\begin{lemma}\label{lemma1}
${\sf SUM}({\bm\sigma}^*)\geq W f\left(\frac{W}{m}\right)$.
\end{lemma}
\begin{proof}
Let $g_m(\bm x):=\sum_{e\in [m]} x_ef(x_e)$ for any $\bm x\in \RP^m$, and let $$opt(W,m):=\min_{\bm x\in C_1(W,m)}\sum_{e\in [m]}g_m(\bm x).$$ Observe that $opt(W,m)\leq {\sf SUM}({\bm\sigma}^*)$. Since $g_m(\bm x)$ is a convex function over $C_1(W,m)$, that is a compact and convex subset of $\R^{m}$, we have that $g_m(\bm x)$ admits at least one minimum point $\bm x:=(x_1,x_2,\ldots, x_m)$. Let $\bm x^j=(x^j_1,x^j_2,\ldots, x^j_m)$ be such that $x_e^j:=x_{e+j}$ for any $e\in [m]$, where the indices are cyclical with respect to $m$ (i.e., $j+m:=j$). Due to the symmetry of function $g_m(\bm x)$, all the $\bm x^j$'s are minima of $g_m$. Due to the convexity of the function $g_m$ and the set $C_1(W,m)$, we have that the set $C^*(W,m)$ of the minima of $g_m$ is a convex subset of $C_1(W,m)$, therefore $\bm x^*:=\sum_{j=0}^{m-1}\frac{\bm x^j}{m}\in C^*(W,m)$ (since $\bm x^*$ is convex combination of all the $\bm x^j$'s). Since $x^*_e=\frac{\sum_{e=1}^mx_e}{m}=\frac{W}{m}$ we have that $
W f\left(\frac{W}{m}\right)=g(\bm x^*)=opt(W,m)\leq {\sf SUM}({\bm\sigma}^*),$ thus concluding the proof.\qed
\end{proof}
By the previous lemma we have that 
\begin{equation}\label{prop4}
 \poa_{\epsilon}(\LB(W,m))\leq \frac{\max_{{\bm\sigma}\in {\sf NE}_{\epsilon}(\LB(W,m))}{\sf SUM}({\bm\sigma})}{W f\left(W/m\right)}.
\end{equation}

\begin{lemma}\label{lemma2}
Let $\bm\sigma$ be an $\epsilon$-approximate pure Nash equilibrium, and let $k_e$ and $k_{e'}$ be respectively the maximum and the minimum congestion in $\bm\sigma$. If $[k_e]_{\epsilon,f}> k_{e'}$, then (i) there is exactly one player $j$ selecting resource $e$ in $\sg$, and (ii) there exists an optimal strategy profile $\bm\sigma^*$ in which player $j$ is the unique player playing resource $e$.
\end{lemma}
\begin{proof}
We first show claim (i). Suppose, by way of contradiction, that there are at least two players using $e$ in $\bm\sigma$. Therefore, there exists a player $i$ among them with weight at most $k_e/2$. By using the fact that $\bm\sigma$ is an equilibrium and by using the definition of $[k_e]_{\epsilon,f}$, we have that $$cost_i({\bm\sigma})\leq (1+\epsilon)cost_i({\bm\sigma}_{-i},\{e'\})\leq (1+\epsilon)f(k_{e'}+k_e/2)<f(k_e)=cost_i({\bm\sigma}),$$
where the last inequality comes from $[k_e]_{\epsilon,f}> k_{e'}$. This fact leads to $cost_i({\bm\sigma})<cost_i({\bm\sigma})$, that is a contradiction. Thus, there is  exactly one player $j$ selecting resource $e$ in $\bm\sigma$.

To show claim (ii), consider an optimal strategy profile $\bm\sigma^*$ in which player $j$ selects resource $e$ (such an optimal strategy profile always exists) and assume that there is another player $i$ selecting $e$. Suppose, without loss of generality, that $e'$ is the resource having minimum congestion in the strategy profile $\bm\sigma^*$. We get
\begin{align}
&{\sf SUM}(\bm\sigma^*)-{\sf SUM}(\bm\sigma^*_{-i},\{e'\})\nonumber\\
=&k_e(\bm\sigma^*)f(k_e(\bm\sigma^*))+k_{e'}(\bm\sigma^*)f(k_{e'}(\bm\sigma^*))-(k_e(\bm\sigma^*)-w_i)f(k_e(\bm\sigma^*)-w_i)\nonumber\\
&-(k_{e'}(\bm\sigma^*)+w_i)f(k_{e'}(\bm\sigma^*)+w_i)\nonumber\\
=&\underbrace{k_e(\bm\sigma^*)f(k_e(\bm\sigma^*))-(k_e(\bm\sigma^*)-w_i)f(k_e(\bm\sigma^*)-w_i)}_{A\left(k_e(\bm\sigma^*)\right)}\nonumber\\
&-\underbrace{((k_{e'}(\bm\sigma^*)+w_i)f(k_{e'}(\bm\sigma^*)+w_i)-k_{e'}(\bm\sigma^*)f(k_{e'}(\bm\sigma^*)))}_{A(k_{e'}(\bm\sigma^*)+w_i)}.\label{inesemiconv}
\end{align}
Since $w_j$ is the maximum equilibrium congestion (because of the first claim of the lemma) and $k_{e'}(\bm\sigma^*)$ is the minimum optimal congestion, we necessarily have that $k_e(\bm\sigma^*)-w_i\geq w_j=k_e\geq W/m\geq k_{e'}(\bm\sigma^*)\Rightarrow k_e(\bm\sigma^*)\geq k_{e'}(\bm\sigma^*)+w_i$ (observe that  $k_e(\bm\sigma^*)-w_i\geq w_j$ holds since players $i$ and $j$ select resource $e$ in $\bm\sigma^*$).  Therefore, because of the semi-convexity of $f$, we have $A(k_e(\bm\sigma^*))-A(k_{e'}(\bm\sigma^*)+w_i)\geq 0$ (in particular, such inequality holds since $F(y+z)-F(x+z)-(F(y)-F(x))\geq 0$ for any convex function $F$, and $x,y,z\in \RP$ such that $x\leq y$). Then, because of (\ref{inesemiconv}), we get ${\sf SUM}(\bm\sigma^*_{-i},\{e'\})\leq {\sf SUM}(\bm\sigma^*)$, which proves the optimality of the strategy profile $(\bm\sigma^*_{-i},\{e'\})$. By applying iteratively the same proof arguments to the strategy profile $(\bm\sigma^*_{-i},\{e'\})$, at the end of the process we obtain an optimal strategy profile in which player $j$ is the unique player selecting resource $e$ in $\bm\sigma^*$, thus showing the claim. \qed
\end{proof}
Let $k_e$ be the maximum equilibrium congestion of $\LB(W,m)$. By the previous lemma, we have that, if there exists a resource $e'$ such that $[k_e]_{\epsilon,f}>k_{e'}$, there is a resource used by a unique player $j$, in both the equilibrium and the optimal strategy profile, denoted respectively with $\bm\sigma$ and $\bm\sigma^*$. Therefore, if we remove player $j$ and resource $e$ from $\LB(W,m)$, $\bm\sigma_{-j}$ and $\bm\sigma^*_{-j}$ are respectively an equilibrium and an optimal strategy profile in the new game $\LB'(W,m)$, verifying ${\sf SUM}(\bm\sigma)-k_e f(k_e)={\sf SUM}(\bm\sigma_{-j})$ and ${\sf SUM}(\bm\sigma^*)-k_e f(k_e)={\sf SUM}(\bm\sigma^*_{-j})$. Therefore $\poa_{\epsilon}(\LB(W,m)')\geq \poa_{\epsilon}(\LB(W,m))$.

Thus, for our aims, we can assume (without loss of generality) that all the approximate pure Nash equilibria of $\LB(W,m)$ verify $[k_e(\bm\sigma)]_{\epsilon,f}\leq k_{e'}(\bm\sigma)$ for each $e,e'\in E$. We conclude that 
\begin{align}
\max_{{\bm\sigma}\in {\sf NE}_{\epsilon}(\LB(W,m))}{\sf SUM}({\bm\sigma})\leq \max_{\bm x\in C_2(k_e(\sg),W,m)}\sum_{e\in [m]}x_ef(x_e),\label{prop5}
\end{align}
where $C_2(x,W,m):=\left\{\bm x\in \RP^m:[x]_{\epsilon,f}\leq x_e\leq x\ \forall e\in [m],\ \sum_{e\in [m]}x_e=y\right\}$ for any $x>0$. 
\begin{lemma}\label{lemprin}
Given $x>0$ and $\bm x\in C_2(x,W,m)$, there exists $\lambda\in [0,1]$ such that (i) $\sum_{e\in [m]}x_ef(x_e)\leq (\lambda xf(x)+(1-\lambda)[x]_{\epsilon,f}f([x]_{\epsilon,f}))m$, and (ii) $W=(\lambda x+(1-\lambda)[x]_{\epsilon,f})m$. 
\end{lemma}
\begin{proof}
Let $x>0$ and $\bm x\in C_2(x,y,l)$. For any $e\in [m]$, let $\lambda_e\in [0,1]$ such that $x_e=\lambda_e x+(1-\lambda_e)[x]_{\epsilon,f}$, and let $\lambda:=(\sum_{e\in [m]}\lambda_e)/m$. By the semi-convexity of $f$ we have that 
\begin{equation}\label{lamb1}
x_ef(x_e)=(\lambda_e x+(1-\lambda)[x]_{\epsilon,f})f(\lambda_e x+(1-\lambda)[x]_{\epsilon,f})\leq \lambda_e xf(x)+(1-\lambda_e)[x]_{\epsilon,f}f([x]_{\epsilon,f}).
\end{equation}
By \eqref{lamb1}, we have that $$\sum_{e\in [m]}x_ef(x_e)\leq \sum_{e\in [m]}( \lambda_e xf(x)+(1-\lambda_e)[x]_{\epsilon,f}f([x]_{\epsilon,f}))=\lambda m  xf(x)+(1-\lambda)m[x]_{\epsilon,f}f([x]_{\epsilon,f}),$$
and this shows claim (i). To show claim (ii), we simply observe that $$W=\sum_{e\in [m]}x_e=\sum_{e\in [m]}(\lambda_e x+(1-\lambda_e)[x]_{\epsilon,f})=\lambda m x+(1-\lambda_e)m[x]_{\epsilon,f}.$$\qed
\end{proof}
By applying Lemma \ref{lemprin}, and by using inequalities \eqref{prop4} and \eqref{prop5}, we get
\begin{align}
\poa_{\epsilon}(\LB(W,m))
&\leq \frac{\max_{{\bm\sigma}\in {\sf NE}_{\epsilon}(\LB(W,m))}{\sf SUM}({\bm\sigma})}{W f\left(W/m\right)}\nonumber\\
&\leq \sup_{x>0}\frac{\max_{\bm x\in C_2(k_e(\sg),W,m)}\sum_{e\in [m]}x_ef(x_e)}{W f\left(W/m\right)}\nonumber\\
&\leq \sup_{x>0}\sup_{\lambda\in (0,1)}\frac{(\lambda xf(x)+(1-\lambda)[x]_{\epsilon,f}f([x]_{\epsilon,f}))m}{(\lambda x+(1-\lambda)[x]_{\epsilon,f})m f(\lambda x+(1-\lambda)[x]_{\epsilon,f})}\label{idfin1}\\
&=\sup_{x>0}\sup_{\lambda\in [0,1]}\gamma_{\epsilon,f}(x,\lambda)\nonumber\\
&=\sup_{x>0}\max_{\lambda\in (0,1)}\gamma_{\epsilon,f}(x,\lambda),\label{idfin2}
\end{align}
where \eqref{idfin1} comes from Lemma \ref{lemprin}, and \eqref{idfin2} holds since $\gamma_{\epsilon,f}(x,\lambda)$, for any fixed $x>0$, is maximized by some $\lambda\in (0,1)$\footnote{For any fixed $x>0$, we have that $\gamma_{\epsilon,f}(x,\lambda)\geq 1$ for any $\lambda\in (0,1)$ (as $\lambda xf(x)+(1-\lambda)[x]_{\epsilon,f}f([x]_{\epsilon,f})\leq \lambda x+(1-\lambda)[x]_{\epsilon,f}$, by semi-convexity of $f$), and $\gamma_{\epsilon,f}(x,\lambda)$ approaches to $1$ as $\gamma$ tends to either $0$ or $1$. Thus, as $\gamma_{\epsilon,f}(x,\lambda)$ is continuous, it admits a maximum point $\lambda\in (0,1)$.}, and this shows the claim. 
\subsection{Proof of Theorem \ref{thm6}.}
It suffices showing that, for each $M<\sup_{x>0}\sup_{\lambda\in (0,1)}\gamma_{\epsilon,f}(x,\lambda)$, there exists $\LB\in {\sf WSLB}(\{f\})$ such that $\poa_{\epsilon}(\LB)>M.$ For any $x>0$ and $\lambda\in (0,1)$, let $opt(x,\gamma)$ denote $\gamma x+(1-\lambda)[x]_{\epsilon,f}$. 

Fix $x>0$. Let $m$ be a sufficiently large even number such that $h(m):=\lceil m \lambda^*(x)\rceil\leq m/2$ (this number exists as $\gamma^*(x)\leq 1/2$). We observe that $opt(x,h(m)/m)-x/2\geq 0$. Indeed, if $\epsilon=0$, we have that $[x]_{\epsilon,f}=x/2$, thus $opt(x,h(m)/m)-x/2\geq (h(m)/m)x/2+(1-h(m)/m)x/2-x/2=0$; instead, if $\epsilon>0$, we have that $opt(x,h(m)/m)-x/2\geq opt(x,\lambda^*(x))-x/2\geq 0$, where the last inequality follows from the hypothesis on $\lambda^*(x)$.

Let $\LB_m(x)$ be a game with $m$ resources, $2h(m)$ ``red" players of weight $x/2$, and several ``blue" players, of total weight $(m-h(m))[x]_{\epsilon,f}$, that can be partitioned in the following two different ways:
\begin{itemize}
\item[$\bullet$] in $m-h(m)$ groups such that each group has total weight equal to $[x]_{\epsilon,f}$;
\item[$\bullet$] in $m$ groups such that any of the first $2h(m)$ ones is made of ``dark blue" players and has total weight equal to $opt(x,h(m)/m)-x/2\geq opt(x,\lambda^*(x))-x/2\geq 0$, and any of the remaining $m-2h(m)$ groups is made of ``light blue" players and has total weight equal to $opt(x,h(m)/m)$; we observe that the above quantities are well-defined as $2h(m)\leq m$. 
\end{itemize}
We can determine a set of blue players satisfying the previous conditions as follows. Consider a segment of $(m-h(m))[x]_{\epsilon,f}$ units of weight and the following two subdivisions of this segment:
\begin{itemize}
\item[$\bullet$] {\bf Subdivision 1:} $m-h(m)$ intervals, each having $[x]_{\epsilon,f}$ units of weight;
\item[$\bullet$] {\bf Subdivision 2:} $2h(m)$ intervals, each having $opt(x,h(m)/m)-x/2$ units of weight, followed by $m-2h(m)$ intervals, each having $opt(x,h(m)/m)$ units of weight. 
\end{itemize}
The desired set of blue players can be obtained by overlapping Subdivisions 1 and 2 and then associating a blue player of weight $w_i$ to each interval of $w_i$ units of weight of the resulting subdivision, call it Subdivision 3. Observe that, being Subdivision 3 a subdivision of Subdivision 2, we can properly associate dark blue players to the first $2h(m)$ intervals of Subdivision 3 and light blue players to the remaining ones. Observe also that Subdivision 3 is well defined (since $(m-h(m))[x]_{\epsilon,f}=2h(m)(opt(x,h(m)/m)-x/2)+(m-2h(m))opt(x,h(m)/m)$, i.e., the total weights of Subdivisions 1 and 2 are equal). 

Consider the strategy profile $\bm\sigma_m$ defined as follows:
\begin{itemize}
\item[$\bullet$] There are $h(m)$ resources such that each resource is selected by two red players; thus the congestion of such resources is $x$. 
\item[$\bullet$] There are $m-h(m)$ resources with congestion $[x]_{\epsilon,f}$ played by blue players only; this configuration can be realized since the weights of the blue players are determined from an interval subdivision that is more refined than Subdivision 1.
\end{itemize}
By exploiting the definition of $[x]_{\epsilon,f}$, we have that $\bm\sigma_m$ is an $\epsilon$-approximate pure Nash equilibrium.

Now, consider the strategy profile $\bm\sigma^*_m$ defined as follows:
\begin{itemize}
\item[$\bullet$] There are $2h(m)$ resources of congestion $opt(x,h(m))$ such that each resource is selected by one red player and some dark blue players only;
\item[$\bullet$] There are $m-2h(m)$ resources of congestion $opt(x,h(m))$ such that each resource is selected by light blue players only. 
\end{itemize}
By exploiting Subdivision 2, we observe that $\bm\sigma^*_m$ is well-defined. See Figure \ref{fig:1} for a clarifying example. 

By using the definitions of $\sg_m$ and $\sg^*_m$ we have that 
\begin{equation*}
\poa_\epsilon(\LB_m(x))\geq \frac{{\sf SUM}(\bm\sigma_m)}{{\sf SUM}(\bm\sigma^*_m)}=\frac{h(m)xf(x)+(m-h(m))[x]_{\epsilon,f}f([x]_{\epsilon,f})}{m\cdot opt(x,h(m)/m)}=\gamma_{\epsilon,f}(x,h(m)/m).
\end{equation*}  
As $\lim_{m\rightarrow \infty}h(m)/m=\lambda^*(x)$, and since $\gamma_{\epsilon,f}(x,t)$ is continuous with respect to $t\in (0,1)$, we have that 
\begin{equation}\label{low}
\limsup_{m\rightarrow \infty}\poa_\epsilon(\LB_m(x))\geq \lim_{m\rightarrow \infty}\gamma_{\epsilon,f}(x,h(m)/m)=\lim_{t\rightarrow \lambda^*(x)}\gamma_{\epsilon,f}(x,t)=\gamma_{\epsilon,f}(x,\lambda^*(x)).
\end{equation}
By taking the supremum of the right-hand part of (\ref{low}) over $x\geq 0$, we get the claim.

\begin{figure}[h]
\centering
\includegraphics[scale=0.5]{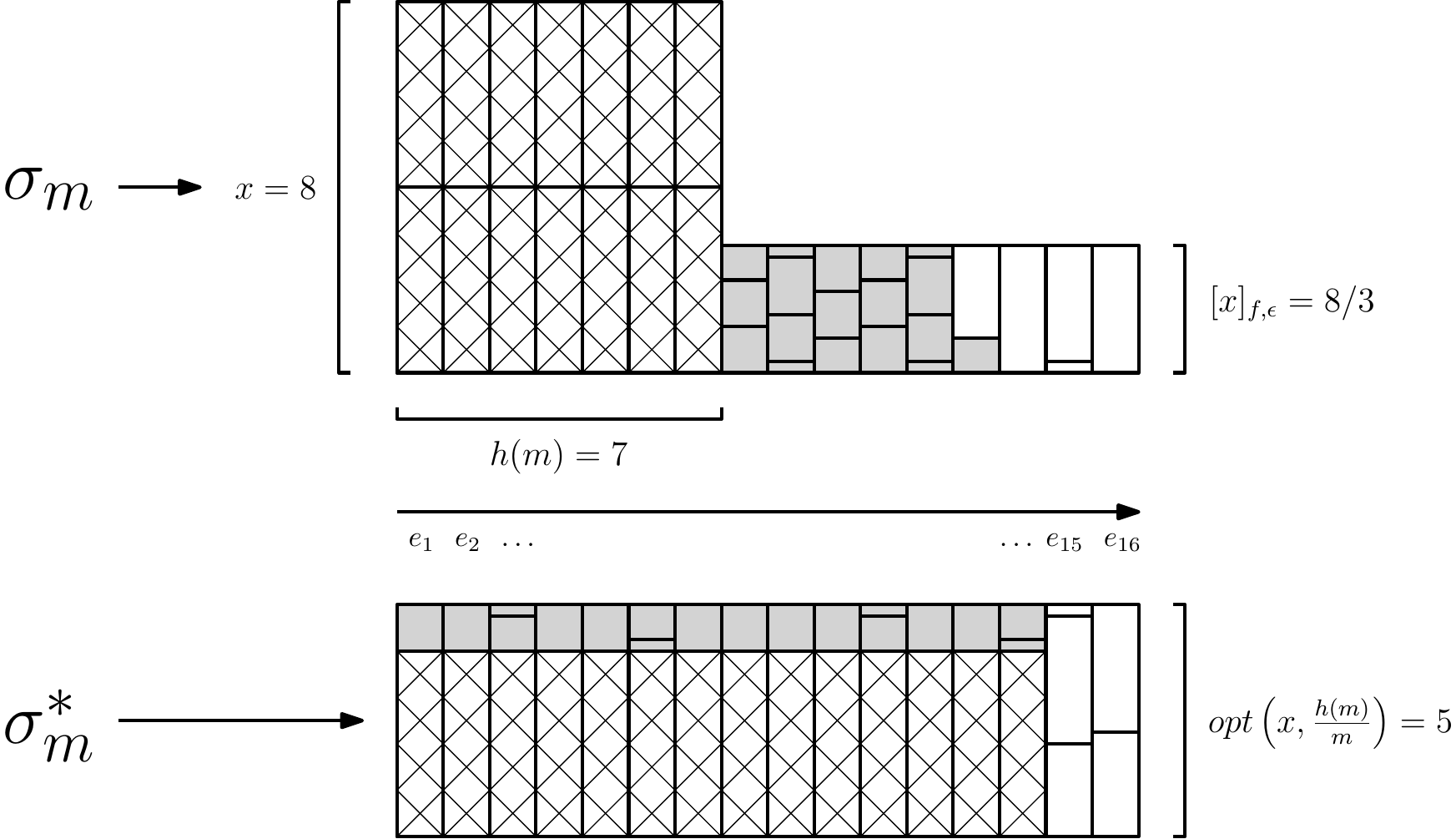}
\caption{An equilibrium configuration $\bm\sigma_m$ and an optimal configuration $\bm\sigma^*$ of a load balancing game $\LB_m(x)$ with $m=16$, $h(m)=7$, $x=8$ and $[x]_{\epsilon,f}=8/3$ (then, $opt(x,h(m)/m)=5$). The abscissa represents the resources, the ordinate represents the congestion of each resource in $\bm\sigma_m$ and $\bm\sigma^*_m$, and each vertical block is a player having weight  proportional to its height. Red players are filled with crosses, dark blue players are coloured, and light blue players are white.}\label{fig:1}
\end{figure}
\subsection{Proof of Corollary \ref{thmpoly}.}
Let $f$ be such that $f(x):=\sum_{i=0}^p\alpha_i\cdot x^i$. We have
\begin{align}
\gamma_{\epsilon,f}(x,\lambda)&=\frac{\lambda \cdot x\cdot f(x)+(1-\lambda)\cdot [x]_{\epsilon,f}\cdot f([x]_{\epsilon,f}) }{(\lambda x +(1-\lambda)[x]_{\epsilon,f})f(\lambda x +(1-\lambda)[x]_{\epsilon,f})}\nonumber\\
&=\frac{\lambda \cdot x\cdot \sum_{i=0}^d\alpha_ix^i+(1-\lambda)\cdot [x]_{\epsilon,f}\sum_{i=0}^d\alpha_i [x]_{\epsilon,f}^i}{(\lambda x +(1-\lambda)[x]_{\epsilon,f})\sum_{i=0}^d\alpha_i(\lambda x +(1-\lambda)[x]_{\epsilon,f})^i}\nonumber\\
&=\frac{\sum_{i=0}^d\alpha_i\left(\lambda x^{i+1}+(1-\lambda) [x]_{\epsilon,f}^{i+1}\right)}{\sum_{i=0}^p\alpha_i(\lambda x +(1-\lambda)[x]_{\epsilon,f})^{i+1}}\nonumber\\
&=\frac{\sum_{i=1}^d\alpha_i\left(\lambda x^{i+1}+(1-\lambda) (x/2)^{i+1}\right)+\lambda x}{\sum_{i=1}^p\alpha_i(\lambda x +(1-\lambda)(x/2))^{i+1}+\lambda x}\nonumber\\
&\leq \frac{\sum_{i=1}^d\alpha_i\left(\lambda x^{i+1}+(1-\lambda) (x/2)^{i+1}\right)}{\sum_{i=1}^p\alpha_i(\lambda x +(1-\lambda)(x/2))^{i+1}}\nonumber\\
&\leq \max_{i\in [d]} \frac{\lambda x^{i+1}+(1-\lambda) (x/2)^{i+1}}{(\lambda x +(1-\lambda)(x/2))^{i+1}}\nonumber\\
&=\max_{i\in [d]} g_i(\lambda),\label{ineqlambda}
\end{align}
where $g_i(\lambda):=\frac{\lambda +(1-\lambda) (1/2)^{i+1}}{(1/2+\lambda/2)^{i+1}}$. We equivalently have that $g_i(\lambda)=\gamma_{\epsilon,f}(x,\lambda)$, where $f_i$ is the monomial function defined $f_i(x):=x^i$. Let $\lambda^*_i\in (0,1)$ be the value maximizing $g_i(\lambda)$; we have that $\lambda^*_i$ necessarily verifies $\frac{\partial}{\partial \lambda}g_i(\lambda)|_{\lambda=\lambda^*}=0$. Thus, by simple calculations, we get
$\lambda^*_i=\frac{2^{i+1}-i-2}{i2^{i+1}-i}$ and
\begin{equation}
g_i(\lambda^*_i)=\frac{i^i(2^{i+1}-1)^{i+1}}{2^i(i+1)^{i+1}(2^i-1)^i}.\label{upppoly}
\end{equation}
We prove that $g_i(\lambda^*_i)$ is non-decreasing with respect to $i\in\mathbb{N}$, i.e., $g_{i+1}(\lambda^*_{i+1})/g_{i}(\lambda^*_{i})\geq 1$ for any $i\in\mathbb{N}$. If $i=1$, the claim holds.
If $i\geq 2$, we get 
\begin{align*}
\frac{g_{i+1}(\lambda^*_{i+1})}{g_i(\lambda^*_i)}
&=\frac{(i+1)^{i+1}(2^{i+2}-1)^{i+2}2^i(i+1)^{i+1}(2^{i}-1)^i}{2^{i+1}(i+2)^{i+2}(2^{i+1}-1)^{i+1}i^i(2^{i+1}-1)^{i+1}}\\
&=\frac{(i+1)^{i+1}(i+1)^{i+1}}{(i+2)^{i+2}i^i}\cdot \frac{(2^{i+2}-1)^{i+2}(2^{i+1}-2)^i}{(2^{i+2}-2)^{i+1}(2^{i+1}-1)^{i+1}}\\
&= \frac{(i+1)^{i+1}(i+1)^{i+1}}{(i+2)^{i+2}i^i}\cdot \frac{(2^{i+2}-1)(2^{i+2}-1)^{i+1}(2^{i+1}-2)^i}{(2^{i+2}-2)^{i+1}(2^{i+1}-1)^{i+1}}\\
&\geq  \frac{(i+1)^{i+1}(i+1)^{i+1}}{(i+2)^{i+2}i^i}\cdot \frac{2(2^{i+1}-2)(2^{i+2}-1)^{i+1}(2^{i+1}-2)^i}{(2^{i+2}-2)^{i+1}(2^{i+1}-1)^{i+1}}\\
&= \frac{(i+1)^{i+1}(i+1)^{i+1}}{(i+2)^{i+2}i^i}\cdot 2\\
&=\left(\frac{(i+1)^2}{(i+2)i}\right)^i \cdot 2\left(\frac{i+1}{i+2}\right)^2\\
&\geq  1,
\end{align*}
where the last inequality holds since $\frac{(i+1)^2}{(i+2)i}\geq 1$ and $\left(\frac{i+1}{i+2}\right)^2\geq \frac{1}{2}$ for each $i\geq 2$.

Therefore, by continuing from (\ref{ineqlambda}), we get $\gamma_{\epsilon,f}(x,\lambda)\leq \max_{i\in [d]}g_i(\lambda^*_i)= g_d(\lambda^*_d)$; thus, we can assume without loss of generality that $f(x)=x^d$ for any $x\geq 0$. 

By Theorem~\ref{thm5} and (\ref{upppoly}) we get
\begin{align*}
\poa_0(\mathcal{P}(d))
&\leq \sup_{x>0}\max_{\lambda\in (0,1)}\gamma_{\epsilon,f}(x,\lambda)\\
&=\frac{d^d(2^{d+1}-1)^{d+1}}{(d+1)2^d(d+1)^d(2^d-1)^d}\\
&=\left(\frac{d}{d+1}\right)^d\left(\frac{(2^{d+1}-2+1)^{d+1}}{(d+1)(2^{d+1}-2)^d}\right)\\
&\sim e^{-1}\left(\frac{2^{d+1}-2+1}{d+1}\right)\left(\frac{2^{d+1}-2+1}{2^{d+1}-2}\right)^d\\
&\sim e^{-1}\left(\frac{2^{d+1}-2+1}{d+1}\right)e^{\frac{d}{2^{d+1}-2}}\\
&\sim \frac{2^d}{d},
\end{align*}
thus showing the desired upper bound on the price of anarchy. Furthermore, we have that $\lambda^*_d\leq 1/2$ for any $d\in\mathbb{N}$. Indeed, if $d=1$ then $\lambda^*_1=1/3$, otherwise, for $d\geq 2$, we get 
$$\lambda^*_d=\frac{2^{d+1}-d-2}{d2^{d+1}-d}\leq \frac{2^{d+1}-d-2+d}{d2^{d+1}-d+d}\leq \frac{2^{d+1}}{d2^{d+1}}\leq 1/2.$$
Thus, since $\lambda^*_d\leq 1/2$, we can apply Theorem \ref{thm6} to show that the above upper bound is tight, and this concludes the proof. 
\end{document}